\documentclass[12pt]{article}
\usepackage[a4paper, margin=1in]{geometry}
\usepackage{amsmath}
\usepackage{graphicx}
\usepackage{algorithm}
\usepackage{algpseudocode}
\usepackage{enumerate}
\usepackage{natbib}
\usepackage{amssymb}
\usepackage{amsthm}
\usepackage{graphicx}
\usepackage{float}
\usepackage{setspace}
\usepackage{rotating}
\usepackage{amsthm}
\usepackage{bm}
\usepackage{adjustbox}
\usepackage{bbm}
\newcommand{\p}{\partial }
\newtheorem{proposition}{Proposition}[section]

\newcommand{\W}{\mathbf{W}}
\newcommand{\I}{\mathbf{I}}
\newcommand{\TM}{\mathrm{M}}
\newcommand{\TO}{\mathrm{O}}
\newcommand{\w}{\mathbf{w}}

\newcommand{\etab}{\bm\eta}
\newcommand{\lwr}{\mathrm{lwr}}
\newcommand{\upr}{\mathrm{upr}}
\newcommand{\etabb}{\bar{\bm\eta}}
\newcommand{\rest}{~\mathrm{rest}}
\renewcommand{\u}{\mathbf{u}}

\newcommand{\R}{\mathbf{R}}
\newcommand{\X}{\mathbf{X}}
\newcommand{\x}{\mathbf{x}}

\newcommand{\Z}{\mathbf{Z}}
\newcommand{\Y}{\mathbf{Y}}
\newcommand{\RMSPE}{\mathrm{RMSPE}}
\newcommand{\EC}{\mathrm{EC}}

\newcommand{\xibb}{\bar{\bm{\xi}}}
\newcommand{\s}{\mathbf{s}}
\newcommand{\Gam}{\mathrm{Gam}}
\renewcommand{\d}{\mathrm{d}}
\newcommand{\Gau}{\mathrm{Gau}}
\newcommand{\GCD}{\mathfrak{g}}
\newcommand{\TCD}{\mathfrak{t}}

\newcommand{\V}{\mathbf{V}}

\renewcommand{\S}{\mathbf{S}}
\newcommand{\MVG}{\mathrm{MVG}}

\newcommand{\E}{\mathbf{E}}
\newcommand{\xib}{\bm\xi}
\newcommand{\SRE}{\mathrm{SRE}}
\newcommand{\TP}{\mathrm{P}}
\renewcommand{\v}{\mathbf{v}}

\newcommand{\trop}{\mathrm{TROP}}
\newcommand{\LG}{\mathrm{LG}}

\begin{document}

\def\spacingset#1{\renewcommand{\baselinestretch}%
{#1}\small\normalsize} \spacingset{1}

\date{4 November, 2025}
  \title{\bf Bayesian copula-based spatial random effects models for inference with complex spatial data}
  \author{Alan R. Pearse$^{1}$\thanks{
    Corresponding author: alan.pearse@unimelb.edu.au}, David Gunawan$^{2}$ and Noel Cressie$^{2}$\\
    $^{1}$School of Mathematics and Statistics, University of Melbourne,\\
    Parkville, Victoria, Australia 3010\\
    $^{2}$NIASRA, School of Mathematics and Applied Statistics,\\University of Wollongong, Wollongong, New South Wales, Australia, 2522}

\maketitle

\begin{abstract}
In this article, we develop fully Bayesian, copula-based, spatial-statistical models for large, noisy, incomplete, and non-Gaussian spatial data. Our approach includes novel constructions of copulas that accommodate a spatial-random-effects structure, enabling low-rank representations and computationally efficient Bayesian inference. The spatial copula is used in a latent process model of the Bayesian hierarchical spatial-statistical model, and, conditional on the latent copula-based spatial process, the data model handles measurement errors and missing data. Our simulation studies show that a fully Bayesian approach delivers accurate and fast inference for both parameter estimation and spatial-process prediction, outperforming several benchmark methods, including fixed rank kriging (FRK). The new class of copula-based models is used to map atmospheric methane in the Bowen Basin, Queensland, Australia, from Sentinel 5P satellite data. 
\end{abstract}


\sloppy

\spacingset{1}
\section{Introduction}
\label{sec:intro}

Non-Gaussian spatial processes where the marginal distributions are non-Gaussian (e.g., skewed, strictly positive, or multimodal), are commonly used in  environmental-science models. Even after transformation of the marginals, the spatial-dependence structure of the process may be unlike that of a Gaussian process. Trans-Gaussian models \citep[e.g., ][]{Cressie1993} and Spatial Generalized Linear Mixed Models \citep[SGLMMs;][]{Diggle1998} use marginal transformations of a Gaussian spatial process; however, the spatial dependence in these models is fundamentally Gaussian-like, with spatial correlation that tends to zero as the values become more extreme. 

Spatial copulas enable modeling with non-Gaussian marginals and non-Gaussian spatial dependence that can, for example, capture substantial correlation among extreme values of the process. \citet{Bardossy2006} first used copulas to model a non-Gaussian random field in hydrology, developing the chi-squared copula for spatial processes exhibiting asymmetric dependence between small and large values. Subsequently, many copulas that generate non-Gaussian spatial processes have been developed for environmental applications, including the Fisher copula \citep{Favre2018}, the conditional-normal extreme-value copulas \citep{Krupskii2021}, and a family of ``Clayton-like'' spatial copulas \citep{Bevilacqua2024}. The Gaussian copula \citep[e.g.,][]{Song2000} and the t copula \citep[e.g.,][]{Dermata2005} can also be used for spatial-statistical inference \citep[e.g.,][]{Kazianka2011, Erhardt2015}. In fact, the aforementioned trans-Gaussian models are special cases of Gaussian-copula models \citep{Kazianka2010}, and we expand on this result for `transformed t' spatial processes in Section \ref{sec:t_copula_SRE}. 

In this article, we not only develop new dimension-reduced spatial copulas, but we also show how they should be used in a Bayesian hierarchical statistical model. That is, we show how to capture the large-, small-, and micro-scale variability of the spatial process and separate them from measurement errors in noisy and incomplete spatial data. Not distinguishing between the data and the process has led in the past to the unrealistic assumption of zero measurement error or to assuming an overly smooth process with no micro-scale variation \citep[pp. 120-123]{CressieWikle2011}. Previous attempts at hierarchical spatial-copula models include \citet{Ghosh2011}, \citet{Beck2019}, and \citet{garcia2021bayesian}. In these models, a copula is inserted at the data-model level of the hierarchy, so measurement errors are still confounded with the variability of the scientific process. On the other hand, \citet[Sec. IV.2]{Ghosh2009},  \citet{Prates2015} and \citet{Musgrove2016} defined copula-based SGLMMs for count and binary data, featuring a data model that deals with the discreteness of the observed data. However, the general notion of measurement errors and missingness is absent from these works. 

The past literature on spatial-copula modeling has not adequately considered inference that scales to very large datasets. In today's data-rich world of large remote-sensing and sensor-network datasets, spatial-copula models that can efficiently handle tens of thousands of observations or more are needed.  Fixed-rank kriging (FRK) and, more broadly, spatial random effects (SRE) models \citep{Cressie2008, ZM2021, SainsburyDale2024} use basis functions for dimension reduction, which we generalize here to the spatial-copula model. 

In this article, our contributions are as follows. We develop novel Gaussian copulas and t copulas with embedded SREs that enable fast Bayesian computation for large spatial datasets, based on Markov chain Monte Carlo (MCMC). The Gaussian-copula SRE models generalize existing trans-Gaussian SRE models and SGLMMs \citep{Sengupta2013, SainsburyDale2024}. The t-copula models with embedded SREs result in non-Gaussian spatial processes generated from a multivariate t distribution. Importantly, these models are developed in a Bayesian hierarchical spatial-statistical-modeling framework \citep[e.g.,][]{CressieWikle2011} that enables inference on a latent scientific process from large, noisy, incomplete spatial data. Simulation studies and an application to a satellite-remote-sensing dataset of atmospheric methane concentrations demonstrate the inferential and computational advantages of our hierarchical spatial-copula modeling of noisy, incomplete, non-Gaussian spatial data. 

The rest of this article is organized as follows. Section \ref{sec:spatial_copula_model} defines the general copula-based Bayesian hierarchical spatial-statistical model, which includes our copula-based SRE models. Section \ref{sec:Bayesian_inference} discusses computationally efficient Bayesian inference, and Section \ref{sec:simulation_studies} presents simulation studies that demonstrate the validity of our Bayesian inferences. In Section \ref{sec:methane}, we present an application to predicting atmospheric methane concentrations from satellite data taken in and around the Bowen Basin, Queensland, Australia, which is a region of intense coal-mining activity. Conclusions and a discussion are given in Section \ref{sec:discussion}. This article has a Supplement containing proofs for all propositions as well as additional technical details and empirical results.

\section{Copula-based spatial random effects model}\label{sec:spatial_copula_model}

In this section, we define a general Bayesian hierarchical spatial-statistical model (BHM; Section \ref{sec:hierarchical_models}), and then we particularize it to incorporate spatial copulas (Section \ref{sec:copulas}).  Sections \ref{sec:Gaussian_copula_SRE}--\ref{sec:t_copula_SRE} present Gaussian copulas and t copulas that accommodate a spatial random effects (SRE) structure for efficient computations. 

\subsection{Bayesian hierarchical spatial-statistical models}\label{sec:hierarchical_models}

Let $Y(\cdot) \equiv \{Y(\s): \s \in D\}$ be a spatial stochastic process defined over points $\s$ in spatial domain $D$, usually a $d$-dimensional subset of Euclidean space. Define a set of polygons $\mathcal{A} \equiv  \{A_j: j = 1, ..., N\}$ with areas $|A_j| > 0$, $j = 1, ..., N$, that tile $D$. The polygons in $\mathcal{A}$ are called Basic Areal Units \citep[BAUs;][]{Nguyen2012}, and they are constructed to represent the smallest areal regions over which spatial-statistical inference is meaningful. The BAUs need not be the same size and shape, although they often are \citep{ZM2021}. For each $j = 1, ..., N$, define aggregation over $A_j$ as $Y(A_j) \equiv |A_j|^{-1}\int_{\u \in A_j} Y(\u)~\mathrm{d}\u$, which is the areal average of $\{Y(\u): \u \in A_j\}$, for $j = 1, ..., N$. The random vector $\Y \equiv (Y(A_1), ..., Y(A_N))^\top$ is called the \textit{BAU-level process} and fully describes the (discretized) spatial process over $D$ \citep{ZM2021}. Explicit process models for the marginal behavior and dependence structure of the spatial process are provided by the BAU-level process. The advantage of this approach is that it simplifies \textit{change-of-support} problems \citep[e.g.,][Sec 5.2]{Cressie1993}, where predictions are made at geographic `blocks' made up of several BAUs \citep[e.g., see the examples in][]{SainsburyDale2024}.

Let $\mathcal{A}\equiv \mathcal{A}_\TO \cup \mathcal{A}_\TM$, where $\mathcal{A}_\TO \equiv \{A_{\TO k}: k = 1, ..., K\}$ and $\mathcal{A}_\TM \equiv \{A_{\TM l}: l = 1, ..., L\}$ are two mutually exclusive sets of BAUs (i.e., $\mathcal{A}_\TO\cap \mathcal{A}_\TM = \emptyset$) with $K \leq N$ and $L = N - K$. We call $\mathcal{A}_\TO$ the set of `observed BAUs' (with subscript `O') and $\mathcal{A}_\TM$ the set of `missing BAUs' (with subscript `M'). The corresponding latent-process vectors are $\Y_\TO \equiv (Y(A_{\TO1}), ..., Y(A_{\TO K}))^\top$ and $\Y_\TM \equiv (Y(A_{\TM1}), ..., Y(A_{\TM L}))^\top$, and we can write $\Y \equiv (\Y_\TO^\top, \Y^\top_\TM)^\top$. At the `observed BAUs', we observe the noisy, incomplete spatial data, $\Z_\TO \equiv (Z(A_{\TO 1}), ..., Z(A_{\TO K}))^\top$. Spatial data are observed at or aggregated to the BAU resolution. The `missing BAUs' are so called, because data are considered missing there.

We now define the Bayesian hierarchical spatial-statistical model, comprising a \textit{data model}, a \textit{process model}, and a \textit{parameter model}. Let $\bm\theta_\TO$ be a vector of parameters for the measurement error associated with the data $\Z_\TO$. The data model specifies the distribution of $\Z_\TO$ given $\Y$, written as $[\Z_\TO \mid \Y, \bm\theta_\TO]$. With very few exceptions \citep[e.g.,][]{Wikle2005}, conditional independence of $\Z_\TO$ given $\Y_\TO$ is assumed, resulting in $[\Z_\TO \mid \Y, \bm\theta_{\TO}] = [\Z_\TO \mid \Y_\TO, \bm\theta_\TO] = \prod_{k=1}^K[Z(A_{\TO k}) \mid Y(A_{\TO k}), \bm\theta_{\TO}]$. 

The \textit{process model} describes the conditional distribution of the latent spatial process $\Y \equiv (\Y_\TO^\top, \Y_\TM^\top)^\top$, usually denoted as $[\Y \mid \bm\theta_\TP]$ but equivalently as $[\Y_\TO, \Y_\TM \mid \bm\theta_\TP]$, where $\bm\theta_\TP$ is a vector of parameters that govern the latent spatial process.  From this point, $\bm\theta_\TO$ is assumed known (e.g., from independent calibration experiments of the scientific instrument), so only $\bm\theta_\TP$ is given a distribution. That is, the \textit{parameter model} is $[\bm\theta_\TP]$, which is often referred to as the prior distributions of the parameters.

The goal of Bayesian inference is to infer the `unknowns', namely the latent process $\Y$ and the parameters $\bm\theta_\TP$, given the noisy, incomplete spatial data $\Z_\TO$. Inference proceeds from the joint posterior distribution, $[\Y, \bm\theta_\TP \mid \Z_\TO]$. Since $\Y \equiv (\Y_\TO^\top, \Y_\TM^\top)^\top$, we use the decomposition, $[\Y, \bm\theta_\TP \mid \Z_\TO]=[\Y_\TM \mid \Y_\TO, \bm\theta_\TP]\times[\Y_\TO, \bm\theta_\TP \mid \Z_\TO]$, since conditional on $\Y_\TO$, $\Y_\TM$ is independent of $\Z_\TO$. This enables us to address inference on the parameters and $\Y_\TO$ separately from inference on $\Y_\TM$. By Bayes' rule, the joint posterior of $\Y_\TO$ and $\bm\theta_\TP$ is given by $[\Y_\TO, \bm\theta_\TP \mid \Z_\TO] \propto [\Z_\TO \mid \Y_\TO, \bm\theta_\TO] \times [\Y_\TO \mid \bm\theta_\TP] \times [\bm\theta_\TP]$. The predictive distribution, $[\Y_\TO \mid \Z_\TO]$ can be obtained by integrating out $\bm\theta_\TP$. Similarly, predictions of $\Y_\TM$ alone are obtained via the integral $[\Y_\TM \mid \Z_\TO] = \int\int [\Y_\TM \mid \Y_\TO, \bm\theta_\TP] \times [\Y_\TO, \bm\theta_\TP \mid \Z_\TO]~\d\bm\theta_\TP ~\d\Y_\TO$, using the decomposition introduced above. Practical aspects are discussed in Section \ref{sec:Bayesian_inference}.

\subsection{Copula-based spatial models}\label{sec:copulas}

An $N$-variate copula function, $\mathfrak{C}_{1:N}(u_1, ..., u_N)$,  $u_1, ..., u_N \in [0, 1]$, is a joint cumulative distribution function (CDF) for a random vector $\mathbf{U} \equiv (U_1, ..., U_N)^\top$, where marginally each $U_j$, $j = 1, ..., N$, is uniformly distributed on $[0, 1]$. Let $\Y \equiv (Y_1,...,Y_N)^\top$ be an $N$-variate random vector. The notation here is not necessarily spatial, but it can be made so by letting $Y_j$ be $Y(A_j)$, for $j = 1, ..., N$. 

Now write the joint CDF of $\Y$ as $F_{1:N}(y_1, ..., y_N)$ and the joint PDF as $f_{1:N}(y_1, ..., y_N)$ for $y_1, ..., y_N \in \mathbb{R}$. Let $F_{j}(y_j)$ denote the marginal CDF of $Y_j$, for $j=1,...,N$. We link $\mathbf{U}$ and $\mathbf{Y}$ by transforming $U_j \equiv F_{j}(Y_j)$ or, conversely, $Y_j \equiv F_{j}^{-1}(U_j)$, where $F_j^{-1}(\cdot)$ is the quantile function, for $j=1,...,N$. For realizations $y_j$ and $u_j$, the same relationship holds: $u_j \equiv F_j(y_j)$ for $u_j \in [0, 1]$ and $y_j \in \mathbb{R}$. By Sklar's theorem \citep{Sklar1959}, combining the marginal CDFs $\{F_{j}(y_j): y_j \in \mathbb{R}\}_{j=1}^N$ with a copula function, $\mathfrak{C}_{1:N}(\cdot)$,  results in the joint CDF of $\Y$. That is, for $y_1, ..., y_N \in \mathbb{R}$,
\begin{equation}
F_{1:N}(y_1, ..., y_N) \equiv \mathfrak{C}_{1:N}(F_{1}(y_1), ..., F_{N}(y_N)).\label{eqn:sklars_theorem}
\end{equation}
When the marginal distributions are all continuous, \eqref{eqn:sklars_theorem} is unique. Assuming that $\mathfrak{C}_{1:N}(\cdot)$ and all marginal CDFs are absolutely continuous, we obtain the joint density of $\mathbf{Y} \equiv (Y_1, ..., Y_N)^\top$ from \eqref{eqn:sklars_theorem}. That is, for $y_1, ..., y_N \in \mathbb{R}$,
\begin{equation}
f_{1:N}(y_1, ..., y_N) \equiv \left(\prod_{j=1}^N f_{j}(y_j)\right) \times \mathfrak{c}_{1:N}(F_{1}(y_1), ..., F_{N}(y_N)),\label{eqn:copula_pdf}
\end{equation}
where $f_{j}(y_j)$ is the marginal PDF of $Y_j$, for $j = 1, ..., N$; and $\mathfrak{c}_{1:N}(\cdot)$ is the joint copula density, $\mathfrak{c}_{1:N}(u_1, ..., u_N) \equiv \partial^N \mathfrak{C}(u_1, ..., u_N)/\partial u_1... \partial u_N$, where $u_1, ..., u_N \in [0,1]$. 

In a copula-based hierarchical spatial-statistical model, \eqref{eqn:copula_pdf} defines the process model. We set $[\Y \mid \bm\theta_\TP] = (\prod_{j=1}^N f_{j}(y_j)) \times \mathfrak{c}_{1:N}(F_{1}(y_1), ..., F_{N}(y_N))$, where the right-hand side implicitly depends on $\bm\theta_\TP$ through the marginal CDFs and PDFs, $\{F_{j}(\cdot): j=1,...,N\}$ and $\{f_j(\cdot): j = 1, ..., N\}$, respectively, and through the joint copula density $\mathfrak{c}_{1:N}(\cdot)$. The vector $\bm\theta_\TP$ typically includes parameters such as linear-model coefficients $\bm\beta$ that parameterize the marginal expectations (e.g., $E(Y(A_j) \mid \bm\theta_\TP) = \mathbf{x}(A_j)^\top \bm\beta$, where $\x(A_j)^\top$ denotes a vector of covariates defined on BAU $A_j$), and a marginal variance parameter $\sigma^2_\TP \equiv \mathrm{var}(Y(A_j) \mid \bm\theta_\TP)$, for $j = 1 , ..., N$. Next, we develop spatial copulas to insert into the process model, $[\Y \mid \bm\theta_\TP]$.

\subsection{The Gaussian spatial random effects (Gau-SRE) copula}\label{sec:Gaussian_copula_SRE}

In order to construct the Gaussian SRE (Gau-SRE) copula, we first introduce the Gaussian SRE model. Let the process $Y(\cdot)$ have an SRE component, which is a linear combination of the $b$-dimensional vector of basis functions $\S(\cdot)$, where the coefficients are random \citep{Cressie2008}. Then, let $\S(A_j) \in \mathbb{R}^b$ be a vector of $b$ spatial basis functions \citep{Cressie2008} integrated over the BAU $A_j$ or perhaps evaluated at its centroid, for $j = 1, ..., N$. Let $\S \in \mathbb{R}^{N \times b}$ be a matrix where the $j$-th row is $\S(A_j)^\top$, for $j = 1, ..., N$. We can also write $\S \equiv (\S_\TO^\top, \S_\TM^\top)^\top$, where $\S_\TO$ and $\S_\TM$ are the $K\times b$ and $L\times b$ matrices of spatial basis functions for the observed and missing BAUs, respectively. 

Let the basis-function coefficients $\etab \in \mathbb{R}^b$ be a $b$-variate Gaussian random vector with mean vector $\mathbf{0}$ and positive-definite covariance matrix $\E$, and let $\xib$ be an $N$-variate Gaussian random vector, independent of $\etab$, with mean vector $\mathbf{0}$ and covariance matrix  $\sigma^2_\xi\I_N$, where $\I_N$ is the $N$-dimensional identity matrix. For reasons explained in Section \ref{sec:identifiability}, we fix $\sigma^2_\xi = 1$. Then, define the Gaussian random vector $\W \equiv (W(A_1), ..., W(A_N))^\top$ as,
\begin{equation}
    \W = \S\etab + \xib.\label{eqn:Gaussian_SRE_model}
\end{equation}
Now $\W$ in \eqref{eqn:Gaussian_SRE_model} can also be represented hierarchically as follows \citep{kang2011bayesian}: $\W \mid \etab \sim \MVG(\S\etab, \I_N)$ and $\etab \mid \E \sim \MVG(\mathbf{0}, \E)$, where $\mathrm{MVG}(\bm\mu,\bm\Sigma)$ denotes a multivariate Gaussian distribution with mean vector $\bm\mu$ and covariance matrix $\bm\Sigma$. Integrating out $\etab$ gives the unconditional distribution of $\W$, which is multivariate Gaussian with mean vector $\mathbf{0}$ and covariance matrix $\bm\Sigma^{\SRE}\equiv \S\E\S^\top + \I_N$. The form of $\bm\Sigma^{\SRE}$ enables us to leverage the Sherman-Morrison-Woodbury identities to efficiently compute its inverse and determinant \citep{Henderson1981}; that is, $(\bm\Sigma^{\SRE})^{-1} = \I_N - \S (\S^\top \S + \E^{-1})^{-1}\S^\top$ and $\det(\bm\Sigma^{\SRE}) = \det(\S^\top \S + \E^{-1})\det(\E)$.

In Proposition \ref{prop:Gau_cop_with_SRE}, we show how to obtain the Gau-SRE copula, where we use the notation $\Phi$ and $\phi$ to denote the CDF and PDF, respectively, of a univariate standard Gaussian random variable, and hence $\Phi^{-1}$ is its quantile function. 

\begin{proposition}\label{prop:Gau_cop_with_SRE}
Let $F_{\mathcal{G},1:N}(\cdot;\bm\Sigma^{\mathrm{SRE}})$ denote an $N$-variate Gaussian CDF with mean vector $\mathbf{0}$ and covariance matrix $\bm\Sigma^{\mathrm{SRE}}\equiv \S\E\S^\top + \I_N$. For $0 \leq u_1, ..., u_N \leq 1$, a Gaussian copula with parameter given by an SRE covariance-matrix is,
\begin{equation}
\mathfrak{G}_{1:N}^{\mathrm{SRE}}(u_1, ..., u_N; \bm\Sigma^{\mathrm{SRE}}) = F_{\mathcal{G}, 1:N}(\sigma_{1}^{\mathrm{SRE}} \Phi^{-1}(u_1),...,\sigma_{N}^{\mathrm{SRE}} \Phi^{-1}(u_N); \bm{\Sigma}^{\mathrm{SRE}}),\label{eqn:appendix_SRE_gau_copula}
\end{equation}
where $\sigma_j^{\mathrm{SRE}}$ is the square-root of the $j$-th diagonal element of $\bm{\Sigma}^{\mathrm{SRE}}$. The associated joint copula density is,
\begin{equation}
\mathfrak{g}_{1:N}^{\mathrm{SRE}}(u_1, ..., u_N; \bm\Sigma^{\mathrm{SRE}}) = \frac{f_{\mathcal{G},1:N}(\sigma_{1}^{\mathrm{SRE}} \Phi^{-1}(u_1),...,\sigma_{N}^{\mathrm{SRE}} \Phi^{-1}(u_N); \bm{\Sigma}^{\mathrm{SRE}})}{\prod_{j=1}^N \phi(\Phi^{-1}(u_j))/\sigma_j^{\mathrm{SRE}}}.\label{eqn:full_SRE_gau_copula_density}
\end{equation}
\end{proposition}
\noindent The proof (and all other proofs) is given in Supplement \ref{sec:appendix_proofs}. 

Copulas are invariant to monotonic transformations of the marginal variables (e.g., scaling operations). If the covariance matrix's diagonals $\sigma_1^2, ..., \sigma_N^2$ are free parameters only constrained to be positive, the invariance to scaling operations on the marginals makes them non-identifiable. However, the diagonal elements of $\bm\Sigma^\SRE$, namely $(\sigma^{\mathrm{SRE}}_1)^2, ..., (\sigma^{\mathrm{SRE}}_N)^2$, are constrained by the SRE parameterization. Provided that the constraints imposed in Section \ref{sec:identifiability} are followed, all parameters remain identifiable. The purpose of the parameterization given in \eqref{eqn:appendix_SRE_gau_copula} in Proposition \ref{prop:Gau_cop_with_SRE} is that the diagonal elements of $\bm\Sigma^{\mathrm{SRE}}$, which depend on $\S$ and $\E$, can be computed quickly and directly without computing and storing the large $N\times N$ matrix prior to extracting the diagonals. Traditionally, copulas have relied on the correlation matrix for their parameterization, but our new approach using \eqref{eqn:appendix_SRE_gau_copula} avoids this. 

The proposition below establishes the connection between the process $\Y$ and $\W$ given in \eqref{eqn:Gaussian_SRE_model}, through a special trans-Gaussian model. 

\begin{proposition}\label{prop:Gau_anamorphosis}
Define the (Gaussian) anamorphosis \citep{Matheron1976}: Set  $W(A_j) = \sigma_j^{\SRE}\Phi^{-1}(F_j(Y(A_j)))$, for $j = 1, ..., N$, where $\W \equiv (W(A_1), ..., W(A_N))^\top$ jointly follows a multivariate Gaussian distribution with mean vector $\mathbf{0}$ and covariance matrix $\bm\Sigma^\SRE$. Then, a transformation of variables shows that, for $y_1, ..., y_N\in\mathbb{R}$, the joint PDF of $\Y$ is, 
\begin{align}
   [\Y \mid \bm\theta_\TP] 
   &=
   \left(\prod_{j=1}^Nf_j(y_j)\right)\mathfrak{g}_{1:N}^{\mathrm{SRE}}(F_1(y_1), ..., F_N(y_N); \bm\Sigma^{\mathrm{SRE}}),\label{eqn:Gau_SRE_cop_model}
\end{align}
where $\mathfrak{g}_{1:N}^{\mathrm{SRE}}$ is given by \eqref{eqn:full_SRE_gau_copula_density}.
\end{proposition}

Proposition \ref{prop:Gau_anamorphosis} allows us to obtain the Gau-SRE copula model for any subvector of $\Y$. Let $\{F_{\TO k}: k = 1, ..., K\}$ and $\{f_{\TO k}: k = 1, ..., K\}$ denote the marginal CDFs and PDFs of $\{Y(A_{\TO k}):k=1,...,K\}$, respectively. Also let $\xib_\TO$ be the $K$-variate sub-vector of $\xib$ in \eqref{eqn:Gaussian_SRE_model}, with mean vector $\mathbf{0}$ and covariance matrix $\I_K$.  Hence, $\W_\TO \equiv (W(A_{\TO 1}), ..., W(A_{\TO K}))^\top \equiv \S_\TO \etab + \xib_\TO$, where $\S_\TO$ is the corresponding $K\times b$ submatrix of $\S$. Unconditional on $\etab$, $\W_\TO$ follows a multivariate Gaussian distribution with mean vector $\mathbf{0}$ and covariance matrix, $\bm\Sigma_{\TO\TO}^\SRE \equiv \S_\TO \E\S_\TO^\top +\I_K$. Applying Proposition \ref{prop:Gau_anamorphosis}, we see that, for $y_{\TO 1}, ..., y_{\TO K} \in \mathbb{R}$,
\begin{align}
[\Y_\TO \mid \bm\theta_\TP]= \left(\prod_{k=1}^Kf_{\TO k}(y_{\TO k})\right) \GCD_{1:K}^{\SRE}\left(F_{\TO 1}(y_{\TO 1}),...,F_{\TO K}(y_{\TO K}); \bm{\Sigma}^{\mathrm{SRE}}_{\TO\TO}\right).\label{eqn:Gau_OBS_ONLY}
\end{align}
This property, that the copula representation is conserved under marginalization, is an attractive feature of our parameterization.

\subsection{Parameter identifiability in $\bm\Sigma^{\mathrm{SRE}}$}\label{sec:identifiability}

The construction of \eqref{eqn:appendix_SRE_gau_copula} finesses the normalization of $\bm\Sigma^{\mathrm{SRE}}$ into its correlation matrix, $\mathbf{R}^{\mathrm{SRE}}$, so avoiding unnecessary and inefficient matrix operations. However, care is needed when parameterizing $\bm\Sigma^\SRE$.  Suppose $\E \equiv \sigma^2_E \times \R$, where $\sigma^2_E > 0$ and $\R$ is a $b\times b$ correlation matrix; and suppose $\xib$ has covariance matrix $\sigma^2_\xi \I_N$. Then, $\bm\Sigma^{\mathrm{SRE}} = \sigma^2_E \S \R \S^\top + \sigma^2_{\xi}\mathbf{I}_N$,
which is over-parameterized. The parameters $\sigma^2_E$ and $\sigma^2_{\xi}$ cannot be identified together, but their ratio can be. Fixing $\sigma_\xi^2 = 1$, we have $\bm\Sigma^\SRE = \theta_s \left(\S\R\S^\top\right) + \I_N$, where $\theta_s$ is interpreted as the relative importance of the spatially dependent $\S\etab$, relative to the vector of independent standard Gaussian vector, $\xib$. In Supplement \ref{sec:appendix_multiresolution}, we extend this parameterization to multiple resolutions of spatial basis functions.

\subsection{The t spatial random effects (t-SRE) copula}\label{sec:t_copula_SRE}

Here, we define a novel t spatial process with an SRE component. Let $\gamma$ follow a Gamma distribution with shape and rate parameters both equal to $\nu/2$, where $\nu > 2$. That is, $\gamma \mid \nu \sim \Gam(\nu/2, \nu/2)$. Then, define the $N$-variate random vector,
\begin{equation}
\V \equiv \gamma^{-1/2}\times \W = \gamma^{-1/2} \times \left(\S\etab + \xib\right).\label{eqn:V}
\end{equation}
As with the Gaussian SRE model, $\V$ in \eqref{eqn:V} also has a hierarchical representation. Write \eqref{eqn:V} as $\V \equiv \S\etabb + \xibb$, where $\etabb \equiv \gamma^{-1/2}\etab$ and $\xibb \equiv \gamma^{-1/2}\xib$. Then, $\V \mid \etabb, \gamma \sim \MVG(\S\etabb, \gamma^{-1}\I_N)$ and $\etabb \mid \gamma,\E \sim \MVG(\mathbf{0}, \gamma^{-1}\E)$. This uses a well known representation of the multivariate t distribution as a scale-mixture of Gaussian distributions \citep[e.g.,][pp. 256--258]{Koop2007}. From this hierarchical representation, integrating out $\etabb$ and $\gamma$ yields the unconditional distribution of $\V$, namely a multivariate t distribution on $\nu > 2$ degrees of freedom with mean vector $\mathbf{0}$ and positive-definite scale matrix $\bm\Sigma^{\SRE}$.

In Proposition \ref{prop:t_copula_SRE} that follows, we show how to obtain the t-SRE copula. Let $F_{t,1:N}(\cdot; \bm\Sigma, \nu)$ and $f_{t,1:N}(\cdot; \bm\Sigma, \nu)$ denote the CDF and PDF of the $N$-variate t distribution with $\nu > 2$ degrees of freedom, respectively, with mean vector $\mathbf{0}$ and positive-definite scale matrix $\bm\Sigma$. For $\mathbf{v} \equiv (v_1, ..., v_N)^\top\in \mathbb{R}^N$,   
\begin{equation}
f_{t,1:N}(v_1, ..., v_N;\bm\Sigma, \nu) = \frac{\Gamma\left(\frac{\nu+N}{2}\right)\det(\bm\Sigma)^{-\frac{1}{2}}}{\Gamma\left(\frac{\nu}{2}\right)\left(\nu\pi\right)^{\frac{N}{2}}}\left(1+\frac{1}{\nu}\mathbf{v}^{\top}\bm\Sigma^{-1}\mathbf{v}\right)^{\left(-\frac{\nu+N}{2}\right)},\label{eqn:t_PDF}
\end{equation}
where $\Gamma(\cdot)$ is the Gamma function. Then, the t copula with parameter given by an SRE covariance matrix, can be derived as follows.

\begin{proposition}\label{prop:t_copula_SRE}
Let $F_{t,1:N}(\cdot; \bm\Sigma, \nu)$ and $f_{t,1:N}(\cdot; \bm\Sigma, \nu)$ be the CDF and PDF, respectively, of an $N$-variate t distribution on $\nu > 2$ degrees of freedom with mean vector $\mathbf{0}$ and positive-definite scale matrix $\bm\Sigma$. For $0 \leq u_1, ..., u_N \leq 1$, the $N$-variate t copula with an SRE scale matrix and its joint copula density are, respectively,
\begin{align}
\mathfrak{T}_{1:N}^{\mathrm{SRE}}(u_1, ..., u_N; \bm{\Sigma}^{\mathrm{SRE}}, \nu) &= F_{t,1:N}(\sigma_1^{\mathrm{SRE}} T^{-1}_\nu(u_1), ..., \sigma_N^{\mathrm{SRE}} T^{-1}_\nu(u_N); \bm{\Sigma}^{\mathrm{SRE}}, \nu)\label{eqn:appendix_SRE_t_copula}\\
\mathfrak{t}_{1:N}^{\mathrm{SRE}}(u_1, ..., u_N; \bm{\Sigma}^{\mathrm{SRE}}, \nu) &= \frac{f_{t,1:N}(\sigma_1^{\mathrm{SRE}} T^{-1}_\nu(u_1), ..., \sigma_N^{\mathrm{SRE}} T^{-1}_\nu(u_N); \bm{\Sigma}^{\mathrm{SRE}}, \nu)}{\prod_{j=1}^N t_\nu(T^{-1}_\nu(u_j))/ \sigma_j^{\mathrm{SRE}}},\label{eqn:appendix_SRE_t_copula_density}
\end{align}
where $t_\nu(\cdot)$, $T_\nu(\cdot)$ , and $T_\nu^{-1}(\cdot)$ are the PDF, CDF, and quantile function of a standardized t distribution on $\nu > 2$ degrees of freedom, respectively, and recall that $\sigma_1^{\mathrm{SRE}}, ..., \sigma_N^{\mathrm{SRE}}$ are the square-roots of the diagonal elements of $\bm\Sigma^{\mathrm{SRE}}$.
\end{proposition}

Similar to the Gaussian-SRE copula, the proposition below shows how a particular elementwise transformation of $\Y$ to $\V$ induces the t-SRE copula model for $[\Y\mid \bm\theta_\TP]$. 

\begin{proposition}\label{prop:t_anamorphosis}
    Define the t anamorphosis function as $V(A_j) \equiv \sigma^{\SRE}_jT_\nu^{-1}(F_j(Y(A_j)))$, $j = 1, ..., N$. Let $\V \equiv (V(A_1), ..., V(A_N))^\top$ follow a multivariate t distribution on $\nu > 2$ degrees of freedom with mean vector $\mathbf{0}$ and positive-definite scale matrix $\bm\Sigma^\SRE$. Then, by a transformation of variables, and for $y_1, ..., y_N \in \mathbb{R}$, it follows that,
    \begin{align}
        [\Y \mid \bm\theta_\TP]
        &= \left(\prod_{j=1}^N f_j(y_j)\right)\mathfrak{t}_{1:N}^{\mathrm{SRE}}(F_1(y_1), ..., F_N(y_N); \bm{\Sigma}^{\mathrm{SRE}}, \nu),\label{eqn:t_SRE_copula_model}
    \end{align}
    where $\mathfrak{t}_{1:N}^\SRE$ is given by \eqref{eqn:appendix_SRE_t_copula_density}
\end{proposition}

Proposition \ref{prop:t_anamorphosis} provides a straightforward method for deriving $[\Y_\TO \mid \bm\theta_\TP]$ under the t-SRE copula model. Define $\V_\TO  \equiv (V(A_{\TO 1}), ..., V(A_{\TO K}))^\top \equiv \gamma^{-1/2}\times \W_\TO$, which unconditionally follows a $K$-variate t distribution on $\nu > 2$ degrees of freedom with mean vector $\mathbf{0}$ and positive-definite scale matrix $\bm\Sigma_{\TO\TO}^\SRE$. By a transformation of variables associated with the t anamorphosis in Proposition \ref{prop:t_anamorphosis}, we obtain,
\begin{align}
[\Y_\TO \mid \bm\theta_\TP]= \left(\prod_{k=1}^Kf_{\TO k}(y_{\TO k})\right)\TCD_{1:K}^\SRE(F_{\TO1}(y_{\TO 1}), ..., F_{\TO K}(y_{\TO K}); \bm\Sigma^\SRE_{\TO\TO}, \nu).\label{eqn:tanamorphosis}
\end{align}
Hence, the representation of a t-SRE copula is conserved under marginalization.

\section{Markov chain Monte Carlo (MCMC) sampling}\label{sec:Bayesian_inference}

Our copula-based SRE models enable fast Bayesian inference for large spatial datasets. In this section, we show how the hierarchical representation of the copula-based model for $[\Y \mid \bm\theta_\TP]$ enables a computationally efficient Markov chain Monte Carlo (MCMC) algorithm to obtain the posterior distributions of $\Y$ and $\bm\theta_\TP$ given the data $\Z_\TO$. Section \ref{sec:Gibbs} sets out the type of algorithm we use. Sections \ref{sec:MCMC_Gaussian_case} and \ref{sec:MCMC_t_case} particularize this to, respectively, the Gau-SRE copula model and the t-SRE copula model.

\subsection{Gibbs sampler with Metropolis-Hastings steps}\label{sec:Gibbs}

We use a Gibbs sampler \citep{Gelfand1990} with some Metropolis-Hastings steps \citep{Metropolis1953, Hastings1970}, which we refer to as Metropolis-within-Gibbs (MwG) steps. We require the full-conditional distributions of all parameters and latent variables in the model. We denote the full-conditional density for $\bm\theta_\TP$ (say) as $[\bm\theta_\TP \mid \rest]$, where ``rest'' is used to represent all remaining variables. In the example, ``rest'' stands for $\Z_\TO$ and $\Y_\TO$. We write $[\bm\theta_\TP \mid \rest] \propto [\Z_\TO \mid \Y_\TO, \bm\theta_\TO] \times [\Y_\TO \mid \bm\theta_\TP] \times [\bm\theta_\TP]$, which simplifies to $[\bm\theta_\TP \mid \rest] \propto [\Y_\TO \mid \bm\theta_\TP] \times [\bm\theta_\TP]$, because the term with $\Z_\TO$ does not depend on $\bm\theta_\TP$.

\subsection{MCMC for the Gau-SRE copula model}\label{sec:MCMC_Gaussian_case}

The key to computationally efficient Bayesian inference in our copula-based SRE models is to leverage the fact that, due to \eqref{eqn:Gaussian_SRE_model}, the elements of $\W$ (and therefore $\Y$) are conditionally independent of each other given $\etab$.  Therefore, by Proposition \ref{prop:Gau_anamorphosis}, which connects $\W$ to $\Y$ in the Gau-SRE copula model, we can represent the process model in terms of $[\Y \mid \etab, \bm\theta_\TP]  = [\Y_\TO \mid \etab, \bm\theta_\TP] \times [\Y_\TM \mid \etab, \bm\theta_\TP]$ and $[\etab \mid \bm\theta_\TP]$. Using \eqref{eqn:Gaussian_SRE_model}, we can see that $\W\mid \etab, \bm\theta_\TP \sim \MVG(\S\etab, \I_N)$; it also easily follows that $\W_\TO \mid \etab, \bm\theta_\TP \sim \MVG(\S_\TO\etab, \I_N)$ and $\W_\TM \mid \etab, \bm\theta_\TP \sim \MVG(\S_\TM\etab, \I_N)$. Then by applying the same transformation of variables used in Proposition \ref{prop:Gau_anamorphosis}, we have, for $y_{\TO 1}, ..., y_{\TO K}, y_{\TM 1}, ..., y_{\TM L}\in\mathbb{R}$,
\begin{align}
    [\Y_\TO \mid \etab, \bm\theta_\TP] &= \prod_{k=1}^K\left(\frac{\sigma^{\SRE}_{\TO k}f_{\TO k}(y_{\TO k})}{\phi(\Phi^{-1}(F_{\TO k}(y_{\TO k})))} \frac{\exp\{-0.5(w_{\TO k} - \S(A_{\TO k})^\top\etab)^2\}}{(2\pi)^{1/2}}\right),\label{eqn:obs_gau_conditional}\\
    [\Y_\TM \mid \etab, \bm\theta_\TP] &= \prod_{l=1}^L\left(\frac{\sigma^{\SRE}_{\TM l}f_{\TM l}(y_{\TM l})}{\phi(\Phi^{-1}(F_{\TM l}(y_{\TM l})))} \frac{\exp\{-0.5(w_{\TM l} - \S(A_{\TM l})^\top\etab)^2\}}{(2\pi)^{1/2}}\right),\label{eqn:preds_gau_conditional}
\end{align}
where $\sigma^\SRE_{\TO k}$ is the square-root of the $k$-th diagonal element of $\bm\Sigma^\SRE_{\TO\TO} \equiv \S_\TO \E \S_\TO^\top + \I_K$; $\sigma_{\TM l}^{\SRE}$ is the square-root of the $l$-th diagonal element of $\bm\Sigma^\SRE_{\TM\TM} \equiv \S_\TM \E \S_\TM^\top + \I_L$; $w_{\TO k} \equiv \sigma^\SRE_{\TO k}\Phi^{-1}(F_{\TO k}(y_{\TO k}))$ for $k=1,...,K$; and $w_{\TM l} \equiv \sigma^\SRE_{\TM l}\Phi^{-1}(F_{\TM l}(y_{\TM l}))$ for $l = 1, ..., L$.

Under the hierarchical representation of the Gau-SRE copula model, the full joint posterior  can be written as, 
\begin{align}
    [\Y_\TM, &\Y_\TO, \etab, \bm\theta_\TP \mid \Z_\TO] = [\Y_\TM \mid \Y_\TO, \etab, \bm\theta_\TP] [\Y_\TO, \etab,\bm\theta_\TP \mid \Z_\TO]\nonumber\\
    &\propto [\Y_\TM \mid \etab, \bm\theta_\TP] \Big\{[\Z_\TO\mid \Y_\TO, \bm\theta_\TO] [\Y_\TO \mid \etab, \bm\theta_\TP] [\etab\mid \bm\theta_\TP][\bm\theta_\TP]\Big\}.\label{eqn:joint_posterior_GauSRE}
\end{align}
To sample from the full joint posterior on the left-hand side (LHS) of \eqref{eqn:joint_posterior_GauSRE}, we sample from the posterior distribution $[\Y_\TO, \etab, \bm\theta_\TP \mid \Z]$ and then sample from the conditional distribution $[\Y_\TM \mid \Y_\TO, \etab, \bm\theta_\TP] = [\Y_\TM \mid \etab, \bm\theta_\TP]$ in \eqref{eqn:preds_gau_conditional}. Hence, we need the full-conditional distributions of $\Y_\TO$, $\etab$, and $\bm\theta_\TP$. In Supplement \ref{sec:appendix_FCs_Gau}, we show the full conditionals and $[\Y_\TM \mid \Y_\TO, \etab, \bm\theta_\TP]$ are either available in closed form or can be sampled efficiently using MwG steps by leveraging the conditional independence of the elements of $\Y_\TO$ and $\Y_\TM$ given $\etab$.  

\subsection{MCMC for the t-SRE copula model}\label{sec:MCMC_t_case}

Efficient MCMC for the t-SRE copula model works similarly. From \eqref{eqn:V}, the elements of $\V$ are conditionally independent given $\etabb$, $\gamma$, and $\bm\theta_\TP$. Then we can exploit the fact that Proposition \ref{prop:t_anamorphosis} links $\V$ to $\Y$ via a transformation of random variables to rewrite the process model of the t-SRE copula models in terms of $[\Y \mid \etabb, \gamma, \bm\theta_\TP] = [\Y_\TO \mid \etabb, \gamma, \bm\theta_\TP] \times [\Y_\TM \mid \etabb, \gamma, \bm\theta_\TP]$, $[\etabb\mid \gamma,\bm\theta_\TP]$, and $[\gamma \mid \bm\theta_\TP]$.  In particular, we have $\V \mid \etabb, \gamma, \bm\theta_\TP \sim \MVG(\S\etabb, \gamma^{-1}\I_N)$. It also follows that $\V_\TO \mid \etabb, \gamma, \bm\theta_\TP \sim \MVG(\S_\TO\etabb, \gamma^{-1}\I_K)$ and $\V_\TM \mid \etabb, \gamma, \bm\theta_\TP \sim \MVG(\S_\TM\etabb, \gamma^{-1}\I_L)$. By applying the same transformation of variables used in Proposition \ref{prop:t_anamorphosis}, we obtain,
for $y_{\TO 1}, ..., y_{\TO K}, y_{\TM 1}, ..., y_{\TM L} \in \mathbb{R}$,
\begin{align}
    [\Y_\TO \mid \etabb, \gamma, \bm\theta_\TP] &= \prod_{k=1}^K\left(\frac{\sigma^{\SRE}_{\TO k}f_{\TO k}(y_{\TO k})}{t_\nu(T_\nu^{-1}(F_{\TO k}(y_{\TO k})))} \frac{\exp\{-0.5\gamma(v_{\TO k} - \S(A_{\TO k})^\top\etabb)^2\}}{(2\pi\gamma^{-1})^{1/2}}\right),\label{eqn:obs_t_conditional}\\
    [\Y_\TM \mid \etabb,\gamma, \bm\theta_\TP] &= \prod_{l=1}^L\left(\frac{\sigma^{\SRE}_{\TM l}f_{\TM l}(y_{\TM l})}{t_\nu(T_\nu^{-1}(F_{\TM l}(y_{\TM l})))} \frac{\exp\{-0.5\gamma(v_{\TM l} - \S(A_{\TM l})^\top\etabb)^2\}}{(2\pi\gamma^{-1})^{1/2}}\right),\label{eqn:preds_t_conditional}
\end{align}
where $v_{\TO k} \equiv \sigma^\SRE_{\TO k}T_\nu^{-1}(F_{\TO k}(y_{\TO k}))$, $k = 1, ..., K$, and $v_{\TM l} \equiv  \sigma^\SRE_{\TM l}T_\nu^{-1}(F_{\TM l}(y_{\TM l}))$, $l = 1, ..., L$. 

The conditional distributions \eqref{eqn:obs_t_conditional} and \eqref{eqn:preds_t_conditional} are needed to write the full joint posterior distribution under the hierarchical representation of the t-SRE copula model. This is, 
\begin{align}
    [&\Y_\TM, \Y_\TO, \etabb, \gamma, \bm\theta_\TP \mid \Z_\TO] = [\Y_\TM \mid \Y_\TO, \etabb, \gamma, \bm\theta_\TP][\Y_\TO, \etabb, \gamma, \bm\theta_\TP \mid \Z_\TO]\nonumber\\
    &\propto [\Y_\TM \mid \etabb, \gamma, \bm\theta_\TP]\Big\{[\Z_\TO \mid \Y_\TO, \bm\theta_\TO][\Y_\TO \mid \etabb, \gamma, \bm\theta_\TP][\etabb \mid \gamma, \bm\theta_\TP][\gamma \mid \bm\theta_\TP][\bm\theta_\TP]\Big\},\label{eqn:hierarchical_joint_t}
\end{align}
where $\bm\theta_\TP$ includes the degrees-of-freedom parameter $\nu > 2$, in addition to all those present for the Gau-SRE model. To sample from the full joint posterior on the LHS of \eqref{eqn:hierarchical_joint_t}, we sample from the joint posterior $[\Y_\TO, \etabb, \gamma,\bm\theta_\TP \mid \Z_\TO]$ and also the conditional distribution, $ [\Y_\TM \mid \Y_\TO, \etabb, \gamma,\bm\theta_\TP] = [\Y_\TM \mid \etabb, \gamma, \bm\theta_\TP]$ in \eqref{eqn:preds_t_conditional}. Like in the Gaussian case, Supplement \ref{sec:appendix_FCs_t} shows that the full-conditional distributions of $\Y_\TO$, $\etabb$, $\gamma$, and $\bm\theta_\TP$, as well as \eqref{eqn:preds_t_conditional}, are either available in closed form or can be sampled efficiently with MwG steps as a consequence of our SRE representation of the t-SRE copula model. 

\section{Simulation studies}\label{sec:simulation_studies}

The simulation studies given in this section demonstrate parameter recovery, prediction, uncertainty quantification, and computational efficiency for the copula-based hierarchical spatial-statistical models developed in Section \ref{sec:spatial_copula_model}. Section \ref{sec:models} outlines four copula-based spatial-statistical models. Section \ref{sec:setup} outlines the structure of the simulation experiments. Section \ref{sec:sim_results} shows selected results and interprets them, with the remaining results given in Supplement \ref{sec:appendix_sims}. 

\subsection{Models}\label{sec:models}

Four basic models are used in this simulation study. They are obtained by using log-Gaussian (LG) and skew-Gaussian (SG) \citep{Azzalini1985} marginals with joint dependence captured by the Gau-SRE and t-SRE copula models. Hence, there are four models ($2\times 2$), which are labelled LG-Gau-SRE, SG-Gau-SRE, LG-t-SRE, and SG-t-SRE models, whose names indicate the marginals and the copula used. For both sets of marginals, we set $E(Y(A_j) \mid \bm\theta_\TP) = \exp\{\beta_0\}$ (a log-linear model with only an intercept term) and $\mathrm{var}(Y(A_j) \mid \bm\theta_\TP) = \sigma^2_\TP$, for $j = 1, ..., N$. Here, $\bm\theta_\TP$ depends on the copula model. For the LG-Gau-SRE model, $\bm\theta_\TP =(\beta_0, \sigma^2_\TP, \theta_s, \theta_r)^\top$, with an extra degrees-of-freedom parameter $\nu > 2$ for the t-SRE copula models, and a skewness parameter $\lambda$ for the SG-Gau-SRE and SG-t-SRE models. The parameters $\theta_s$ and $\theta_r$ (explained below) control spatial dependence in the copula. The log-Gaussian CDF for $Y(A_j)$ is written as $F_{LG,j}(y_j;\beta_0, \sigma^2_\TP )$, $y_j \in (0, \infty)$, and the skew-Gaussian CDF is written as $F_{SG, j}(y_j; \beta_0, \sigma^2_\TP, \lambda)$, $y_j \in \mathbb{R}$, where $\lambda \in \mathbb{R}$ is the skewness parameter of the skew-Gaussian distribution. Note that the LG-Gau-SRE model is equivalent to a log-Gaussian SRE model, where a version is  implemented  in the R package FRK \citep{SainsburyDale2024}. See Supplement \ref{sec:appendix_LGGau}--\ref{sec:appendix_FRK_SG} for more details.

\subsection{Structure of the simulation study}\label{sec:setup}

The spatial domain for the simulation study was the unit square, $D \equiv [0, 1]^2$, split into 10,000 equal-sized BAUs. Of these, 5,000 were `observed BAUs', $\mathcal{A}_{\TO}$, and 5,000 were `missing BAUs', $\mathcal{A}_\TM$. The choice of $K = 5,000$ observed BAUs was motivated by the size of the real datasets analyzed in Section \ref{sec:methane}. The missing BAUs were either randomly selected, hence `missing at random' (MAR), or systematically missing in blocks, hence `missing by design' (MBD). The MBD case was designed to provide a true test of the predictive performance of our models, since most missing BAUs are far from any observations. For the MBD case, the 2,500 BAUs in the top-left and bottom-right corners of the unit square were all designated as members of $\mathcal{A}_\TM$. Once chosen, the same $\mathcal{A}_{\TO}$ and $\mathcal{A}_{\TM}$ were used for each model. See Supplement \ref{sec:appendix_setup} for maps of the BAUs in the MAR and MBD configurations. 

Then, $R = 100$ realizations of $\Y$ and $\Z_\TO$ were simulated. Details are given below for a t-SRE copula model, and Supplement \ref{sec:appendix_other_models_sims} shows how datasets were simulated for models with Gau-SRE copulas. First, the true scale matrix for the random-effects vector was calculated as $\E^*(\theta_s, \theta_r) \equiv \theta_s\exp\{-\mathbf{D}/\theta_r\}$, elementwise, where $\mathbf{D}$ is a matrix whose $(i,j)$-th entry is the distance between the $i$-th and $j$-th spatial basis functions. The true values of the parameters were set as $\beta_0^* = \log(1000)$, $\lambda^* = -5$, $\theta_s^* = 10$, $\theta_r^* = \sqrt{2}/4$, and $\nu^* = 4$ for all models. For the models with log-Gaussian marginals, we set $(\sigma^*_\TP)^2 = (0.1)^2$. For the models with skew-Gaussian marginals, we set $(\sigma^*_\TP)^2 = (100)^2$. In both cases, we set $(\sigma^*_\TO)^2 = 0.05 (\sigma^*_\TP)^2$. The FRK package \citep{ZM2021, SainsburyDale2024} in R  was used to generate 36 bisquare spatial basis functions \citep{Cressie2008} with apertures of $0.375$ units on a regular grid across $D$ (see Supplement \ref{sec:appendix_setup}). To generate the simulated datasets, let $\bm\theta_\TP^*$ be a vector of the true parameter values for the model. For $r= 1, ..., R$, we simulated $\gamma^{(r)} \sim [\gamma \mid \bm\theta_\TP^*]$; then we simulated $\bar\etab^{(r)} \sim [\etabb \mid \bm\theta_\TP^*, \gamma^{(r)}]$; then we simulated $\Y^{(r)} \sim [\Y\mid \bar\etab^{(r)}, \gamma^{(r)}, \bm\theta_\TP^*]$, where recall $\Y^{(r)} \equiv ((\Y_\TO^{(r)})^\top, (\Y_\TM^{(r)})^\top)^\top$; and finally, we simulated the observed spatial data $\Z_\TO^{(r)} \sim [\Z_\TO \mid \Y_\TO^{(r)}, (\sigma^*_\TO)^2]$. 

For each of $R=100$ simulated datasets and for each model, fully Bayesian inference was carried out via MCMC. The weakly informative priors for the parameters are shown in Supplement \ref{sec:appendix_sim_priors}, and the initial values are described in Supplement \ref{sec:appendix_sim_initial_values}. We ran 45,000 MCMC iterations for every simulated dataset, discarding the first 5,000 as burn-in, and subsequently keeping one in every four samples. For the LG-Gau-SRE and LG-t-SRE models, it took, respectively, 7-8 minutes and 18-20 minutes to run the MCMC for 45,000 iterations. For the SG-Gau-SRE and SG-t-SRE models, it took, respectively, 16-18 minutes and 27-28 minutes. To speed up computations, the skew-Gaussian CDF and quantile function were evaluated at 101 points and interpolated using splines. We ran the MCMC for each dataset in parallel on Australia's National Computing Infrastructure's GADI high-performance computing cluster. 

\subsection{Prediction results}\label{sec:sim_results}

In what follows, data were simulated from the LG-t-SRE and SG-t-SRE models. We investigated the predictive performance and uncertainty quantification of the LG-t-SRE and SG-t-SRE models, and we compared them to several baseline models. The baseline models include a no-measurement-error (NME) version of the models (Supplement \ref{sec:appendix_NH_LGGau}) to assess whether accounting for measurement error provides any benefits; the corresponding LG-Gau-SRE and SG-Gau-SRE models to gauge the effect of using a misspecified copula model to spatially predict a spatial process based on the t-SRE copula; and fixed rank kriging (FRK; Supplement \ref{sec:appendix_FRK_general}--\ref{sec:appendix_FRK_SG}), which is an important dimension-reduction model for large spatial data that accounts for measurement error. 

For illustration, Fig. \ref{fig:sim_maps} displays the first ($r=1$) simulated dataset (including both the latent-process vector $\Y$ at $\mathcal{A}_\TO$ and $\mathcal{A}_\TM$ and spatial data $\Z_\TO$ at $\mathcal{A}_\TO$ only) from the LG-t-SRE and SG-t-SRE models when the missing BAUs are MBD. The predictions (posterior means) and their uncertainties, posterior standard deviations (PSDs), from the respective model fitted back to the respective dataset are also shown. For $j = 1, ..., N$ and $r = 1, ..., R$, the predicted value of $Y(\cdot)$ at the $j$-th BAU given the spatial data $\Z^{(r)}_\TO$ is taken to be $E(Y(A_j) \mid \Z^{(r)}_\TO)$, and the PSD, $\sqrt{\mathrm{var}(Y(A_j)\mid \Z^{(r)}_\TO)}$, is used to summarize prediction uncertainty. The prediction problem in Fig. \ref{fig:sim_maps} is clearly a challenging one, but the LG-t-SRE and SG-t-SRE models successfully recover the surface of the latent-process values. In what follows, we carry out extensive statistical comparisons of model properties across all $R = 100$ datasets to make valid statistical comparisons with alternative models.

\begin{figure}[!ht]
    \centering
    \includegraphics[width=\textwidth]{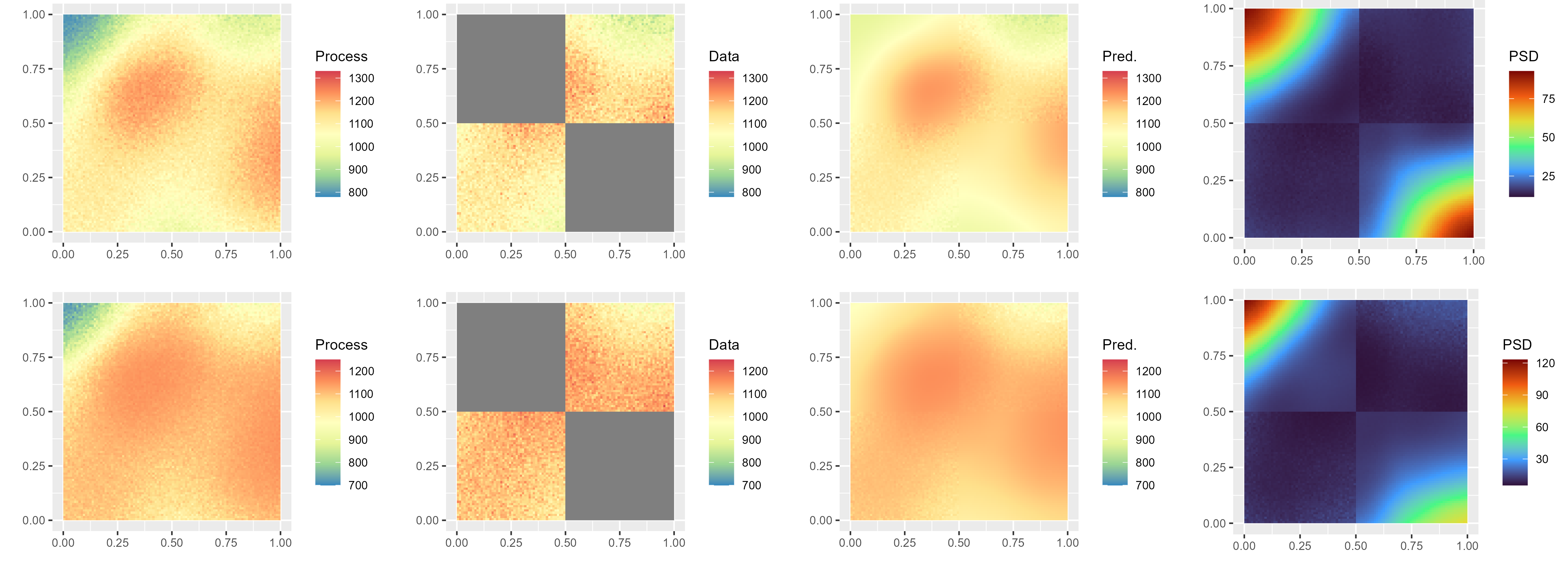}
    \caption{Top row and left to right: First, map of a single realization of a copula-based spatial process with log-Gaussian marginal distributions and a t-SRE copula; then the spatial data when the missing BAUs are missing by design; then a map of the predicted values of the spatial-statistical process; and a map of the posterior standard deviations (PSDs). Bottom row and left to right: As for the top row but with skew-Gaussian marginals.}
    \label{fig:sim_maps}
\end{figure}

To compare models in terms of predictive performance and uncertainty quantification, we use BAU-wise root-mean-squared prediction errors (RMSPEs) and empirical coverages (ECs). Let $\delta( \Z_\TO^{(r)}; A_j)$ be the predicted value of $Y(A_j)$ given the $r$-th spatial dataset $\Z_\TO^{(r)}$, for $j = 1, ..., N$ and $r= 1, ..., R$. Using the ubiquitous squared-error loss function, we take $\delta(\Z_\TO^{(r)}; A_j) \equiv E(Y(A_j) \mid \Z_\TO^{(r)})$ to be the optimal prediction of $Y(A_j)$. The RMSPE at BAU $A_j$ is defined as,
\begin{equation}
    \RMSPE(A_j) \equiv \sqrt{\frac{1}{R}\sum_{r=1}^R\left(\delta(\Z_\TO^{(r)}; A_j) - Y(A_j)^{(r)}\right)^2}, ~~~j = 1, ..., N.\label{eqn:RMSPE_definition}
\end{equation}
For EC, let $\alpha \in (0, 1)$, and let $I_{\lwr}^{(r,\alpha)}(A_j)$ and $I_{\upr}^{(r,\alpha)}(A_j)$ be the $\alpha/2$ and $(1-\alpha/2)$ quantiles of the marginal predictive distribution $[Y(A_j)\mid \Z_\TO^{(r)}]$, respectively, for $j = 1, ..., N$ and $r = 1, ..., R$. Estimates are denoted by $\hat{I}_{\lwr}^{(r,\alpha)}(A_j)$ and $\hat{I}_{\upr}^{(r,\alpha)}(A_j)$. Then, for any given model, the empirical coverage at BAU $A_j$ and level $\alpha$ is defined as,
\begin{equation}
    \EC^{(\alpha)}(A_j) \equiv \frac{1}{R}\sum_{r=1}^R \mathbb{I}\!\left(\hat{I}_{\lwr}^{(r,\alpha)}(A_j) \leq Y(A_j)^{(r)} \leq \hat{I}_{\upr}^{(r,\alpha)}(A_j)\right),~~~j = 1, ..., N,\label{eqn:EC_definition}
\end{equation} 
where $\mathbb{I}(\cdot)$ is the indicator function. Here, $R = 100$, and we chose $\alpha = 0.10$. 

In the present simulation study, recall that $R = 100$ datasets were simulated from the LG-t-SRE and SG-t-SRE models. For the datasets simulated from the LG-t-SRE model, Fig. \ref{fig:LG_prediction_comparisons} compares the RMSPEs under the LG-t-SRE model (the true model) to those under the NME version of the LG-t-SRE model (no measurement error), the LG-Gau-SRE model (a model with the correct marginals but the wrong copula), and a log-Gaussian spatial model in FRK. Similarly, for datasets simulated under the SG-t-SRE model, Fig. \ref{fig:SG_prediction_comparisons} compares the RMSPEs from the SG-t-SRE model (the true model) to the RMSPEs from the NME version of the SG-t-SRE model, the SG-Gau-SRE model (a model with the correct marginals but the wrong copula), and a Gaussian spatial model in FRK. Both plots show that the true model (LG-t-SRE or SG-t-SRE) achieves lower RMSPEs than the competing models at a majority of BAUs. The differences are clearest when the missing BAUs are MBD. As expected, the differences between the models are less pronounced when the missing BAUs are MAR, since all missing BAUs are close to an observed BAU, making the prediction task simpler. Still, it is notable that the NME versions of the LG-t-SRE and SG-t-SRE models appear to achieve the same predictive performance as the true models. Yet the next set of results show that they fail to achieve nominal coverage in terms of EC, despite their good performance in terms of RMSPE.

\begin{figure}[!ht]
    \centering
    \includegraphics[width=0.8\linewidth]{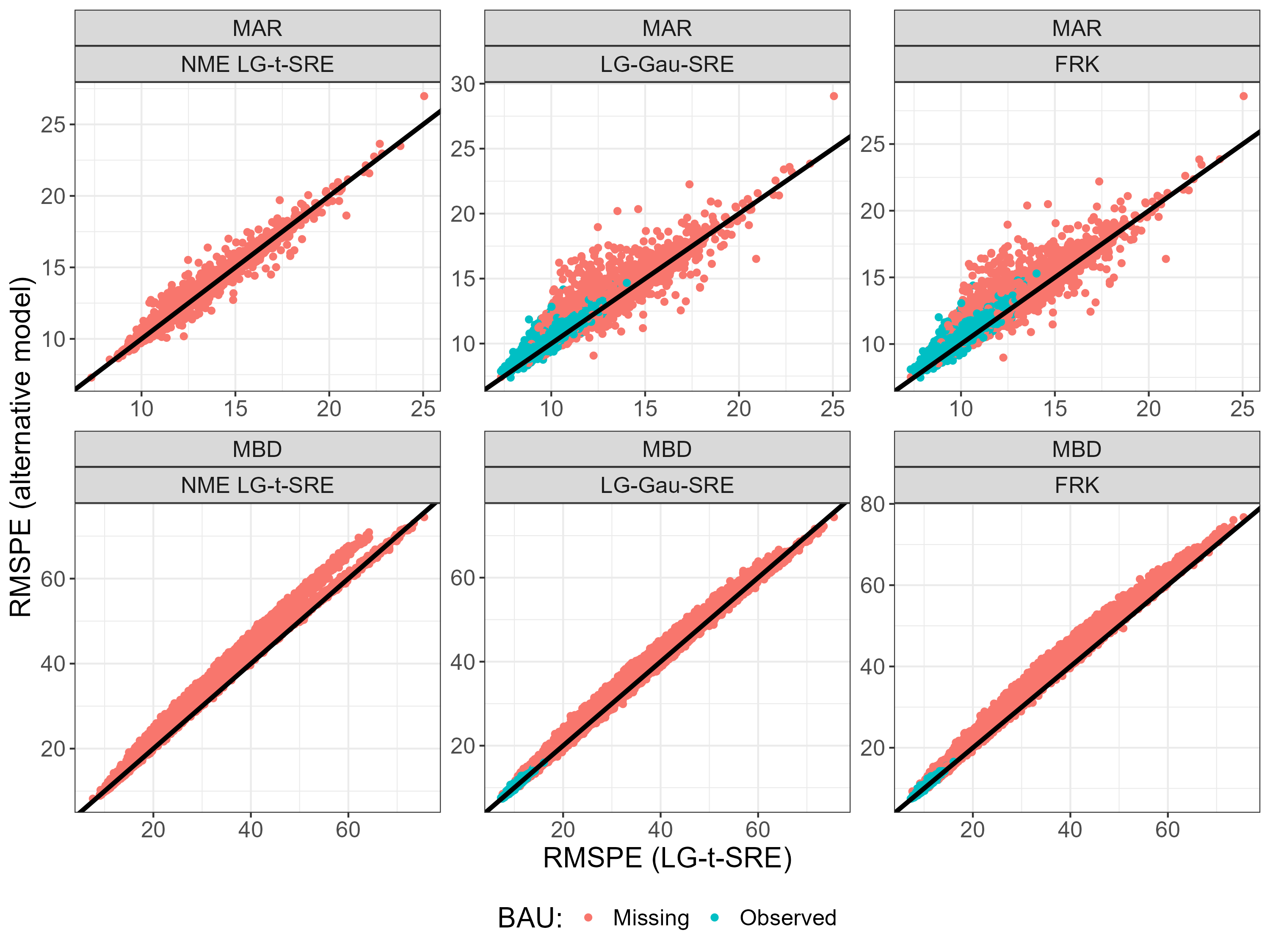}
    \caption{For data generated under the LG-t-SRE model, root-mean-squared prediction errors (RMSPEs) in \eqref{eqn:RMSPE_definition} under the LG-t-SRE model plotted against the RMSPEs under a no-measurement-error (NME) version of the LG-t-SRE model, the LG-Gau-SRE model, and the FRK predictor at each missing and observed basic areal unit (BAU), when the BAUs are missing at random (MAR) or missing by design (MBD). For points above the line, the alternative model has a higher RMSPE than the LG-t-SRE model for that BAU.}
    \label{fig:LG_prediction_comparisons}
\end{figure}

\begin{figure}[!ht]
    \centering
    \includegraphics[width=0.8\linewidth]{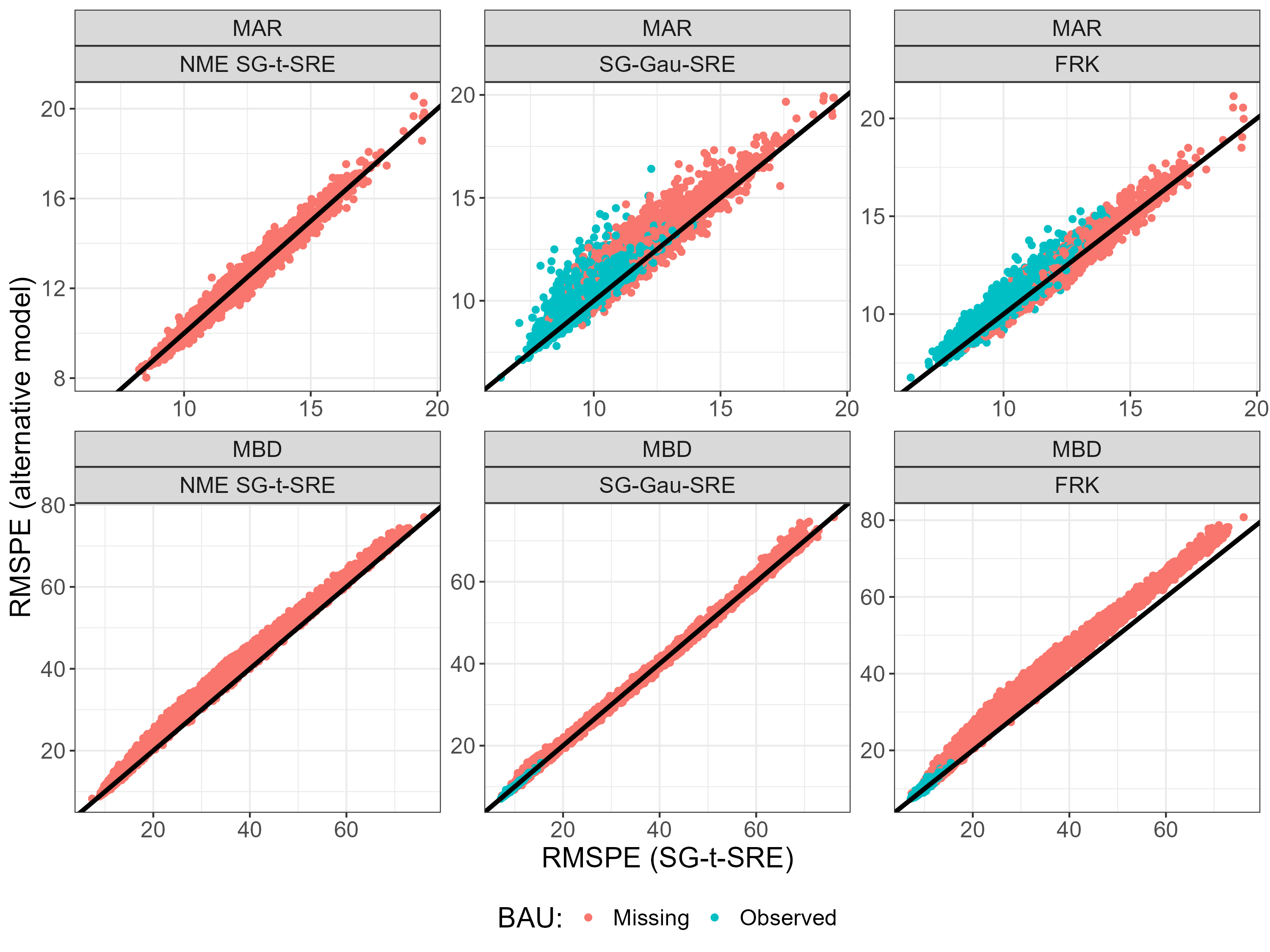}
    \caption{For data generated under the SG-t-SRE model, root-mean-squared prediction errors (RMSPEs) in \eqref{eqn:RMSPE_definition} under the SG-t-SRE model plotted against the RMSPEs under a no-measurement-error version of the SG-t-SRE model, the SG-Gau-SRE model, and the FRK predictor at each missing and observed basic areal unit (BAU), when the BAUs are missing at random (MAR) or missing by design (MBD). Points above the line indicate the alternative model has a higher RMSPE than the SG-t-SRE model for that BAU.}
    \label{fig:SG_prediction_comparisons}
\end{figure}

Fig. \ref{fig:empirical_coverages} shows plots of $\{\EC^{(0.10)}(A_j): j = 1, ..., 10000\}$ for the models fitted to data simulated from the LG-t-SRE and SG-t-SRE models. In Fig. \ref{fig:empirical_coverages}, with nominal coverage of 90\%, the results show that the NME versions of the true models do not have valid 90\% prediction intervals, showing substantial overcoverage in all configurations of the simulation study. This is caused by a failure to separate out variability due to measurement error from the variability of the latent process. 

\begin{figure}
    \centering
    \includegraphics[width=\linewidth]{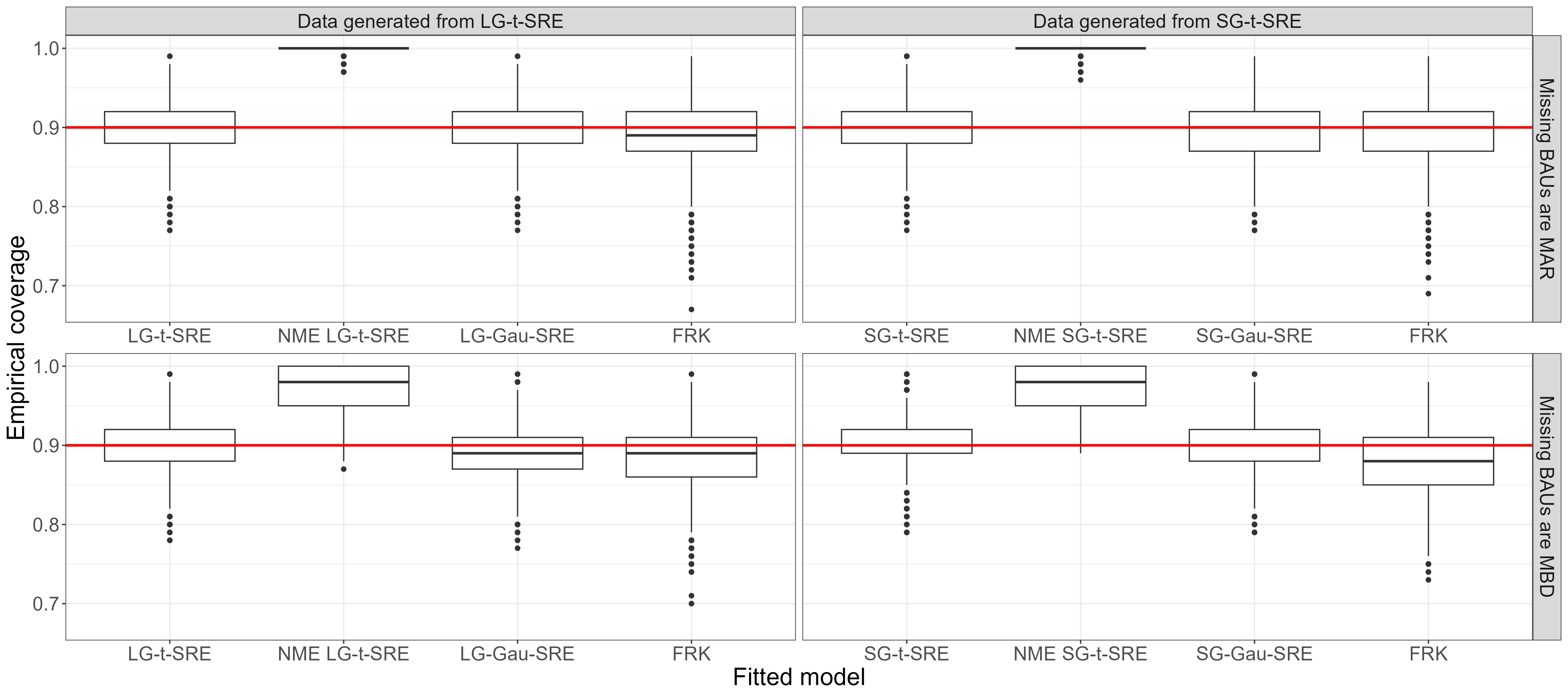}
    \caption{Empirical coverages in \eqref{eqn:EC_definition} calculated for all BAUs in $D$ for models fitted to data simulated from either the LG-t-SRE model or the SG-t-SRE model. The panel titles indicate which copula-based model was used to simulate the data and whether the missing BAUs were `Missing at Random' (MAR) or `Missing by Design' (MBD). The horizontal red line indicates the nominal coverage of 90\%.}\label{fig:empirical_coverages}
\end{figure}

Overall, the results here establish the following: First, there are advantages in using a copula-based SRE model to predict non-Gaussian spatial processes compared to simply using FRK (say). Second, Gaussian spatial models, log-Gaussian spatial models, and the Gau-SRE copula model with the correct marginal distributions (but wrong copula) exhibit poor predictive performance when predicting a spatial process generated by the t-SRE copula. Finally, accounting for measurement errors in spatial data is important. Spatial-copula models with no measurement errors exhibit poor prediction-interval coverage when measurement errors are known to be present in the spatial data.

In Supplement \ref{sec:appendix_additional_predictions_t}--\ref{sec:appendix_timing}, we show additional results. In particular, Supplement \ref{sec:appendix_additional_predictions_t} shows how the marginal predictive distributions of the LG-t-SRE/SG-t-SRE models are different from those of the benchmark models. Supplement \ref{sec:appendix_additional_predictions} presents versions of Figs. \ref{fig:LG_prediction_comparisons}--\ref{fig:empirical_coverages} where the `true models' are the LG-Gau-SRE and SG-Gau-SRE models, respectively. The LG-Gau-SRE and SG-Gau-SRE models are shown to outperform NME versions of these models as well as FRK in terms of RMSPEs and/or ECs. Supplement \ref{sec:appendix_posterior_parameters} shows the parameters of the LG-Gau-SRE, LG-t-SRE, SG-Gau-SRE, and SG-t-SRE models can be accurately recovered via our computationally efficient MCMC algorithm. Supplement \ref{sec:appendix_timing} presents the time taken to run 45,000 MCMC iterations for the LG-Gau-SRE and SG-Gau-SRE models.

\section{Monitoring atmospheric methane in the Bowen Basin, Queensland, Australia}\label{sec:methane}

Satellite-based monitoring of atmospheric column-averaged methane (XCH$_4$) concentrations allows us to track regional trends and local variations in this potent greenhouse gas. Using data from the TROPOMI instrument on the Sentinel 5P satellite \citep{Sentinel5P},  \citet{Sadavarte2021} recently showed that some of the approximately 40 coal mines in the Bowen Basin, Queensland, Australia, were emitting notable quantities of methane gas. Here, we use Sentinel 5P data \citep{Sentinel5P} for August, 2020, and a Bayesian hierarchical spatial-copula model to predict XCH$_4$ for each day of the month over the Bowen Basin. Section \ref{sec:data_description} briefly describes the data. Section \ref{sec:daily_maps} presents maps of predicted XCH$_4$ and the prediction standard deviations (PSDs). The predictions and PSDs are compared to those from FRK. Extra details are given in Supplement \ref{sec:appendix_real_data}.

\subsection{Methane data for August, 2020}\label{sec:data_description}

The data used were \citet{GESDISC}. We extracted the data for the study area, which is the region between $26^{\circ}$S and $20.5^{\circ}$S degrees of latitude and between $147^{\circ}$E and $151^{\circ}$E degrees of longitude or the coastline. The BAUs were constructed as squares with sides of 0.05 decimal degree (dd) length (Supplement \ref{sec:appendix_methane_baus}). Retrievals from August, 2020 were analyzed, with a total of 41,750 XCH$_4$ retrievals over the study area in the state of Queensland, Australia. These retrievals occurred across 24 of the 31 days in August. Histograms of the daily retrievals and the number of retrievals in each day are shown in Supplement \ref{sec:appendix_daily_histograms}-\ref{sec:appendix_daily_counts}. Every XCH$_4$ retrieval is accompanied by a measurement-error quantification from the TROPOMI instrument \citep[Supplement \ref{sec:appendix_measurement_error};][]{Hu2016}.  

\subsection{Model for mapping methane}\label{sec:daily_maps}

Based on the fact that the histograms of XCH$_4$ changed in shape by day, it was decided that spatial prediction of XCH$_4$ should be based on daily data. In what follows, we produce maps of XCH$_4$ for 23 and 28 August 2020 (Days 23 and 28, hereafter), due to the abundance and geographical coverage of the retrievals on these days (2,110 and 2,537 retrievals, respectively). The histograms of the XCH$_4$ retrievals on each day (Fig. \ref{fig:maps_and_histograms}) also appear to show opposite skewness (positive on Day 23 and negative on Day 28). In order to capture the time-varying skewness of the marginal distributions, the SG-Gau-SRE model (Section \ref{sec:models}) was chosen to model XCH$_4$ in and around the Bowen Basin, with two departures in this application. First, the SG-Gau-SRE model for XCH$_4$ uses a spherical covariance function for its covariance matrix $\E(\theta_s, \theta_r)$. Its $(i,j)$ element is $E_{ij}(\theta_s, \theta_r)= \mathbb{I}(d_{i,j} < \theta_r) \times \theta_s \times \{1 - 1.5 (d_{ij}/\theta_r) + 0.5 (d_{ij}/\theta_r)^3\}$, where $d_{ij}$ is the great-circle distance (in 100s of km) between the centers of the $i$-th and $j$-th bisquare spatial basis functions in the study area (Supplement \ref{sec:appendix_methane_basis}), for $i, j = 1, ..., 101$, and $\theta_r$ is a range parameter in the same units. Second, the data model has been modified to accommodate different measurement-error standard deviations for the observations, denoted as $\{\sigma_\TO(A_{\TO k}): k = 1, ..., K\}$. Approximately 30 BAUs on Days 23 and 28 contained two retrievals, and Supplement \ref{sec:appendix_averaging_data} shows how we averaged the retrievals and dealt with the measurement-error standard deviations. Supplement \ref{sec:appendix_skew_Gaussian} and \ref{sec:appendix_FRK_methane} provides details of the SG-Gau-SRE model and the competing FRK model, respectively.

\begin{figure}[!ht]
\centering
\includegraphics[width=0.9\textwidth]{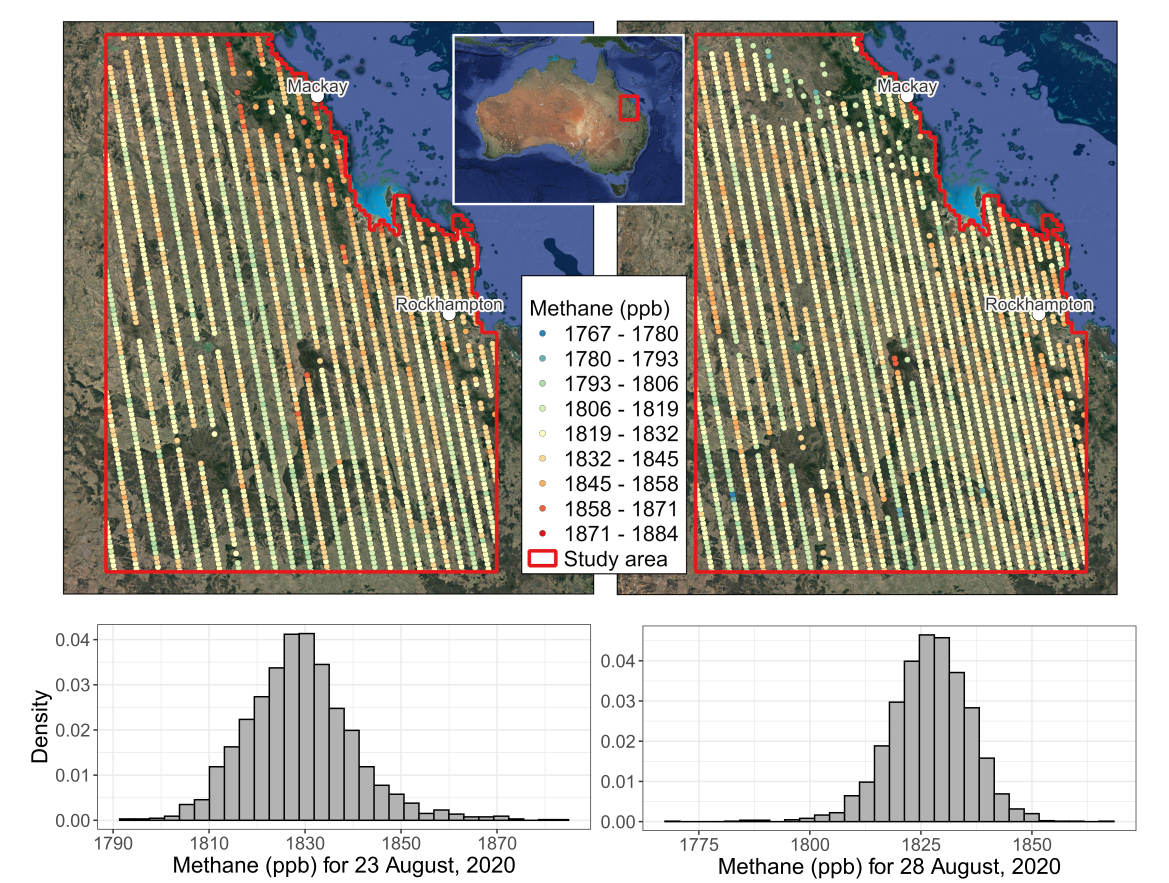}
\caption{Top: Maps of the study area $D$ in the Bowen Basin, Queensland, Australia, with TROPOMI/Sentinel 5P methane retrievals in parts per billion (ppb) for 23 August, 2020 (left) and 28 August, 2020 (right). Bottom: Histogram of methane retrievals (ppb) for 23 August, 2020 (left) and 28 August, 2020 (right). Two cities, Mackay and Rockhampton, are labelled and marked as white circles. The inset shows $D$ on a map of Australia.}
\label{fig:maps_and_histograms}
\end{figure}

An MCMC sampler for the SG-Gau-SRE model was run for Days 23 and 28. See Supplement \ref{sec:appendix_methane_priors}--\ref{sec:appendix_methane_initial} for the priors and initial values. After a burn-in of 5,000 iterations, 40,000 MCMC iterations were run, keeping one in every four samples. The runtime was approximately 20 minutes. Trace plots and Gelman-Rubin statistics \citep{Gelman1992} indicated convergence of the MCMC sampler (Supplement \ref{sec:appendix_methane_trace_plots} and \ref{sec:appendix_convergence_diagnostics}). The posterior mean, 95\% credible interval, and effective sample size \citep[ESS;][]{coda} of the model parameters are summarized in Table \ref{tab:methane_parms}. The ESS values were calculated using the 10,000 MCMC samples left after thinning. The posterior means of $\theta_s$ were 0.61 (Day 23) and 0.56 (Day 28). The value of $\theta_s$ can be thought of as the ratio of systematic spatial variation to the micro-scale variation (see Section \ref{sec:identifiability}), so these values of $\theta_s$ in Table \ref{tab:methane_parms} indicate that a substantial amount of micro-scale variation is present in the process. The posterior mean of $\beta_0$ indicates that the average value of XCH$_4$ in the study area was $1833$ ppb on Day 23 and $1823$ ppb on Day 28. As expected, the skewness parameter $\lambda$ was positive on Day 23 ($1.38$) and negative on Day 28 ($-1.79$), with their 95\% credible intervals not containing zero. The range parameter $\theta_r$, which controls correlations between spatial basis functions, had a posterior mean of $42.0$km on Day 23 and $40.1$km on Day 28. The minimum distance between any two basis-function centers was $49.4$km, so the posterior means of $\theta_r$ suggest `nearest-neighbor' spatial dependence between the basis functions. 

\begin{table}[!ht]
    \centering
    \caption{Posterior means, 95\% credible intervals, and effective sample sizes (ESS) from the R package `coda' \citep{coda}, for the parameters in the SG-Gau-SRE for the methane dataset. The ESS values were calculated from 10,000 thinned MCMC samples.\vspace{3pt}}\label{tab:methane_parms}
\begin{adjustbox}{width=\columnwidth,center}
    \begin{tabular}{|c|ccc|ccc|}
        \hline
	& \multicolumn{3}{c|}{23 August, 2020} & \multicolumn{3}{c|}{28 August, 2020}\\
         Parameter & Posterior mean & 95\% credible interval & ESS & Posterior mean & 95\% credible interval & ESS\\
         \hline
         $\sigma_{\TP}$ & $13.242$ & (11.522, 15.621) & 788 & $11.670$ & (10.066, 14.143) & 562\\
         $\lambda$ & $1.375$ & (1.029, 1.736) & 1163 & $-1.788$ & (-2.195, -1.431) & 819\\
         $\beta_0$ & $7.513$ & (7.511, 7.516) & 1192 & $7.508$ & (7.505, 7.510) & 735\\
         $\theta_s$ & $0.607$ & (0.352, 0.999) & 777 & $0.557$ & (0.334, 0.898) & 514 \\
         $\theta_r$ & $0.419$ & (0.158, 0.727) & 884 & $0.401$ & (0.157, 0.669) & 419 \\
         \hline
    \end{tabular}
    \end{adjustbox}
\end{table}

Predictions of XCH$_4$ and the prediction uncertainties from the SG-Gau-SRE model and FRK are mapped in Fig. \ref{fig:XCH4_maps}. On both days, the SG-Gau-SRE model and FRK yielded similar prediction surfaces. Due to the high levels of micro-scale variability present in the spatial process, the posterior means at the observed BAUs (i.e., $E(\Y_\TO \mid \Z_\TO)$) may be quite different from the posterior means at adjacent missing BAUs. This behavior of the SRE model is consistent with the results of \citet{kang2011bayesian}, but the strip-like nature of the Sentinel 5P satellite tracks results in a less smooth prediction surface. Mapping only the posterior means at the missing BAUs gives a clearer picture of the underlying spatial pattern. From Fig. \ref{fig:XCH4_maps}, on Day 23, high concentrations of methane can be seen around Mackay (Fig. \ref{fig:maps_and_histograms}) and the surrounding coastal region. On Day 28, a high concentration of methane was detected near Mackay as well as in the coastal region near Rockhampton. Regions with high methane concentrations are also revealed inland. The prediction uncertainties from the SG-Gau-SRE model and FRK are similar, but some differences can be noted from Fig. \ref{fig:XCH4_maps}. The skewness in the shape of the marginal distributions of the SG-Gau-SRE model also means that, on Day 23, areas with higher predicted XCH$_4$ have higher PSDs than areas with lower predicted XCH$_4$, and \textit{vice versa} on Day 28. Meanwhile, the PSDs from FRK do not appear to vary strongly with the magnitude of the predictions.

\begin{figure}[!ht]
\centering
\includegraphics[width=\textwidth]{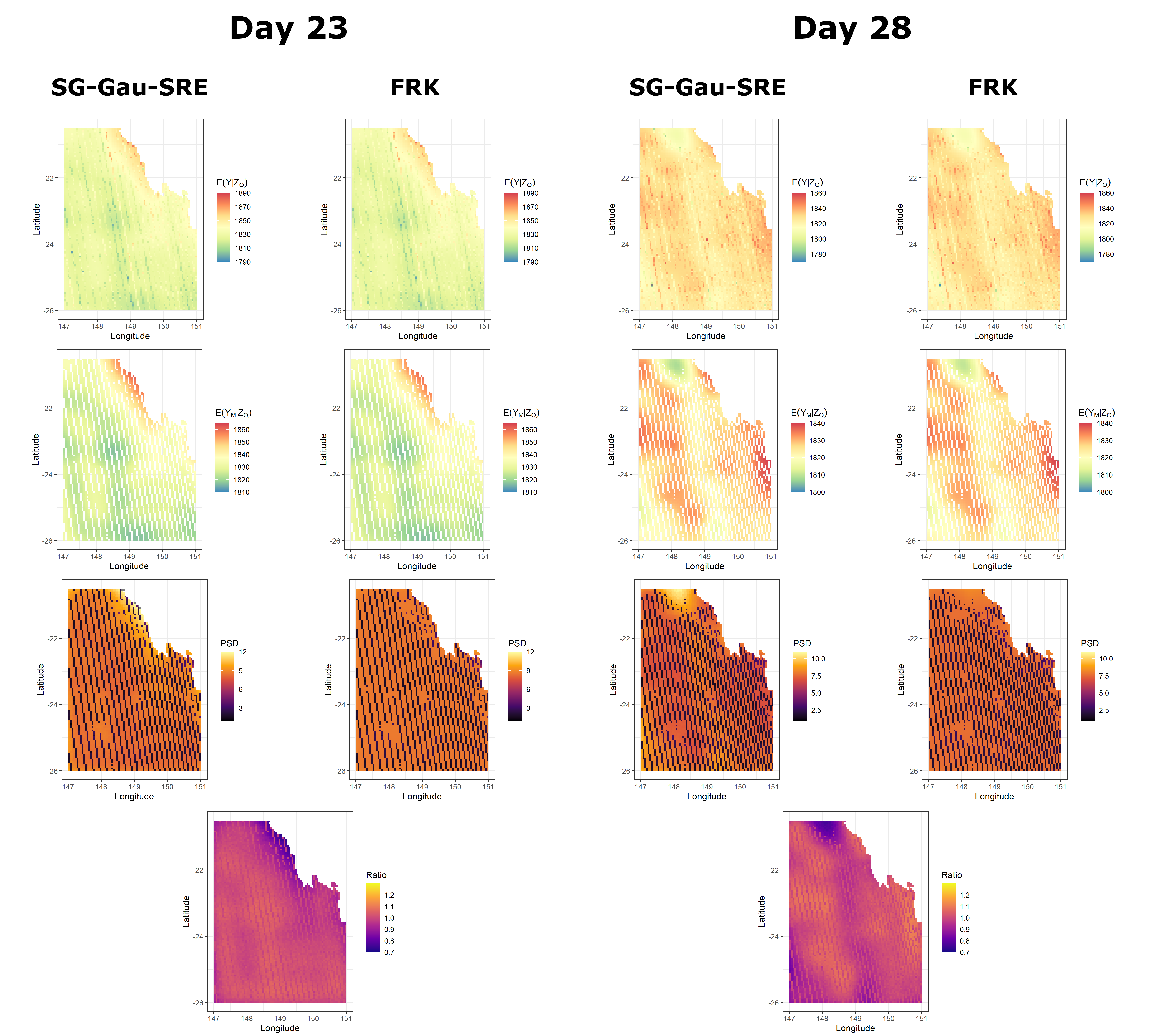}
\caption{Top row: Maps of predicted methane concentrations (ppb) at all basic areal units (BAUs), namely $E(\Y\mid\Z_\TO)$. Second row: Maps of predicted methane concentrations at only the missing BAUs (i.e., $E(\Y_\TM\mid\Z_\TO)$). Third row: Maps of posterior standard deviations (PSDs), where $\mathrm{PSD}(A_j; \Z_\TO) \equiv \sqrt{\mathrm{var}(Y(A_j)\mid \Z_\TO)}$, $j = 1, ..., 7370$. Bottom row: Ratios of the PSD from FRK over the PSD from the SG-Gau-SRE model. }
\label{fig:XCH4_maps}
\end{figure}

\section{Discussion and conclusion}\label{sec:discussion}

In this article, we develop copula-based hierarchical spatial-statistical models for non-Gaussian spatial processes that can be fitted to very large spatial datasets while properly accounting for measurement errors in the spatial data. We show how embedding spatial random effects into Gaussian copulas and t copulas, enables fast Bayesian spatial prediction from large, incomplete, noisy, and non-Gaussian spatial datasets. Simulation studies demonstrate reliable parameter recovery, superior predictive performance, and improved uncertainty quantification, compared to several benchmark spatial-statistical models. To illustrate the practical utility of our approach, we apply it to mapping atmospheric methane over a coal-producing region in Queensland, Australia. Future research could investigate copula-based hierarchical spatio-temporal modeling with spatio-temporal random effects, which would allow smoothing, filtering, and forecasting of non-Gaussian spatio-temporal random fields from large amounts of noisy, incomplete, non-Gaussian spatio-temporal data.

\section*{Acknowledgements}

This article is based on research presented in ARP's PhD dissertation completed at the University of Wollongong, Australia. ARP and NC would like to acknowledge Australian Research Council Discovery Project DP190100180 and NASA ROSES-2023 award 23-OCOST23-0001. NC also acknowledges funding from the Air Force Office of Scientific Research under award number FA2386-23-1-4100. We are grateful to  Associate Professor Andrew Zammit Mangion for insightful discussions and comments. This research was undertaken using resources from the National Computational Infrastructure (NCI Australia), an NCRIS enabled capability supported by the Australian Government.

\newpage
\appendix

\renewcommand{\thepage}{S\arabic{page}}
\renewcommand{\theequation}{S\arabic{equation}}
\renewcommand{\thealgorithm}{S\arabic{algorithm}}
\renewcommand{\thetable}{S\arabic{table}}
\renewcommand{\thefigure}{S\arabic{figure}}
\renewcommand{\thesection}{S\arabic{section}}
\renewcommand{\thesubsection}{S\arabic{section}.\arabic{subsection}}
\setcounter{equation}{0}
\setcounter{figure}{0}
\setcounter{algorithm}{0}
\setcounter{page}{1}
\setcounter{section}{0}
\setcounter{table}{0}

\section*{{\Large Supplementary information}}

This is the Supplementary Information (`Supplement') for ``Bayesian copula-based spatial random effects models for inference with complex spatial data'' by Alan R. Pearse, David Gunawan, and Noel Cressie. We have organized the Supplement as follows: Section 1, equation (1), Table 1, Figure 1, and Algorithm 1, Proposition 1, etc., refer to the main article, while Section S1, equation (S1), Table S1, Figure S1, and Algorithm S1, Proposition S1, etc., refer to the contents of this Supplement.

\section{Proofs of propositions}\label{sec:appendix_proofs}

In this section, we prove the propositions stated in the main article. 

\begin{proof}[Proof of Proposition \ref{prop:Gau_cop_with_SRE}]
Let $\mathbf{W} \equiv (W(A_1), ..., W(A_N))^\top$ follow an $N$-variate Gaussian distribution with mean vector $\mathbf{0}$ and covariance matrix $\bm\Sigma^{\mathrm{SRE}}$ given in Section \ref{sec:Gaussian_copula_SRE} in the main text. Its cumulative distribution function (CDF) is $F_{\mathcal{G},1:N}(w_1, ..., w_N; \bm\Sigma^{\mathrm{SRE}})$, for $w_1, ..., w_N \in \mathbb{R}$. Marginally, for $j=1,...,N$, each $W(A_j)$ follows a mean-zero univariate Gaussian distribution with variance $(\sigma_j^{\mathrm{SRE}})^2$, which is the $j$-th diagonal element of $\bm\Sigma^{\mathrm{SRE}}$. Hence, we can write the CDF as $\Phi(w_j/\sigma_j^{\mathrm{SRE}}),~w_j \in \mathbb{R}$, where $\Phi$ is the CDF of a standard Gaussian distribution. Also define $U_j \equiv \Phi(W(A_j)/\sigma^{\mathrm{SRE}}_j)$, where it follows that $U_j$ has a uniform distribution on $[0, 1]$. Conversely, $W(A_j) = \sigma_{j}^{\mathrm{SRE}} \Phi^{-1}(U_j)$. In obvious notation for specific realizations of these random variables, we also write $u_j = \Phi(w_j/\sigma^\SRE_j)$ and $w_j = \sigma^\SRE_j\Phi^{-1}(u_j)$ for $u_j \in [0, 1]$ and $w_j \in \mathbb{R}$. Then the Gau-SRE copula is,
$$
\mathfrak{G}_{1:N}^{\mathrm{SRE}}(u_1, ..., u_N; \bm\Sigma^{\mathrm{SRE}}) = F_{\mathcal{G},1:N}(\sigma_{1}^{\mathrm{SRE}} \Phi^{-1}(u_1),...,\sigma_{N}^{\mathrm{SRE}} \Phi^{-1}(u_N); \bm{\Sigma}^{\mathrm{SRE}}),
$$
as required. For the copula density, take $\partial \mathfrak{G}_{1:N}^{\mathrm{SRE}}(u_1, ..., u_N; \bm\Sigma^{\mathrm{SRE}})/\partial u_1 \dots \partial u_N$. By the chain rule, 
\begin{align*}
&\frac{\partial \mathfrak{G}_{1:N}^{\mathrm{SRE}}(u_1, ..., u_N; \bm\Sigma^{\mathrm{SRE}})}{\partial u_1 \dots \partial u_N}\\
&~~~~~~~~= \left(\prod_{j=1}^N \frac{\partial (\sigma_j^{\mathrm{SRE}}\Phi^{-1}(u_j))}{\partial u_j}\right) \times f_{\mathcal{G},1:N}(\sigma_{1}^{\mathrm{SRE}} \Phi^{-1}(u_1),...,\sigma_{N}^{\mathrm{SRE}} \Phi^{-1}(u_N); \bm{\Sigma}^{\mathrm{SRE}}),
\end{align*}
where $\partial (\sigma_j^{\mathrm{SRE}}\Phi^{-1}(u_j))/\partial u_j = \sigma_j^{\mathrm{SRE}}/ \phi(\Phi^{-1}(u_j))$, for $j = 1, ..., N$. This is precisely \eqref{eqn:full_SRE_gau_copula_density}.
\end{proof}

\begin{proof}[Proof of Proposition \ref{prop:Gau_anamorphosis}]
By construction, $\W \mid \bm\theta_\TP$ follows an $N$-variate Gaussian distribution with mean vector $\mathbf{0}$ and covariance matrix $\bm\Sigma^\SRE$. Now, let $w_j \in \mathbb{R}$ be realizations of $W(A_j)$, $j = 1, ..., N$, and let $\w \equiv (w_1, ..., w_N)^\top \in \mathbb{R}^N$ be a realization of $\W$. Then, the probability density function (PDF) of $\W \mid \bm\theta_\TP$ can be written as follows: For $w_1, ..., w_N \in \mathbb{R}$ and $\w \in \mathbb{R}^N$,
$$
[\W \mid \bm\theta_\TP] = f_{\mathcal{G},1:N}\left(w_1, ..., w_N; \bm\Sigma^\SRE\right) = \frac{\exp\{-0.5\w^\top(\bm\Sigma^\SRE)^{-1}\w\}}{(2\pi)^{N/2}\det(\bm\Sigma^\SRE)^{1/2}}.
$$
Consider the transformation of variables $W(A_j) = \sigma_j^\SRE\Phi^{-1}(F_{j}(Y(A_j)))$ and the corresponding realizations $w_j = \sigma^\SRE_j\Phi^{-1}(F_j(y_j))$, $j = 1, ..., N$. After transformation of random variables, it follows that
$$
[\Y \mid \bm\theta_\TP] = \left|\prod_{j=1}^N \frac{\p w_j}{\p y_j}\right| \times f_{\mathcal{G},1:N}\!\left(\sigma_{1}^{\mathrm{SRE}} \Phi^{-1}(F_1(y_1)),...,\sigma_{N}^{\mathrm{SRE}} \Phi^{-1}(F_N(y_N)); \bm{\Sigma}^{\mathrm{SRE}}\right).
$$
The $j$-th term in the absolute value of the Jacobian is,
$$
\frac{\p w_j}{\p y_j} = \frac{\sigma_j^\SRE f_j(y_j)}{\phi(\Phi^{-1}(F_j(y_j)))},~y_j \in \mathbb{R}.
$$
Therefore, we have
\begin{align*}
[\Y \mid \bm\theta_\TP] &= \left(\prod_{j=1}^N\frac{\sigma_j^\SRE f_j(y_j)}{\phi(\Phi^{-1}(F_j(y_j)))}\right) \times f_{\mathcal{G},1:N}\!\left(\sigma_{1}^{\mathrm{SRE}} \Phi^{-1}(F_1(y_1)),...,\sigma_{N}^{\mathrm{SRE}} \Phi^{-1}(F_N(y_N)); \bm{\Sigma}^{\mathrm{SRE}}\right)\\
&= \left(\prod_{j=1}^N f_j(y_j)\right) \times \frac{f_{\mathcal{G},1:N}\!\left(\sigma_{1}^{\mathrm{SRE}} \Phi^{-1}(F_1(y_1)),...,\sigma_{N}^{\mathrm{SRE}} \Phi^{-1}(F_N(y_N)); \bm{\Sigma}^{\mathrm{SRE}}\right)}{\prod_{j=1}^N (\sigma_j^\SRE)^{-1}\phi(\Phi^{-1}(F_j(y_j)))}\\
&= \left(\prod_{j=1}^N f_j(y_j)\right) \times \GCD^\SRE_{1:N}\!\left(F_1(y_1), ..., F_N(y_N); \bm\Sigma^\SRE\right),
\end{align*}
and the proposition follows.
\end{proof}

\begin{proof}[Proof of Proposition \ref{prop:t_copula_SRE}]
Let $\mathbf{V} \equiv (V(A_1), ..., V(A_N))^\top$ follow an $N$-variate t distribution on $\nu > 2$ degrees of freedom with mean vector $\mathbf{0}$ and positive-definite scale matrix $\bm\Sigma^{\mathrm{SRE}}$. The CDF is $F_{t,1:N}(v_1, ..., v_N; \bm\Sigma^{\mathrm{SRE}}, \nu)$. Marginally, each $V(A_j)$ follows a mean-zero univariate t distribution on $\nu > 2$ degrees of freedom with scale parameter $(\sigma_j^{\mathrm{SRE}})^2$. Its CDF can be written as $T_{\nu}(v_j/\sigma^{\mathrm{SRE}}_j),~v_j \in \mathbb{R}$. For $j = 1, ..., N$, define $U_j \equiv T_\nu(V(A_j)/\sigma^{\mathrm{SRE}}_j)$, where $U_j$ is uniform on $[0, 1]$ and $V(A_j)/\sigma^{\mathrm{SRE}}_j$ follows a standardized t distribution on $\nu$ degrees of freedom; conversely, $V(A_j) = \sigma_{j}^{\mathrm{SRE}} T_\nu^{-1}(U_j)$. In obvious notation, we write the same transformations for realizations $u_j \in [0, 1]$ and $v_j \in \mathbb{R}$, as $u_j = T_\nu(v_j/\sigma^\SRE_j)$ and $v_j = \sigma^\SRE_j T_\nu^{-1}(u_j)$. Then the t copula with SRE covariance-matrix parameter can be written as follows: For $u_1, ..., u_N \in [0, 1]$,
$$
\mathfrak{T}_{1:N}^{\mathrm{SRE}}\!\left(u_1, ..., u_N; \bm\Sigma^{\mathrm{SRE}}, \nu\right) = F_{t,1:N}\!\left(\sigma_{1}^{\mathrm{SRE}} T_\nu^{-1}(u_1),...,\sigma_{N}^{\mathrm{SRE}} T_\nu^{-1}(u_N); \bm{\Sigma}^{\mathrm{SRE}}, \nu\right),
$$
which is \eqref{eqn:appendix_SRE_t_copula}. For the copula density, simply take partial derivatives with respect to all $u_1, ..., u_N$ and use the chain rule to arrive at \eqref{eqn:appendix_SRE_t_copula_density}. 
\end{proof}

\begin{proof}[Proof of Proposition \ref{prop:t_anamorphosis}]
By construction, $\V \mid \bm\theta_\TP$ follows an $N$-variate t distribution on $\nu$ degrees of freedom with mean vector $\mathbf{0}$ and positive-definite scale matrix $\bm\Sigma^\SRE$. Therefore, its PDF can be written as follows: For $\v \equiv (v_1, ..., v_N)^\top \in \mathbb{R}^N$,
$$
[\V \mid \bm\theta_\TP] = f_{t,1:N}\left(v_1, ..., v_N; \bm\Sigma^\SRE, \nu\right) = \frac{\Gamma\left(\frac{\nu+N}{2}\right)|\bm\Sigma^\SRE|^{-\frac{1}{2}}}{\Gamma\left(\frac{\nu}{2}\right)\left(\nu\pi\right)^{\frac{N}{2}}}\left(1+\frac{1}{\nu}\mathbf{v}^{\top}\bm\left(\bm\Sigma^\SRE\right)^{-1}\mathbf{v}\right)^{\left(-\frac{\nu+N}{2}\right)}.
$$
For $j = 1, ..., N$, consider the transformation $V(A_j) = \sigma_j^\SRE T_\nu^{-1}(F_{j}(Y(A_j)))$ for the marginal random variables and $v_j = \sigma_j^\SRE T_\nu^{-1}(F_j(y_j))$ for their realizations. After transformation, the density is given by,
$$
[\Y \mid \bm\theta_\TP] = \left|\prod_{j=1}^N \frac{\p v_j}{\p y_j}\right| \times f_{t,1:N}(\sigma_{1}^{\mathrm{SRE}} T_\nu^{-1}(F_1(y_1)),...,\sigma_{N}^{\mathrm{SRE}} T_\nu^{-1}(F_N(y_N)); \bm{\Sigma}^{\mathrm{SRE}}, \nu).
$$
The $j$-th term in the absolute value of the Jacobian is,
$$
\frac{\p v_j}{\p y_j} = \frac{\sigma_j^\SRE f_j(y_j)}{t_\nu(T_\nu^{-1}(F_j(y_j)))},~y_j \in \mathbb{R},
$$
where recall that $t_\nu(\cdot)$ is the PDF of a standardized t distribution on $\nu > 2$ degrees of freedom. Therefore, 
\begin{align*}
[\Y \mid \bm\theta_\TP] &= \left(\prod_{j=1}^N\frac{\sigma_j^\SRE f_j(y_j)}{t_\nu(T_\nu^{-1}(F_j(y_j)))}\right) \times f_{t,1:N}(\sigma_{1}^{\mathrm{SRE}} T_\nu^{-1}(F_1(y_1)),...,\sigma_{N}^{\mathrm{SRE}} T_\nu^{-1}(F_N(y_N)); \bm{\Sigma}^{\mathrm{SRE}}, \nu)\\
&= \left(\prod_{j=1}^N f_j(y_j)\right) \times \frac{f_{t,1:N}(\sigma_{1}^{\mathrm{SRE}} T_\nu^{-1}(F_1(y_1)),...,\sigma_{N}^{\mathrm{SRE}} T_\nu^{-1}(F_N(y_N)); \bm{\Sigma}^{\mathrm{SRE}}, \nu)}{\prod_{j=1}^N (\sigma_j^\SRE)^{-1}t_\nu(T_\nu^{-1}(F_j(y_j)))}\\
&= \left(\prod_{j=1}^N f_j(y_j)\right) \times \TCD^\SRE_{1:N}(F_1(y_1), ..., F_N(y_N); \bm\Sigma^\SRE, \nu),
\end{align*}
as required.
\end{proof}

\section{Parameterization of $\bm\Sigma^\SRE$ with multiple resolutions of basis functions}\label{sec:appendix_multiresolution}

Recall from Section \ref{sec:identifiability} that $\S$ is an $N\times b$ matrix of spatial basis functions, and $\etab \sim \MVG(\mathbf{0}, \E)$ is a $b$-variate random vector of basis-function coefficients with mean vector $\mathbf{0}$ and $b\times b$ covariance matrix $\E$. Also recall that $\xib \sim \MVG(\mathbf{0}, \I_N)$, where $\I_N$ is the $N$-dimensional identity matrix, is an $N$-variate standard Gaussian random vector that is independent of $\etab$. Then we have $\W \equiv \S\etab + \xib$ per \eqref{eqn:Gaussian_SRE_model}, where the unconditional distribution of $\W$ is an $N$-variate Gaussian distribution with mean vector $\mathbf{0}$ and covariance matrix $\bm\Sigma^\SRE \equiv \S\E\S^\top + \I_N$. In this section, we show how to validly parameterize $\bm\Sigma^\SRE$ when $\S$ contains multiple resolutions of spatial basis functions. This is a common strategy to capture spatial dependence at various spatial length scales. 

Let $P \geq 2$, and define $\{\S_p \in \mathbb{R}^{N \times b_p}: p = 1 ,..., P\}$ to be a set of spatial-basis-function matrices for $P$ different resolutions of basis functions, with a total of $b = \sum_{p=1}^P b_p$ basis functions. In \eqref{eqn:Gaussian_SRE_model}, the matrix $\S \in \mathbb{R}^{N\times b}$ can be partitioned as $\S \equiv (\S_1, \cdots,\S_P)$. The corresponding $b$-variate vector of random effects can also be decomposed as $\etab \equiv (\etab_1^\top, ..., \etab_P^\top)^\top$, where $\etab_p$ is a $b_p$-variate Gaussian random vector with mean vector $\mathbf{0}$ and positive-definite covariance matrix $\E_p$ for $p=1,...,P$. Parameterize the covariance matrix as $\E_p \equiv \theta_{s,p}\times \R_p$, where $\R_p$ is an $b_p\times b_p$ correlation matrix. Finally, assume $\{\etab_p: p=1,...,P\}$ are mutually independent. Therefore, after defining $\etab \equiv (\etab_1^\top , ..., \etab_P^\top)^\top$, we have $\etab \sim \MVG(\mathbf{0}, \mathrm{diag}\{\E_1, ..., \E_P\})$ or, equivalently, $\etab \sim \MVG(\mathbf{0}, \mathrm{diag}\{\theta_{s,1}\R_1, \cdots, \theta_{s,P}\R_P\})$. Each of the parameters $\{\theta_{s,p}: p = 1, ..., P\}$ represents the relative importance of $\{\S_p\etab_p: p = 1, ..., P\}$ relative to $\xib$. The parameters $\{\theta_{s,p}\}$ and any parameters of the correlation matrices $\{\R_p: p = 1, ..., P\}$ remain identifiable in $\bm\Sigma^\SRE = \S \E \S^\top + \I_N$ because they do not merely scale the marginal random variables and the covariance matrix; see Section \ref{sec:identifiability} of the main article for further details on parameter identification.   

\section{Additional details for Markov chain Monte Carlo (MCMC) implementation}\label{sec:appendix_mcmc}

In this section, detailed derivations of the full-conditional densities in the Markov chain Monte Carlo algorithms from Sections \ref{sec:MCMC_Gaussian_case}--\ref{sec:MCMC_t_case} are presented. In those derivations, recall that the notation ``$\rest$'' refers to conditioning on all the remaining variables. Section \ref{sec:appendix_FCs_Gau} relates to the Gau-SRE copula model. Section \ref{sec:appendix_FCs_t} relates to the t-SRE copula model. 

\subsection{Derivations of the full-conditional densities for the Gau-SRE copula model}\label{sec:appendix_FCs_Gau}

In this section, we derive the full-conditional distributions needed for MCMC with the Gau-SRE copula model (Section \ref{sec:Gaussian_copula_SRE} and Section \ref{sec:MCMC_Gaussian_case}). We also give information for sampling from these full-conditional distributions, if they are not available in closed form. 

\subsubsection{Derivation of $[\Y_\TO \mid \rest]$ and sampling from it}

From \eqref{eqn:joint_posterior_GauSRE}, we see that the joint posterior of $\Y_\TO$, $\etab$, and $\bm\theta_\TP$ is given by,
\begin{equation}
    [\Y_\TO, \etab, \bm\theta_\TP\mid \Z_\TO] \propto [\Z_\TO \mid \Y_\TO, \bm\theta_\TO] \times [\Y_\TO \mid \etab, \bm\theta_\TP] \times [\etab \mid \bm\theta_\TP] \times [\bm\theta_\TP].\label{eqn:appendix_Gau_SRE_joint_posterior_marginal}
\end{equation}
To obtain the full-conditional distribution for $\Y_\TO$, we write $[\Y_\TO \mid \rest] \propto [\Z_\TO \mid \Y_\TO, \bm\theta_\TO] \times [\Y_\TO \mid \etab, \bm\theta_\TP] \times [\etab \mid \bm\theta_\TP] \times [\bm\theta_\TP]$, which simplifies to $[\Y_\TO \mid \rest] \propto [\Z_\TO \mid \Y_\TO, \bm\theta_\TO] \times [\Y_\TO \mid \etab, \bm\theta_\TP]$, where recall that the data-model parameters $\bm\theta_\TO$ are assumed known. Conditional independence in the data model and the conditional independence of the elements of $\Y_\TO$ given $\etab$ (see \eqref{eqn:obs_gau_conditional}) mean that we can write,
$$
[\Y_\TO \mid \rest] \propto \prod_{k=1}^K [Z(A_{\TO k}) \mid Y(A_{\TO k}), \bm\theta_\TO] \times [Y(A_{\TO k}) \mid \etab, \bm\theta_\TP].
$$
The $k$-th term ($k = 1, ..., K$) in the product (i.e., $[Z(A_{\TO k}) \mid Y(A_{\TO k}), \bm\theta_\TO] \times [Y(A_{\TO k}) \mid \etab, \bm\theta_\TP]$) is in fact equivalent to $[Y(A_{\TO k}) \mid \rest]$, which is the full-conditional density of $Y(A_{\TO k})$. This means the latent values $\{Y(A_{\TO k}): k = 1, ..., K\}$ can be sampled independently and in parallel in $K$ separate Metropolis-within-Gibbs (MwG) steps with random-walk proposals. Each step only involves evaluating two univariate conditional densities, which makes the computations very fast.

\subsubsection{Derivation of $[\etab \mid \rest]$ and sampling from it}

Recall from Section \ref{sec:Gaussian_copula_SRE} that $\etab \sim \MVG(\mathbf{0}, \E)$ is a random vector of spatial-basis function coefficients from \eqref{eqn:Gaussian_SRE_model}, where $\E$ is a $b\times b$ covariance matrix. Also recall that $\S_\TO$ is a $K \times b$ dimensional matrix of spatial basis functions at the observed BAUs, and let $\w_\TO \in \mathbb{R}^K$ be a realization of $\W_\TO$ (see Section \ref{sec:Gaussian_copula_SRE}). 

From \eqref{eqn:appendix_Gau_SRE_joint_posterior_marginal}, we have $[\etab \mid \rest ] \propto [\Z_\TO \mid \Y_\TO, \bm\theta_\TO] \times [\Y_\TO \mid \etab, \bm\theta_\TP] \times [\etab \mid \bm\theta_\TP] \times [\bm\theta_\TP]$. This simplifies to $[\etab \mid \rest]\propto [\Y_\TO \mid \etab, \bm\theta_\TP] \times [\etab \mid \bm\theta_\TP]$. The following algebra then shows that the full conditional distribution $[\etab \mid \rest]$ has a recognizable form:
\begin{align*}
    &[\etab \mid \rest]\\     
    &\propto \frac{\exp\{-0.5(\w_\TO - \S_\TO\etab)^\top(\w_\TO - \S_\TO\etab)\}}{(2\pi)^{K/2}} \times \frac{\exp\{-0.5\etab^\top\E^{-1}\etab\}}{(2\pi)^{b/2}\det(\E)^{1/2}}\\
    &= \frac{\exp\{-0.5(\w_\TO^\top\w_\TO - 2\etab^\top\S_\TO^\top\w_\TO + \etab^\top(\S_\TO^\top\S_\TO + \E^{-1})\etab)\}}{(2\pi)^{(K+b)/2}\det(\E)^{1/2}}\\
    &\propto \exp\{-0.5(\etab^\top(\S_\TO^\top\S_\TO + \E^{-1})\etab-2\etab^\top\S_\TO^\top\w_\TO)\}\\
    &\propto \exp\{-0.5(\etab^\top(\S_\TO^\top\S_\TO + \E^{-1})\etab-2\etab^\top\S_\TO^\top\w_\TO + \w_\TO^\top\S_\TO^\top(\S_\TO^\top\S_\TO + \E^{-1})^{-1}\S_\TO\w_\TO)\}\\
    &\propto \exp\{-0.5(\etab - (\S_\TO^\top\S_\TO + \E^{-1})^{-1}\S_\TO\w_\TO)^\top(\S_\TO^\top\S_\TO + \E^{-1})(\etab - (\S_\TO^\top\S_\TO + \E^{-1})^{-1}\S_\TO\w_\TO)\}.
\end{align*}
The right-hand side of the final line is proportional to the density of a $b$-variate Gaussian distribution with mean vector, $\bm\mu_{\etab\mid\rest} \equiv (\S_\TO^\top\S_\TO + \E^{-1})^{-1}\S_\TO^\top\w_\TO$, and covariance matrix, $\bm\Sigma_{\etab\mid\rest} \equiv (\S_\TO^\top\S_\TO + \E^{-1})^{-1}$. This completes the derivation. Sampling from this low-dimensional Gaussian distribution is straightforward. 

\subsubsection{Derivation of $[\bm\theta_\TP \mid \rest]$ and sampling from it}

From \eqref{eqn:appendix_Gau_SRE_joint_posterior_marginal}, we have $[\bm\theta_\TP \mid \rest] \propto [\Z_\TO \mid \Y_\TO, \bm\theta_\TO] \times [\Y_\TO \mid \etab, \bm\theta_\TP] \times [\etab \mid \bm\theta_\TP] \times [\bm\theta_\TP]$, which simplifies to,
$$
[\bm\theta_\TP \mid \rest] \propto [\Y_\TO \mid \etab, \bm\theta_\TP]\times [\etab \mid \bm\theta_\TP] \times [\bm\theta_\TP].
$$
However, integrating out the random effects $\etab$ results in better mixing of the Markov chain. Hence, we use the full-conditional density,
\begin{align*}
[\bm\theta_\TP \mid \rest] &\propto \left\{\int_{\etab}[\Y_\TO \mid \etab, \bm\theta_\TP]\times [\etab \mid \bm\theta_\TP]~\d\etab\right\} \times [\bm\theta_\TP]\\
&\propto [\Y_\TO \mid  \bm\theta_\TP] \times [\bm\theta_\TP]\\
&= \left(\prod_{k=1}^Kf_{\TO k}(y_{\TO k})\right)\times \GCD_{1:K}^\SRE(F_{\TO 1}(y_{\TO 1}), ..., F_{\TO K}(y_{\TO K}); \bm\Sigma^\SRE_{\TO\TO}) \times [\bm\theta_\TP].
\end{align*}
The parameters can be sampled using an adaptive random-walk Metropolis-Hastings (MH) algorithm. Any strictly positive elements of $\bm\theta_\TP$ (e.g., $\theta_s$) are log-transformed; the resulting vector of parameters, where some elements have been transformed, is written as $\Tilde{\bm\theta}_\TP$. The proposal distribution for MH sampling of $\Tilde{\bm\theta}_\TP$ is taken to be a multivariate Gaussian distribution. We follow \citet{garthwaite2016adaptive} to adapt a scaling factor for the covariance matrix of the proposal distribution to achieve an acceptance rate of approximately 24\%. The covariance matrix itself is adapted at each iteration using all previous samples.

\subsubsection{Sampling from $[\Y_\TM \mid \Y_\TO, \etab, \bm\theta_\TP]$}\label{sec:predictions_gau_SRE}

As for sampling from the conditional distribution $[\Y_\TM \mid \Y_\TO, \bm\theta_\TP]$, this is simplified by the fact that $[\Y_\TM \mid \Y_\TO, \etab, \bm\theta_\TP]=[\Y_\TM \mid \etab, \bm\theta_\TP]=\prod_{l=1}^L [Y(A_{\TM l}) \mid \etab, \bm\theta_\TP]$ (see  \eqref{eqn:preds_gau_conditional}), since $\Y_\TM$ and $\Y_\TO$ (and all of their elements) are conditionally independent given $\etab$. 

Now, recall that $\Y_\TM \equiv (Y(A_{\TM 1}), ..., Y(A_{\TM L}))^\top$. For $l = 1, ..., L$ and given posterior samples of $\etab$ and $\bm\theta_\TP$, we can simulate $Y(A_{\TM l})$ from $[Y(A_{\TM l}) \mid \etab, \bm\theta_\TP]$ by the following procedure:
\begin{enumerate}
    \item Simulate $W(A_{\TM l}) \sim \Gau(\S(A_{\TM l})^\top \etab, 1)$, and 
    \item Compute $Y(A_{\TM l}) = F_{\TM l}^{-1}\big(\Phi\big(W(A_{\TM l})/\sigma^\SRE_{\TM l}\big)\big)$, where recall that $\sigma^\SRE_{\TM l}$ is the square root of the $l$-th diagonal element of $\bm\Sigma_{\TM\TM}^\SRE \equiv \S_\TM \E \S_\TM^\top + \I_L$. 
\end{enumerate}

\subsection{Derivations of the full-conditional densities for the t-SRE copula model}\label{sec:appendix_FCs_t}

In this section, we derive the full-conditional distributions needed for MCMC with the t-SRE copula model (Section \ref{sec:t_copula_SRE} and Section \ref{sec:MCMC_t_case}). We also give information for sampling from these full-conditional distributions, if they are not available in closed form.

\subsubsection{Derivation of $[\Y_\TO \mid \rest]$ and sampling from it}

From \eqref{eqn:hierarchical_joint_t} in Section \ref{sec:MCMC_t_case}, the joint posterior of $\Y_\TO$, $\etabb$, $\gamma$, and $\bm\theta_\TP$ under the t-SRE copula model is given by,
\begin{equation}
    [\Y_\TO, \etabb, \gamma, \bm\theta_\TP\mid \Z_\TO] \propto [\Z_\TO \mid \Y_\TO, \bm\theta_\TO] \times [\Y_\TO \mid \etabb, \gamma, \bm\theta_\TP] \times [\etabb \mid \gamma, \bm\theta_\TP] \times [\gamma\mid \bm\theta_\TP] \times [\bm\theta_\TP],\label{eqn:appendix_t_SRE_joint_posterior_marginal}
\end{equation}
where $\bm\theta_\TP$ contains $\nu > 2$, the degrees of freedom parameter of the multivariate t distribution. 

To obtain the full-conditional distribution for $\Y_\TO$, we write $[\Y_\TO \mid \rest] \propto [\Z_\TO \mid \Y_\TO, \bm\theta_\TO] \times [\Y_\TO \mid \etabb, \gamma, \bm\theta_\TP] \times [\etabb \mid \gamma, \bm\theta_\TP] \times [\gamma\mid \bm\theta_\TP] \times [\bm\theta_\TP]$, which simplifies to $[\Y_\TO \mid \rest] \propto [\Z_\TO \mid \Y_\TO, \bm\theta_\TO] \times [\Y_\TO \mid \etabb, \gamma, \bm\theta_\TP]$, where recall that the data-model parameters $\bm\theta_\TO$ are assumed known. Conditional independence in the data model and the conditional independence of the elements of $\Y_\TO$ given $\etabb$ and $\gamma$ (see \eqref{eqn:obs_t_conditional}) mean that we can write,
$$
[\Y_\TO \mid \rest] \propto \prod_{k=1}^K [Z(A_{\TO k}) \mid Y(A_{\TO k}), \bm\theta_\TO] \times [Y(A_{\TO k}) \mid \etabb, \gamma, \bm\theta_\TP].
$$
The $k$-th term ($k = 1, ..., K$) in the product (i.e., $[Z(A_{\TO k}) \mid Y(A_{\TO k}), \bm\theta_\TO] \times [Y(A_{\TO k}) \mid \etabb, \gamma, \bm\theta_\TP]$) is in fact the full-conditional density of $Y(A_{\TO k})$.

While the full-conditional density is not available in closed form, sampling is easy and fast using a random-walk Metropolis-Hastings algorithm. The latent values $\{Y(A_{\TO k}): k = 1, ..., K\}$ can be sampled independently and in parallel. 

\subsubsection{Derivation of $[\etabb \mid \rest]$}

Recall from Section \ref{sec:t_copula_SRE} that $\gamma > 0$ is a scale factor that follows a Gamma distribution with shape and rate parameters equal to $\nu/2$, where $\nu > 2$. Also recall that $\etabb \sim \MVG(\mathbf{0}, \gamma^{-1}\E)$ and $\xibb_\TO \sim \MVG(\mathbf{0}, \gamma^{-1}\I_K)$. Let $\v_\TO \in \mathbb{R}^K$ be a realization of $\V_\TO$ (see Section \ref{sec:t_copula_SRE}). 

From \eqref{eqn:appendix_t_SRE_joint_posterior_marginal}, the full-conditional density of $\etabb$ can be written as $[\etabb \mid \rest]\propto [\Z_\TO \mid \Y_\TO, \bm\theta_\TO] \times [\Y_\TO \mid \etabb, \gamma, \bm\theta_\TP] \times [\etabb \mid \gamma, \bm\theta_\TP] \times [\gamma\mid \bm\theta_\TP] \times [\bm\theta_\TP]$, which simplifies to $[\etabb \mid \rest ]\propto [\Y_\TO \mid \etabb, \gamma, \bm\theta_\TP] \times [\etabb \mid \gamma, \bm\theta_\TP]$. Below, we show that this full-conditional density has a recognizable form. That is,
\begin{align*}
    &[\etabb \mid \rest] \\
    &\propto [\Y_\TO \mid \etabb, \gamma, \bm\theta_\TP] \times [\etabb \mid \gamma, \bm\theta_\TP]\\
    &\propto \frac{\exp\{-0.5\gamma(\v_\TO - \S_\TO\etabb)^\top(\v_\TO - \S_\TO\etabb)\}}{(2\pi)^{K/2}\gamma^{-K/2}} \times \frac{\exp\{-0.5\gamma\etabb^\top\E^{-1}\etabb\}}{(2\pi)^{b/2}\gamma^{-b/2}\det(\E)^{1/2}}\\
    &= \frac{\exp\{-0.5\gamma(\v_\TO^\top\v_\TO - 2\etabb^\top\S_\TO^\top\v_\TO + \etabb^\top(\S_\TO^\top\S_\TO + \E^{-1})\etab)\}}{(2\pi)^{(K+b)/2}\gamma^{-(K+b)/2}\det(\E)^{1/2}}\\
    &\propto \exp\{-0.5\gamma(\etabb^\top(\S_\TO^\top\S_\TO + \E^{-1})\etabb-2\etabb^\top\S_\TO^\top\v_\TO)\}\\
    &\propto \exp\{-0.5\gamma(\etabb^\top(\S_\TO^\top\S_\TO + \E^{-1})\etabb-2\etabb^\top\S_\TO^\top\v_\TO + \v_\TO^\top\S_\TO^\top(\S_\TO^\top\S_\TO + \E^{-1})^{-1}\S_\TO\v_\TO)\}\\
    &\propto \exp\{-0.5\gamma(\etabb - (\S_\TO^\top\S_\TO + \E^{-1})^{-1}\S_\TO^\top\v_\TO)^\top(\S_\TO^\top\S_\TO + \E^{-1})(\etabb - (\S_\TO^\top\S_\TO + \E^{-1})^{-1}\S_\TO^\top\v_\TO)\}.
\end{align*}
The right-hand side of the final line is proportional to the density of a $b$-variate Gaussian distribution with mean vector, $\bm\mu_{\etabb\mid\rest} \equiv (\S_\TO^\top\S_\TO + \E^{-1})^{-1}\S_\TO^\top\v_\TO$, and covariance matrix, $\bm\Sigma_{\etabb\mid\rest} \equiv \gamma^{-1}(\S_\TO^\top\S_\TO + \E^{-1})^{-1}$. This completes the derivation of $[\etabb \mid \rest]$. Sampling from this low-dimensional multivariate Gaussian distribution is straightforward.

\subsubsection{Derivation of $[\gamma \mid \rest]$}

From \eqref{eqn:appendix_t_SRE_joint_posterior_marginal}, the full-conditional density of $\gamma$ can be written as $[\gamma \mid \rest] \propto [\Z_\TO \mid \Y_\TO, \bm\theta_\TO] \times [\Y_\TO \mid \etabb, \gamma, \bm\theta_\TP] \times [\etabb \mid \gamma, \bm\theta_\TP] \times [\gamma\mid \bm\theta_\TP] \times [\bm\theta_\TP]$. This can be simplified to, 
$$
[\gamma \mid \rest] \propto [\Y_\TO \mid \etabb, \gamma, \bm\theta_\TP] \times [\etabb \mid \gamma, \bm\theta_\TP] \times [\gamma \mid \bm\theta_\TP].
$$
In the full-conditional distribution, we condition on the parameters $\bm\theta_\TP$. The transformation that links $\Y_\TO$ with $\V_\TO$ (see Proposition \ref{prop:t_anamorphosis}) is $V(A_{\TO k}) \equiv \sigma_{\TO k}^\SRE T_\nu^{-1}(F_{\TO k}(Y(A_{\TO k})))$, for $k = 1, ..., K$, and this transformation is deterministic when $\bm\theta_\TP$ is fixed. Therefore, we can write,
$$
[\gamma \mid \rest] \propto [\V_\TO \mid \etabb, \gamma, \bm\theta_\TP] \times [\etabb \mid \gamma, \bm\theta_\TP] \times [\gamma \mid \bm\theta_\TP].
$$
We can achieve better mixing of the MCMC sampler by integrating out the random effects $\etabb$ in the equation above. Therefore, we use,
$$
[\gamma \mid \rest] \propto [\V_\TO \mid \gamma, \bm\theta_\TP] \times [\gamma \mid \bm\theta_\TP],
$$
where $[\V_\TO \mid \gamma, \bm\theta_\TP] = \int_{\etabb} [\V_\TO, \etabb, \gamma, \bm\theta_\TP]\times[\etabb\mid \gamma, \bm\theta_\TP]~\d\etabb$. From \eqref{eqn:V}, we know that $\V_\TO \mid \gamma, \bm\theta_\TP \sim \MVG(\mathbf{0}, \gamma^{-1}\bm\Sigma_{\TO\TO}^\SRE)$. Therefore, it follows that
\begin{align*}
    [\gamma \mid \rest] &\propto [\V_\TO \mid \gamma, \bm\theta_\TP] \times [\gamma \mid \bm\theta_\TP]\\    &\propto\gamma^{0.5K}\exp\{-0.5\gamma\v_\TO^\top(\bm\Sigma^\SRE_{\TO\TO})^{-1}\v_\TO\}\gamma^{0.5\nu - 1}\exp\{-0.5\gamma\nu\}\\
    &\propto \gamma^{0.5(K + \nu)-1}\exp\{-0.5\gamma(\nu + \v_\TO^\top(\bm\Sigma^\SRE_{\TO\TO})^{-1}\v_\TO)\}.
\end{align*}
The right-hand side of the last line is proportional to the density of a Gamma-distributed random variable with shape parameter $0.5(K + \nu)$ and rate parameter $0.5(\nu + \v_\TO^\top(\bm\Sigma^\SRE_{\TO\TO})^{-1}\v_\TO)$, which completes the derivation. Sampling from this Gamma distribution is straightforward.

\subsubsection{Derivation of $[\bm\theta_\TP \mid \rest]$ and sampling from it}

From \eqref{eqn:appendix_t_SRE_joint_posterior_marginal}, the na\"ive construction of the full-conditional density for $\bm\theta_\TP$ is, $[\bm\theta_\TP \mid \rest] \propto [\Y_\TO \mid \etabb, \gamma, \bm\theta_\TP] \times [\etabb \mid \gamma, \bm\theta_\TP] \times [\gamma \mid \bm\theta_\TP] \times [\bm\theta_\TP]$. As in the Gaussian-copula case, integrating out the random effects $\etabb$, results in better mixing of the Markov chain for the parameters. Additionally, integrating out $\gamma$ further improves the mixing properties. Therefore, for sampling the process-model parameters, we use the following full-conditional density:
\begin{align*}
    [\bm\theta_\TP \mid \rest] &\propto \left\{ \int_\gamma \left\{ \int_{\etabb} [\Y_\TO \mid \etabb, \gamma, \bm\theta_\TP] \times [\etabb \mid \gamma, \bm\theta_\TP] ~\d{\etabb}\right\} \times [\gamma \mid \bm\theta_\TP]~\d\gamma\right\} \times [\bm\theta_\TP]\\
    &\propto [\Y_\TO \mid \bm\theta_\TP] \times [\bm\theta_\TP],
\end{align*}
where $[\Y_\TO \mid \bm\theta_\TP]$ is given by \eqref{eqn:tanamorphosis}.

The parameters can be sampled using an adaptive random-walk Metropolis-Hastings algorithm. Any strictly positive elements of $\bm\theta_\TP$ (e.g., $\theta_s$) are log-transformed; the resulting vector of parameters, where some elements have been transformed, is written as $\Tilde{\bm\theta}_\TP$. The proposal distribution for $\Tilde{\bm\theta}_\TP$ is taken to be a multivariate Gaussian distribution. We follow \citet{garthwaite2016adaptive} to adapt the scale factor of the proposal distribution's covariance matrix to achieve an acceptance rate of approximately 24\%. The covariance matrix itself is adapted using at each iteration all previous samples.

\subsubsection{Sampling from $[\Y_\TM \mid \Y_\TO, \etabb, \gamma, \bm\theta_\TP]$}\label{sec:predictions_t_SRE}

Recognizing that $[\Y_\TM \mid \Y_\TO, \etabb, \gamma, \bm\theta_\TP] = [\Y_\TM \mid \etabb, \gamma, \bm\theta_\TP]$ and using \eqref{eqn:preds_t_conditional}, it is straightforward to sample from $[\Y_\TM \mid \Y_\TO, \etabb, \gamma, \bm\theta_\TP]$: For $l = 1, ..., L$ and given posterior samples of $\etabb$, $\gamma$, and $\bm\theta_\TP$, 
\begin{enumerate}
    \item Generate $V(A_{\TM l})$ from $\Gau(\S(A_{\TM l})^\top \etabb, \gamma^{-1})$, and 
    \item Compute $Y(A_{\TM l}) = F_{\TM l}^{-1}(T_\nu(V(A_{\TM l})/ \sigma^\SRE_{\TM l}))$.
\end{enumerate}

\section{Additional details for the simulation study}\label{sec:appendix_sims}

This section contains additional details for the simulation study reported in Section \ref{sec:simulation_studies} of the main article, including figures, tables, and results. Sections \ref{sec:appendix_LGGau}--\ref{sec:appendix_FRK_SG} fully describe the models used in the simulation study. Section \ref{sec:appendix_setup} describes the set up of the simulation study. Section \ref{sec:appendix_other_models_sims} describes how data were simulated from the Gau-SRE copula models (i.e., the LG-Gau-SRE and SG-Gau-SRE models). Section \ref{sec:appendix_sim_priors} presents the weakly informative priors used for the MCMC, and Section \ref{sec:appendix_sim_initial_values} gives the initial values used for the MCMC algorithm. Section \ref{sec:appendix_additional_predictions_t} compares the marginal predictive distributions under the LG-t-SRE and SG-t-SRE models in Section \ref{sec:simulation_studies} of the main article to those under the baseline models. Section \ref{sec:appendix_additional_predictions} shows a comparison of the predictive capabilities of the hierarchical copula-based models using LG-Gau-SRE and SG-Gau-SRE, compared to FRK and versions of the LG-Gau-SRE and SG-Gau-SRE models without measurement error. Section \ref{sec:appendix_posterior_parameters} shows results relating to parameter estimation through MCMC. This section concludes with computational-timing information, provided in Section \ref{sec:appendix_timing}.

\subsection{LG-Gau-SRE model}\label{sec:appendix_LGGau}

The LG-Gau-SRE model is specified as follows. Let $Y \sim \LG(m, s^2)$ denote a (generic) random variable $Y$ that follows a log-Gaussian distribution with log-scale mean $m$ and log-scale variance $s^2$. Now, for the data model, the $k$-th observation follows the conditional distribution,
\begin{equation}
    Z(A_{\TO k}) \mid Y(A_{\TO k}), \sigma_\TO^2 \sim \LG(\log(Y(A_j)) - 0.5\sigma^2_\TO, \sigma^2_\TO),\label{eqn:appendix_data_model_lg}
\end{equation}  and we also assume that the spatial data are conditionally independent given the underlying latent-process values. The data model in \eqref{eqn:appendix_data_model_lg} has been parameterized to guarantee that $E(Z(A_{\TO k}) \mid Y(A_{\TO k}), \sigma^2_\TO) = Y(A_{\TO k})$. 

For the marginals of the copula-based process, suppose $Y(A_j) \mid \beta_0, \sigma_\TP^2  \sim \LG(\beta_0-0.5\sigma^2_\TP, \sigma^2_\TP)$ for $j = 1, ..., N$, which means that  $E(Y(A_j) \mid \beta_0, \sigma^2_\TP) = \exp\{\beta_0\}$ and $\mathrm{var}(Y(A_j) \mid \beta_0, \sigma^2_\TP) = \exp\{2\beta_0\}\times(\exp\{\sigma^2_\TP\} - 1)$. In general, it is not necessary to restrict ourselves to a constant mean for the marginals. We can easily replace $\beta_0$ with the linear model, $\x(A_j)^\top \bm\beta$, where $\x(A_j)^\top$ is a vector of covariates at BAU $A_j$, for $j = 1, ..., N$. The dependence structure of $\Y$ is given by the Gau-SRE copula in \eqref{eqn:Gau_SRE_cop_model}, which has covariance-matrix parameter $\bm\Sigma^\SRE \equiv \S\E\S^\top + \I_N$. Here, $\E$ is an exponential covariance matrix; as explained in Section \ref{sec:models}. Letting $\theta_s, \theta_r > 0$ be parameters, we write the matrix as $\E(\theta_s, \theta_r)$, where the $(i,j)$ element of $\E$ is given by $E_{i,j} = \theta_s \exp\{-d_{ij}/\theta_r\}$, and $d_{ij}$ is the Euclidean distance (or, depending on the application, great-circle distance) between the $i$-th and $j$-th basis function centre. Hence, the joint PDF of $\Y$ is defined as follows: For $y_1, ..., y_N > 0$, 
\begin{align*}
&[\Y\mid\beta_0, \sigma^2_\TP, \theta_s, \theta_r]\\
&~~~~~~~= \left(\prod_{j=1}^N f_{LG,j}(y_j; \beta_0, \sigma_\TP^2)\right)\times \GCD^\SRE_{1:N}(F_{LG, 1}(y_1; \beta_0, \sigma_\TP^2)),...,F_{LG, N}(y_N; \beta_0, \sigma_\TP^2)); \bm\Sigma^\SRE),
\end{align*}
where recall that $\{f_{LG,j}(y_j; \beta_0, \sigma_\TP^2):j=1,...,N\}$ and $\{F_{LG,j}(y_j; \beta_0, \sigma_\TP^2):j=1,...,N\}$ denote, respectively, the PDFs and CDFs of the log-Gaussian variables $\{Y(A_j):j=1,...,N\}$, as defined in Section \ref{sec:models}. The priors are given in Section \ref{sec:appendix_sim_priors}.

\subsection{LG-t-SRE model}\label{sec:appendix_LGt}

The LG-t-SRE model is specified as follows. The LG-t-SRE model uses the same data model as the LG-Gau-SRE model. The marginal distributions of the latent process are also set to be log-Gaussian, such that
$$
Y(A_j) \mid \beta_0, \sigma^2_\TP \sim \LG(\beta_0 - 0.5\sigma^2_\TP, \sigma^2_\TP),~j=1, ..., N.
$$
The marginal variables $\{Y(A_j): j = 1, ..., N\}$ have been constructed to have $E(Y(A_j) \mid \beta_0, \sigma^2_\TP) = \exp\{\beta_0\}$ and $\mathrm{var}(Y(A_j)\mid \beta_0, \sigma_\TP^2) = \exp\{2\beta_0\} \times (\exp\{\sigma^2_\TP\} - 1)$. The marginals are used in the t-SRE copula \eqref{eqn:t_SRE_copula_model}, which has the parameters $\theta_s$, $\theta_r$, and $\nu > 2$. Hence, for $y_1, ..., y_N > 0$, 
\begin{align*}
&[\Y\mid\beta_0, \sigma^2_\TP, \theta_s, \theta_r, \nu]\\
&~~~~~~=\left(\prod_{j=1}^N f_{LG,j}(y_j; \beta_0, \sigma_\TP^2)\right)\times \TCD^\SRE_{1:N}(F_{LG, 1}(y_1; \beta_0, \sigma_\TP^2)),...,F_{LG, N}(y_N; \beta_0, \sigma_\TP^2)); \bm\Sigma^\SRE, \nu).
\end{align*}
The priors are given in Section \ref{sec:appendix_sim_priors}.

\subsection{SG-Gau-SRE model}\label{sec:appendix_SGGau}

The SG-Gau-SRE model is specified as follows. The SG-Gau-SRE model uses a Gaussian data model. Let $\sigma_\TO^2$ be a known constant measurement-error variance. Then, for $k = 1, ..., K$, $Z(A_{\TO k}) \mid Y(A_{\TO k}), \sigma_\TO^2 \sim \mathrm{Gau}(Y(A_{\TO k}), \sigma^2_\TO)$, and the spatial data are assumed to be conditionally independent given the underlying latent-process values. Clearly, $E(Z(A_{\TO k}) \mid Y(A_{\TO k}), \sigma^2_\TO) = Y(A_{\TO k})$ for $k = 1, ..., K$. 

For the process model, the marginal variables $\{Y(A_j): j = 1, ..., N\}$ are assumed to follow skew-Gaussian (SG) distributions \citep{Azzalini1985}. Let $Y(A_j) \sim \mathrm{SG}(\psi, \omega, \lambda)$ denote that the $j$-th latent process value, $j = 1, ..., N$, follows a skew-Gaussian distribution variable with location parameter $\psi \in \mathbb{R}$, scale parameter $\omega > 0$, and shape/skewness parameter $\lambda \in \mathbb{R}$. The PDF of $Y(A_j)$ is given by,
\begin{align*}
f_{SG,j}(y_j) = \frac{2}{\omega} \phi\left(\frac{y_j-\psi}{\omega}\right)\Phi\left(\frac{\lambda\times(y_j-\psi)}{\omega}\right); ~y_j \in \mathbb{R}.
\end{align*}
The subscript $SG$ denotes that the PDF $f_{SG,j}$ is for a skew-Gaussian distribution. The CDF of $Y$ is denoted by $F_{SG,j}(y_j), ~y_j\in\mathbb{R}$. The parameters relate to the mean and variance of $Y(A_j)$ as follows:
\begin{align}
E(Y(A_j) \mid \bm\theta_\TP) &= \psi + \omega \times\zeta \times \sqrt{2\pi^{-1}}\\
\mathrm{var}(Y(A_j) \mid \bm\theta_\TP) &= \omega^2 \times \left(1 - \frac{2\zeta^2}{\pi}\right)\label{eqn:var_sn},
\end{align}
where $\zeta = \lambda/\sqrt{1 +\lambda^2}$. Conversely, the skew-Gaussian parameters $\psi$, $\omega$, and $\lambda$ can be expressed in terms of $E(Y(A_j) \mid \bm\theta_\TP)$, $\mathrm{var}(Y(A_j) \mid \bm\theta_\TP)$, and $-1 < \zeta < 1$, as follows:
\begin{align}
\omega &= \sqrt{\mathrm{var}(Y(A_j)\mid \bm\theta_\TP)\Big/\left(1 - \frac{2\zeta^2}{\pi}\right)}\label{eqn:solve_var_sn}\\
\psi(A_j) &=  E(Y(A_j) \mid \bm\theta_\TP) - \sqrt{\mathrm{var}(Y(A_j)\mid \bm\theta_\TP)\Big/\left(1 - \frac{2\zeta^2}{\pi}\right)} \times \zeta \times \sqrt{2\pi^{-1}}\label{eqn:solve_e_sn}\\
\lambda &= \frac{\zeta}{\sqrt{1 - \zeta^2}}.
\end{align}
Therefore, for $j = 1, ..., N$, 
\begin{equation}
[Y(A_j) \mid \bm\theta_{\TP}] = f_{SG,j}(y_j; \bm\beta, \sigma^2_{\TP}, \lambda),~y_j \in \mathbb{R},\label{eqn:appendix_sn_marginals}
\end{equation}
which we parameterize in terms of a log-linear model, with BAU-level marginal variance $\sigma^2_{\TP}$, and skewness parameter $\lambda$, as follows:
\begin{align}
E(Y(A_j) \mid \bm\theta_{\TP}) &= \exp\{\mathbf{x}(A_j)^\top\bm\beta\} = \psi(A_j) + \omega \times \zeta \times \sqrt{2\pi^{-1}}\label{eqn:param_mean_sn}\\
\mathrm{var}(Y(A_j)\mid \bm\theta_{\TP}) &= \sigma^2_{\TP} = \omega^2 \times \left(1 - \frac{2\zeta^2}{\pi}\right).\label{eqn:param_var_sn}
\end{align}
Now $\psi(A_j)$ depends on the BAU, $A_j$, but $\omega$ does not. The skew-Gaussian parameters $\psi(A_j)$ and $\omega$ are uniquely determined by \eqref{eqn:solve_var_sn}--\eqref{eqn:solve_e_sn}. For the simulation study, we do not have covariates, so the log-linear model is simply $\exp\{\beta_0\}$. The specification of the SG-Gau-SRE model is completed by using these marginals in the Gau-SRE copula \eqref{eqn:Gau_SRE_cop_model}. This means that the density for the process model can be written as,
$$
[\Y \mid \beta_0, \sigma^2_\TP, \lambda, \theta_s, \theta_r] = \left(\prod_{j=1}^Nf_{SG, j}(y_j)\right) \times \GCD_{1:N}^\SRE(F_{SG, 1}(y_1), ..., F_{SG, N}(y_N); \bm\Sigma^\SRE),
$$
where it is understood that $F_{SG,j}$ and $f_{SG,j}$, $j = 1, ..., N$, are functions of $\beta_0$, $\sigma^2_\TP$, and $\lambda$, and $\bm\Sigma^\SRE$ is a function of $\theta_s$ and $\theta_r$. The priors for the parameters are given in Section \ref{sec:appendix_sim_priors}.

\subsection{SG-t-SRE model} \label{sec:appendix_SGt}

The SG-t-SRE model uses the same data model as the SG-Gau-SRE model, and it uses the same marginal distributions in the process model. The marginals are then used in the t-SRE copula \eqref{eqn:t_SRE_copula_model}, which has parameters $\theta_s$, $\theta_r$, and $\nu > 2$. Therefore, the process model can be written as follows: For $y_1, ..., y_N \in \mathbb{R}$,
$$
[\Y\mid\beta_0, \sigma^2_\TP, \lambda, \theta_s, \theta_r, \nu] = \left(\prod_{j=1}^N f_{SG,j}(y_j)\right)\times \TCD^\SRE_{1:N}(F_{SG, 1}(y_1)),...,F_{SG, N}(y_N)); \bm\Sigma^\SRE, \nu),
$$
where it is understood that $F_{SG,j}$ and $f_{SG,j}$, $j = 1, ..., N$, are functions of $\beta_0$, $\sigma^2_\TP$, and $\lambda$, and $\bm\Sigma^\SRE$ is a function of $\theta_s$ and $\theta_r$.
The priors are given in Section \ref{sec:appendix_sim_priors}.

\subsection{No-measurement-error (NME) versions of the copula-based models}\label{sec:appendix_NH_LGGau}

This section briefly describes the no-measurement-error (NME) versions of the LG-Gau-SRE, SG-Gau-SRE, LG-t-SRE, and SG-t-SRE models presented in Section \ref{sec:models} of the main article. 

In the NME versions of the models, it is assumed that the spatial process $Y(\cdot)$ can be observed directly at the observed BAUs, and there is no distinction between the process and the data. Hence models without measurement errors can be viewed as special cases of the hierarchical models of Section \ref{sec:models}, where the measurement-error variance is zero. Consequently, the previously `latent' spatial-process vector $\Y\equiv (\Y^\top_\TO, \Y^\top_\TM)^\top$ has first sub-vector that is no longer latent; the vector $\Y_\TO$ denotes observed spatial data. However, the vector $\Y_\TM$ is latent in the sense that it denotes the remaining `missing' values of the spatial process. 

Adapting the MCMC algorithms for the LG-Gau-SRE, SG-Gau-SRE, LG-t-SRE, and SG-t-SRE models to their NME versions is straightforward. The MCMC samplers for the process-model parameters $\bm\theta_\TP$ and the random effects $\etab$ (or $\etabb$ and scale factor $\gamma$ for the t-SRE copula models) remain unchanged. However, there are some meaningful changes elsewhere in the MCMC algorithm. First, we do not need to sample $\Y_\TO$ since we treat these as being directly observed. Second, it is meaningless to `predict' $\Y_\TO$ conditional on itself. The `predictions' at the observed BAUs are simply the observed values at those BAUs; hence they are excluded from any prediction-based comparison metrics (see Section \ref{sec:sim_results} in the main article). Third, for the prediction of the `missing' latent-process vector, $\Y_\TM$, we generate predictions through the predictive distribution $[\Y_\TM \mid \Y_\TO]$. It is straightforward to generate samples from the predictive distribution by adapting results of Section \ref{sec:predictions_gau_SRE} or \ref{sec:predictions_t_SRE} as appropriate. 

\subsection{Fixed rank kriging (FRK)}\label{sec:appendix_FRK_general}

Here, we give a brief description of fixed rank kriging (FRK) based on the Spatial Random Effects (SRE) model implemented in the FRK R package \citep{ZM2021, SainsburyDale2024}. It is based on a (hierarchical) Spatial Generalized Linear Mixed Model \citep[SGLMM;][]{Diggle1998, Sengupta2013} as follows. The data model is the usual one that assumes conditional independence of the spatial data given the underlying latent-process values: Letting $\bm\theta_\TO$ be a vector of known parameters for the data model, we write,
$$
[\Z_\TO \mid \Y_\TO, \bm\theta_\TO] = \prod_{k=1}^K[Z(A_{\TO k}) \mid Y(A_{\TO k}), \bm\theta_\TO].
$$
Each term in the data model follows an exponential-family distribution (e.g., Bernoulli, Gamma, Gaussian, or Poisson). There may or may not be parameters in the data model. For example, the Poisson data model depends on $Y(\cdot)$ but has no parameters, while the Gamma and Gaussian data models do. In the case of a Gaussian data model, each of $k=1, ..., K$ terms in the data model may be parameterized by a BAU-specific measurement-error variance, $\sigma_\TO(A_{\TO k})^2$, in which case we define $\bm\theta_\TO \equiv (\sigma_\TO(A_{\TO 1})^2, ..., \sigma_\TO(A_{\TO K})^2)^\top$. The measurement-error variances may be the same at all observed BAUs, in which case $\sigma_\TO(A_{\TO 1})^2 = \dots = \sigma_\TO(A_{\TO K})^2 = \sigma^2_\TO > 0$.

The process model specifies the joint distribution of the latent spatial process $\Y$ to be a trans-Gaussian spatial random effects (SRE) model. Let $\S(A_j)$ be a $b$-variate vector of spatial basis functions at BAU $A_j$, and let $\S$ be an $N\times b$ matrix of spatial basis functions, whose $j$-th row is given by $\S(A_j)^\top$, with $j = 1, ..., N$. Define $\bm\eta$ to be a $b$-variate vector of Gaussian spatial random effects with mean vector $\mathbf{0}$ and covariance matrix $\E$. Also let $\xi(A_j)$, $j = 1, ..., N$, be i.i.d Gaussian random variables, independent of $\bm\eta$, that have mean $0$ and variance $\sigma^2_\xi$. Because FRK is not designed to work with general copula-based models, the parameterization of $\E$ and $\sigma^2_\xi$ is different to the parameterization discussed in Section \ref{sec:identifiability} of the main article. Let $g(\cdot)$ be a link function. The FRK R package implements common options including the identity function, the logarithmic function, the inverse function, and the complementary log-log function, among others. Now, define $\Tilde{\Y} \equiv (g(Y(A_1)), ..., g(Y(A_N)))^\top$, which we model as,
$$
\Tilde{\Y} = \bm\mu + \S\etab + \xib,
$$
where $\bm\mu$ is mean vector, usually given by the linear model $\X\bm\beta$, where $\X$ is matrix of BAU-level covariates. We refer readers to \citet{SainsburyDale2024} for more details on Bayesian inference, including parameter estimation and spatial prediction, under this model.

\subsection{FRK for the LG-Gau-SRE and LG-t-SRE data}\label{sec:appendix_FRK_LG}

For the datasets simulated from the LG-Gau-SRE and LG-t-SRE models, we first log-transformed the data and then fitted a model with a Gaussian data model and Gaussian process model (i.e., $g(\cdot)$ is the identity function) with no covariates. In the data model, we set the log-scale variance to be $(\sigma_\TO^*)^2$ (i.e., the true log-scale measurement-error variance, which was assumed known). After samples of the predictive distribution were obtained on the log-scale, they were back-transformed to the original scale (i.e., by exponentiation), and the mean of the back-transformed samples was taken as the value of the predictor. The default settings in the FRK R package were used to fit the model. 

\subsection{FRK for the SG-Gau-SRE and SG-t-SRE data}\label{sec:appendix_FRK_SG}

For the datasets simulated from the SG-Gau-SRE and SG-t-SRE models, we simply used a Gaussian data model and Gaussian process model with no covariates in the FRK R package. The default settings in the FRK package were used to fit the model.

\subsection{Set up of simulation study}\label{sec:appendix_setup}

Fig. \ref{fig:BAU_config} shows the arrangement of 10,000 BAUs on the unit square. Each square is coloured according to whether it is an `Observed BAU' or a `Missing BAU' (i.e., a BAU with no observations). The same sets of observed and missing BAUs were used for all models (LG-Gau-SRE, LG-t-SRE, SG-Gau-SRE, SG-t-SRE) and for all $R = 100$ simulated datasets.

\begin{figure}[H]
\centering
\includegraphics[width=\textwidth]{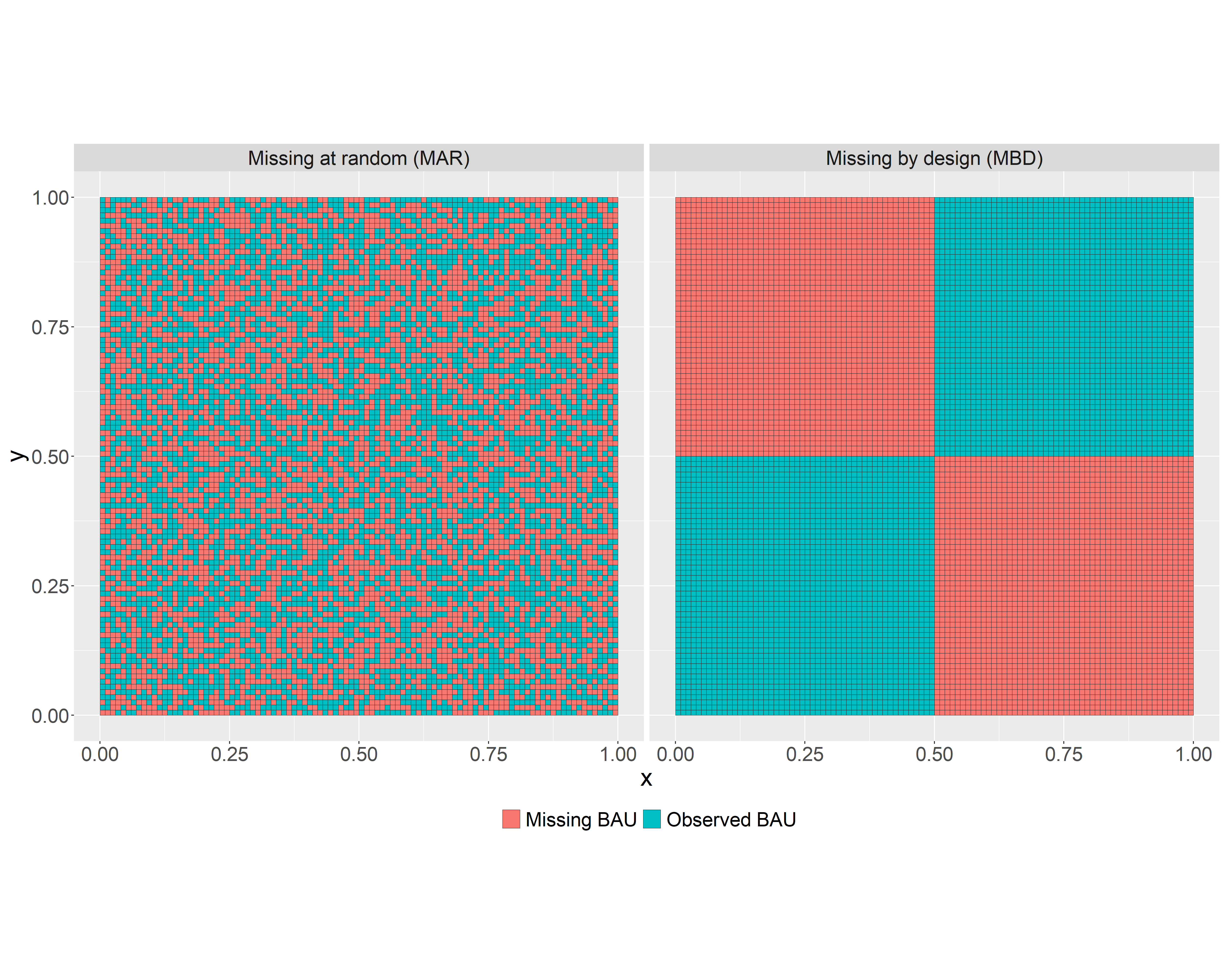}
\caption{Layout of Basic Areal Units (BAUs) with observations and missing observations. Of the total of $10,000$ BAUs, there are $K = 5,000$ BAUs with observations. The legend shows the schema for denoting `Observed BAUs' and `Missing BAUs'.}
\label{fig:BAU_config}
\end{figure}

Fig. \ref{fig:simulation_domain} shows the layout of the 36 bisquare spatial basis functions \citep{Cressie2008} used in the simulation study (Section \ref{sec:simulation_studies}). The bisquare basis functions were evaluated at the centroids of the 10,000 BAUs in the spatial domain; the calculations were performed with the FRK R package \citep{ZM2021, SainsburyDale2024}. The centers of the spatial basis functions are positioned on a $6\times 6$ regular grid on the unit square.

\begin{figure}[H]
\centering
\includegraphics[width=0.75\textwidth]{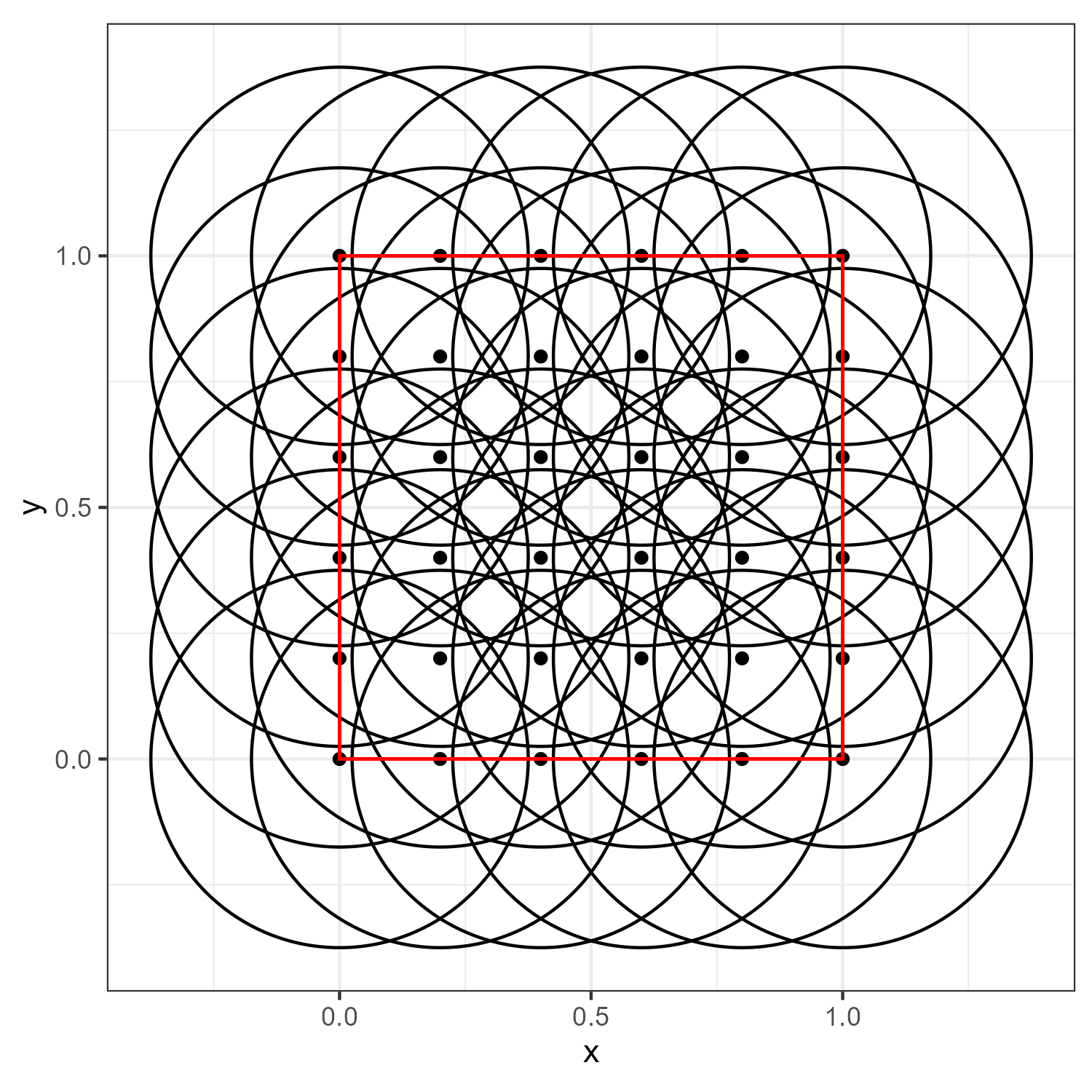}
\caption{The unit square with 36 bisquare spatial basis functions \citep{Cressie2008}. The dots are the centers of the bisquare basis functions. The radii of the circles are equal to 0.375 units.}
\label{fig:simulation_domain}
\end{figure}

\subsection{Simulating data from a Gau-SRE copula model}\label{sec:appendix_other_models_sims}

Section \ref{sec:setup} of the main article describes how data were simulated from the t-SRE copula model. The steps for simulating data from the Gau-SRE copula model are simpler, and they are as follows. Let $\bm\theta_\TP^*$ be a vector of all true parameter values. Then, for $r = 1, ..., R$, we simulated $\etab^{(r)} \sim [\etab \mid \bm\theta_\TP^*]$; then $\Y^{(r)} \sim [\Y\mid \etab^{(r)}, \bm\theta_\TP^*]$, where recall $\Y^{(r)} \equiv ((\Y_\TO^{(r)})^\top, (\Y_\TM^{(r)})^\top)^\top$; and finally, the observed spatial data $\Z_\TO^{(r)} \sim [\Z_\TO \mid \Y_\TO^{(r)}, (\sigma^*_\TO)^2]$. 

\subsection{Priors for the simulation study}\label{sec:appendix_sim_priors}

Table \ref{tab:priors_LG} displays the priors for the LG-Gau-SRE and LG-t-SRE models in the simulation study. Table \ref{tab:priors_SG} displays the priors for the SG-Gau-SRE and SG-t-SRE models in the simulation study. 

\begin{table}[H]
    \centering
    \caption{Prior distributions and their 95\% probability intervals for the parameters in the LG-Gau-SRE and LG-t-SRE models used for the simulation study. The Gamma (Gam) and truncated Gamma priors are parameterized according to the shape and scale; the Gaussian (Gau) prior is parameterized according to the mean and variance; the Half Cauchy (HC) priors have a scale parameter. The function $\mathbb{I}(\cdot)$ is the indicator function.}
    \begin{tabular}{|c|c|c|c|}
        \hline
         Parameter & Prior & 95\% probability interval & Description\\
         \hline
         $\sigma_\TP$ & $\mathrm{HC}(0.1)$ & (0.0039, 2.5452) & Half Cauchy\\
         $\beta_0$ & $\mathrm{Gau}(0, 100^2)$ & (-196, 196) & Gaussian\\
         $\theta_s$ & $\mathrm{Gam}(4, 2)$ & (2.18, 17.53) & Gamma\\
         $\theta_r$ & $\mathrm{HC}(0.25)$ & (0.0098, 6.3629) & Half Cauchy\\
         $\nu$ & $\mathbb{I}(\nu > 2) \times \mathrm{Gam}(3, 2)$ & (2.2366, 14.669) & Truncated Gamma\\
         \hline
    \end{tabular}
    \label{tab:priors_LG}
\end{table}

\begin{table}[H]
    \centering
    \caption{Prior distributions and their 95\% probability intervals for the parameters in the SG-Gau-SRE and SG-t-SRE models used for the simulation study. The Gamma (Gam) and truncated Gamma priors are parameterized according to the shape and scale; the Gaussian (Gau) prior is parameterized according to the mean and variance; the Half Cauchy (HC) priors have a scale parameter. The function $\mathbb{I}(\cdot)$ is the indicator function.}
    \begin{tabular}{|c|c|c|c|}
        \hline
         Parameter & Prior & 95\% probability interval & Description\\
         \hline
         $\sigma_{\TP}$ & $\mathrm{HC}(1000)$ & (39.3, 25451.7) & Half Cauchy\\
         $\lambda$ & $\mathrm{Gau}(0, 4^2)$ & (-7.84, 7.84) & Gaussian\\
         $\beta_0$ & $\mathrm{Gau}(0, 100^2)$ & (-196, 196) & Gaussian\\
         $\theta_s$ & $\mathrm{Gam}(4, 2)$ & (2.18, 17.54) & Gamma\\
         $\theta_r$ & $\mathrm{HC}(0.25)$ & (0.0098, 6.3629) & Half Cauchy\\
         $\nu$ & $\mathbb{I}(\nu > 2) \times \mathrm{Gam}(3, 2)$ & (2.2366, 14.669) & Truncated Gamma\\
         \hline
    \end{tabular}
    \label{tab:priors_SG}
\end{table}

\subsection{Initial values}\label{sec:appendix_sim_initial_values}

For $r= 1,..., R$, let $\sigma_\TP^{(0,r)}$, $\beta_0^{(0,r)}$, $\theta_s^{(0,r)}$, $\theta_r^{(0,r)}$, and $\nu^{(0,r)}$ be the initial values, respectively, of the parameters $\sigma_\TP$, $\beta_0$, $\theta_s$, $\theta_r$, and $\nu$ in the MCMC algorithm for the $r$-th simulated dataset. These were set as follows. For the LG-Gau-SRE and LG-t-SRE models, let $\log(\Z_\TO^{(r)})$ denote the elementwise logarithm of the $r$-th spatial dataset, for $r= 1,. ..., R$. Then, for the $r$-th dataset, $\sigma_\TP^{(0,r)}$ was set to the sample standard deviation of $\{\log(\Z_\TO^{(r)}): r = 1, ..., R\}$ (elementwise). Then, $\beta_0^{(0,r)}$ was set to the sample mean of $\{\log(\Z_\TO^{(r)}):r=1,...,R\}$. An empirical semivariogram was computed using $\log(\Z_\TO)$ detrended with a linear model involving the $x$ and $y$ coordinates of the data. The robust Cressie-Hawkins semivariogram estimator \citep{Cressie1980} was used, and all variogram calculations were done using the R package gstat \citep{Pebesma2004}. Subsequently, a spherical semivariogram model was fitted to the empirical semivariogram. If this failed, an exponential semivariogram model was fitted. If this failed, a wave semivariogram model was fitted. See \citet[pp. 61-63]{Cressie1993} for details about these models. The estimated range parameter from the fitted semivariogram model was used as $\theta_r^{(0,r)}$, and the ratio of the estimated partial sill to the estimated nugget effect was used as $\theta_s^{(0,r)}$. If the estimated nugget effect was zero, $\theta_s^{(0,r)} = 1$ was used instead. For the LG-t-SRE model, $\nu^{(0,r)} = 10$ was always used.

For $r = 1, ..., R$, let $\lambda^{(0,r)}$ denote the initial value of the skewness parameter $\lambda$ in the MCMC for the $r$-th dataset. For the SG-Gau-SRE and SG-t-SRE models, and for simulated datasets indexed by $r = 1, ..., R$, $\sigma_\TP^{(0,r)}$ was set to the sample standard deviation of $\{\Z_\TO^{(r)}:r=1,...,R\}$. For $\beta_0^{(0,r)}$, we used the sample mean of $\{\log(\Z_\TO^{(r)}): r = 1, ..., R\}$. The log-scale data are used because, per \eqref{eqn:param_mean_sn}, we have $E(Y(A_j)\mid \bm\theta_\TP) = \exp\{\beta_0\}$. Therefore, $\exp\{\beta_0\}$ is on the same scale as the spatial data $\{\Z^{(r)}: r = 1, ..., R\}$, and $\beta_0$ is on the same scale as $\{\log(\Z^{(r)}):r=1,...,R\}$. The sample skewness of $\Z_\TO^{(r)}$, calculated via the `moments' R package \citep{moments}, provided an initial value of $\lambda^{(0,r)}$. The same semivariogram-based procedure outlined above was used to initialize $\theta_s^{(0,r)}$ and $\theta_r^{(0,r)}$. Again, for the SG-t-SRE model, we set $\nu^{(0,r)} = 10$ for all datasets. 

\subsection{Additional results for the predictive performance of the LG-t-SRE and SG-t-SRE models}\label{sec:appendix_additional_predictions_t}

Fig. \ref{fig:LGt_Marginals} shows five randomly selected marginal predictive distributions among $\{[Y(A_{\TO k}) \mid \Z^{(1)}_\TO]: k = 1, ..., 5000\}$ and $\{[Y(A_{\TM l}) \mid \Z^{(1)}_\TO]: l = 1, ..., 5000\}$, under the LG-t-SRE, NME LG-t-SRE, and LG-Gau-SRE models, and FRK, based on the first ($r=1$) dataset $\Z_\TO^{(1)}$, which was simulated from the LG-t-SRE model. Fig. \ref{fig:SGt_Marginals} shows the same for the SG-t-SRE, NME SG-t-SRE, and SG-Gau-SRE models, and FRK, based on spatial data, $\Z_\TO^{(1)}$, simulated from the SG-t-SRE model. The marginal predictive distributions from FRK and the NME versions of the LG-t-SRE and SG-t-SRE models appear to deviate substantially from the marginal predictive distributions under the LG-t-SRE and SG-t-SRE models, either being more diffuse or lacking skewness when the marginal predictive distributions from the LG-t-SRE and SG-t-SRE models show skewness. This effect is more pronounced when the missing BAUs are MBD, and when the spatial data $\Z_\TO^{(1)}$ are simulated from the SG-t-SRE model. These results emphasise that using an appropriate copula-based hierarchical model can be advantageous, even when similar predictive performance can be achieved using Gaussian or simple trans-Gaussian models (e.g., FRK).

\begin{figure}
    \centering
    \includegraphics[width=.8\textwidth]{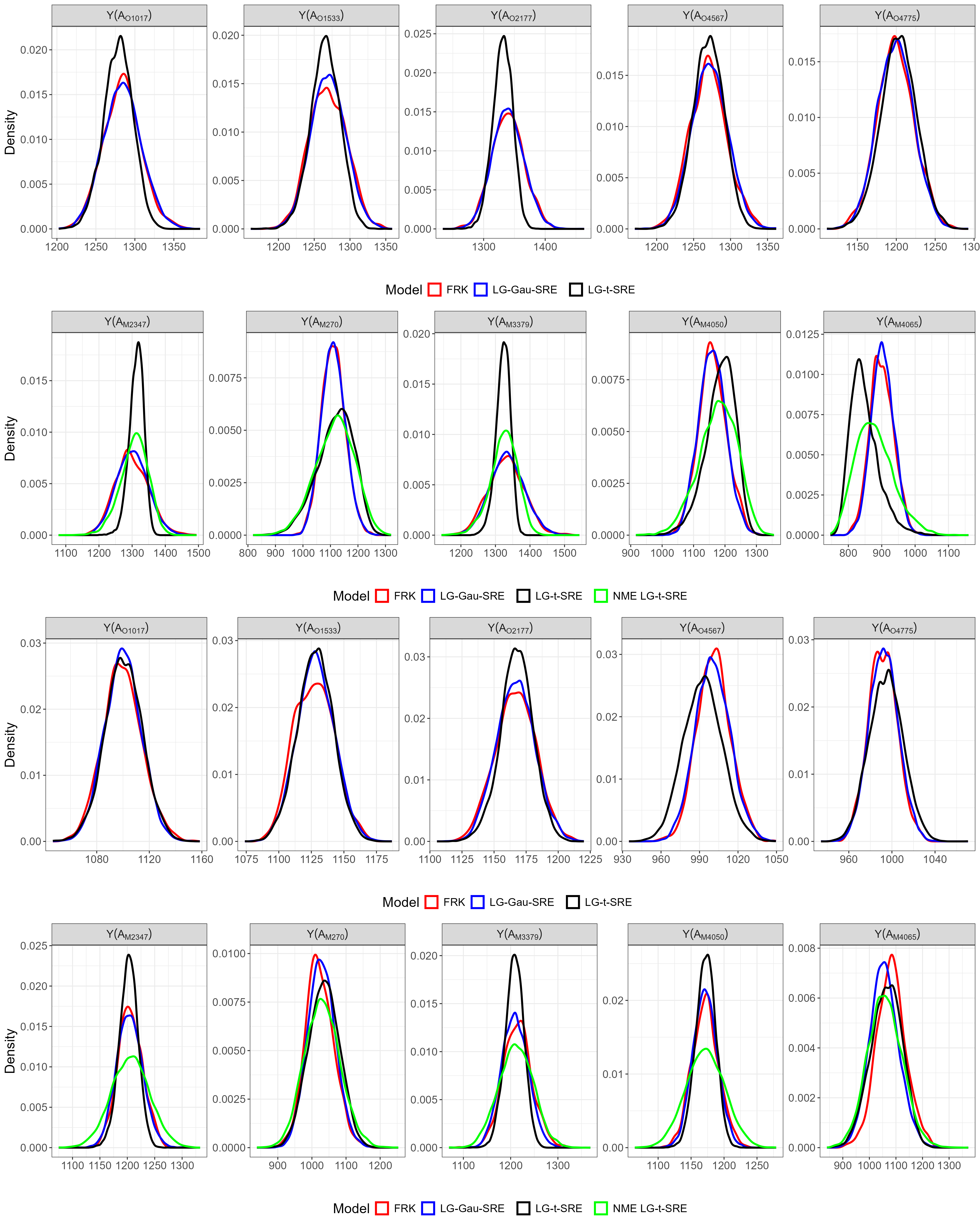}
    \caption{Kernel density estimates of the marginal predictive distributions of randomly selected elements of $\Y$ based on the first simulated dataset, $\Z_\TO^{(1)}$, under the LG-t-SRE model (which generated the data), the no-measurement-error (NME) LG-t-SRE model, the LG-Gau-SRE model, and fixed rank kriging (FRK). Top row: The marginal predictive distributions for five randomly selected elements of $\Y_\TO$ when the missing BAUs are missing at random (MAR). Second row: The marginal predictive distributions for five randomly selected elements of $\Y_\TM$ when the missing BAUs are MAR. Third row: As for the first row, but the missing BAUs are missing by design (MBD). Bottom row: As for the second row, but the missing BAUs are MBD.}
    \label{fig:LGt_Marginals}
\end{figure}

\begin{figure}
    \centering
    \includegraphics[width=.8\textwidth]{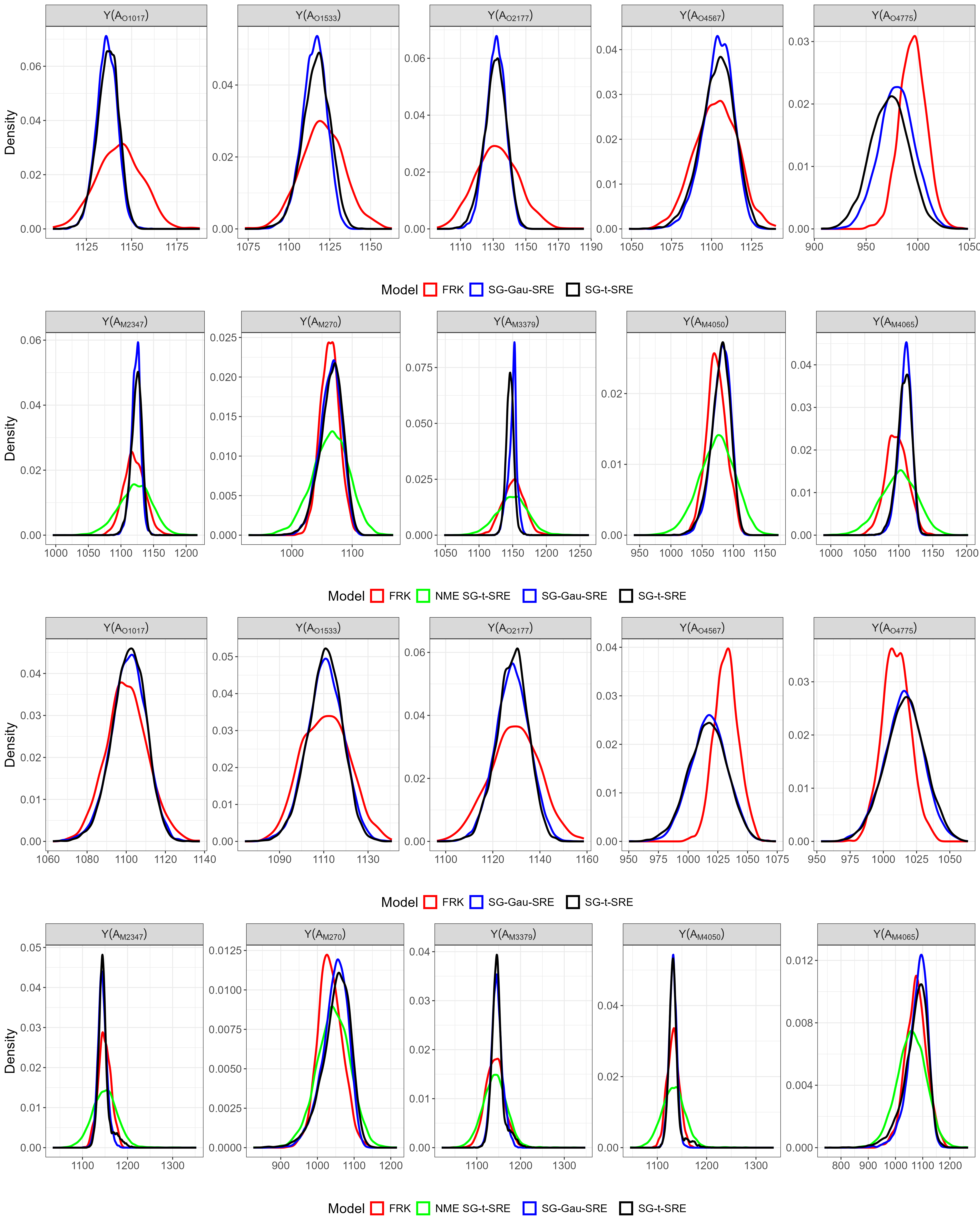}
    \caption{Kernel density estimates of the marginal predictive distributions of randomly selected elements of $\Y$ based on the first simulated dataset, $\Z^{(1)}$, under the SG-t-SRE model (which generated the data), the no-measurement-error (NME) SG-t-SRE model, the SG-Gau-SRE model, and fixed rank kriging (FRK). Top row: The marginal predictive distributions for five randomly selected elements of $\Y_\TO$ when the missing BAUs are missing at random (MAR). Second row: The marginal predictive distributions for five randomly selected elements of $\Y_\TM$ when the missing BAUs are MAR. Third row: As for the first row, but the missing BAUs are missing by design (MBD). Bottom row: As for the second row, but the missing BAUs are MBD.}
    \label{fig:SGt_Marginals}
\end{figure}

\subsection{Additional results for the predictive performance of the LG-Gau-SRE and SG-Gau-SRE models}\label{sec:appendix_additional_predictions}

Here we show that, when data are simulated from the LG-Gau-SRE and SG-Gau-SRE models, the respective model accurately predicts the latent spatial process, achieves adequate coverage of the 90\% prediction intervals, and outperforms the no-measurement-error versions of these models (NME LG-Gau-SRE and NME SG-Gau-SRE) and FRK. 

Fig. \ref{fig:appendix_LGGau_prediction_comparisons} shows the BAU-wise root-mean-squared prediction errors (RMSPEs) for the LG-Gau-SRE model compared to those from the NME LG-Gau-SRE model and FRK when the data are simulated from the LG-Gau-SRE model. Fig. \ref{fig:appendix_SGGau_prediction_comparisons} shows the same for the SG-Gau-SRE model compared to the NME SG-Gau-SRE model, and FRK when the data are simulated from the SG-Gau-SRE model. The plots show that, when the missing BAUs are MAR, the SG-Gau-SRE and LG-Gau-SRE models show, at most, marginal improvements in terms of predictive ability compared to the reference models, namely their no-measurement-error counterparts and FRK. However, the LG-Gau-SRE and SG-Gau-SRE models decisively outperform the reference models when the missing BAUs are MBD. The latter case is a better test of the models' predictive capabilities, since the former case (MAR) represents a scenario where accurate prediction of the process at the missing BAUs is easy since every missing BAU is virtually guaranteed to be close to several observed BAUs. 

\begin{figure}[!ht]
    \centering
    \includegraphics[width=\linewidth]{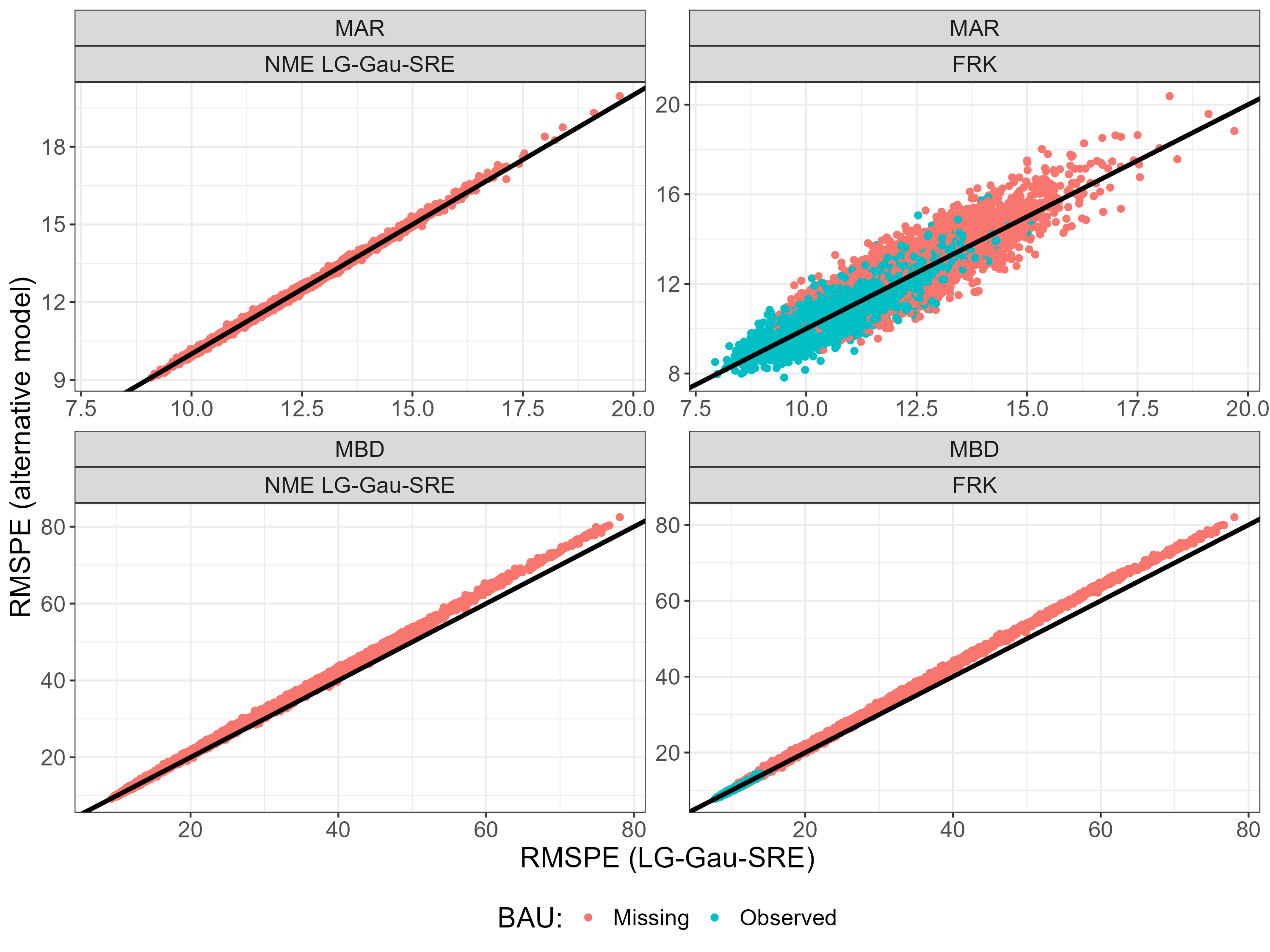}
    \caption{For data generated under the LG-Gau-SRE model, root-mean-squared prediction errors (RMSPEs) under the LG-Gau-SRE model plotted against the RMSPEs under a no-measurement-error version of the LG-Gau-SRE model, and the FRK predictor at each missing and observed basic areal unit (BAU), when the BAUs are missing at random (MAR) or missing by design (MBD). Points above the 45 degree line indicate the alternative model has a higher RMSPE than the LG-Gau-SRE model for that BAU.}
    \label{fig:appendix_LGGau_prediction_comparisons}
\end{figure}

\begin{figure}[!ht]
    \centering
    \includegraphics[width=\linewidth]{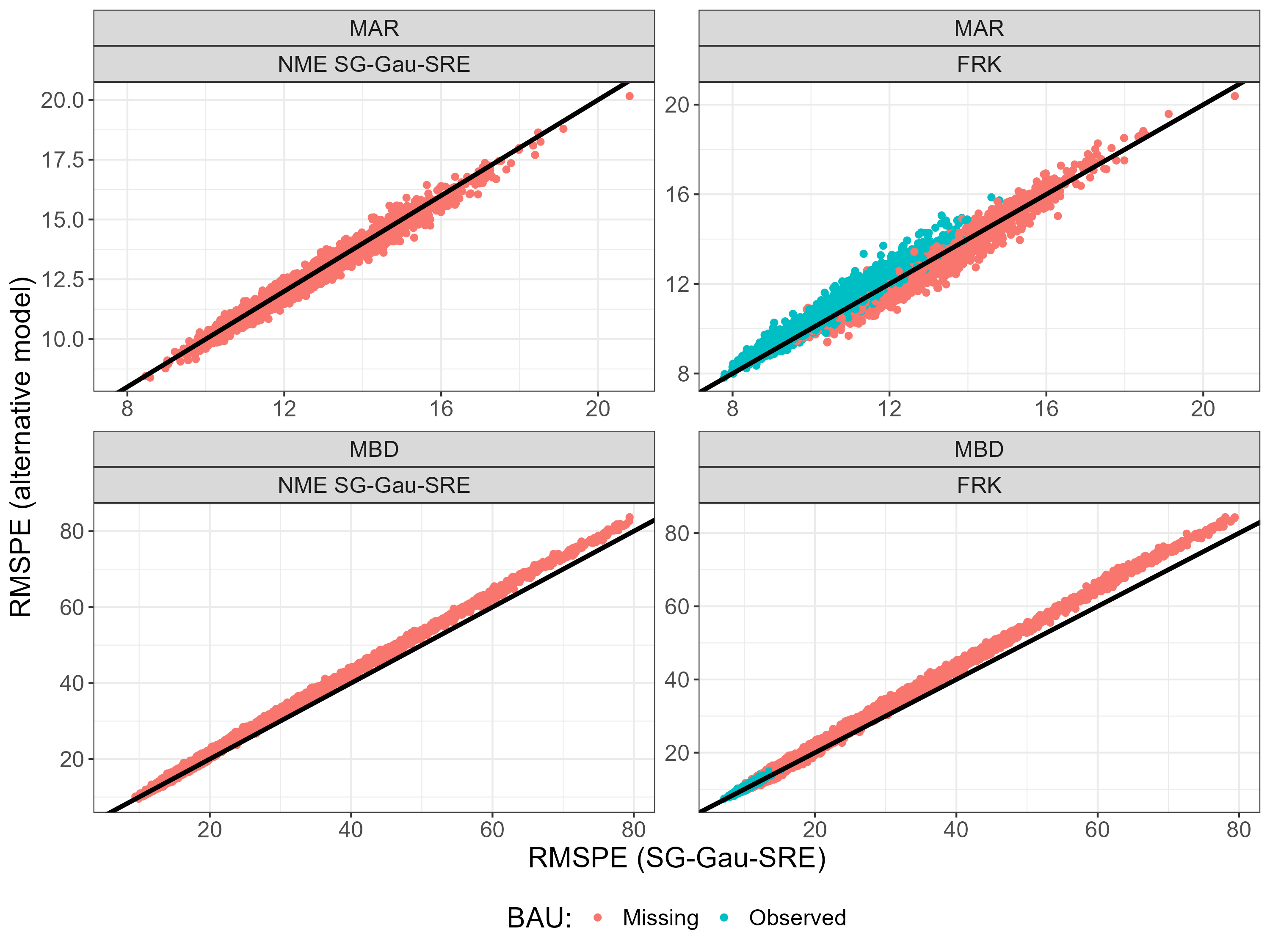}
    \caption{For data generated under the SG-Gau-SRE model, root-mean-squared prediction errors (RMSPEs) under the SG-Gau-SRE model plotted against the RMSPEs under a no-measurement-error version of the SG-Gau-SRE model and the FRK predictor at each missing and observed basic areal unit (BAU), when the BAUs are missing at random (MAR) or missing by design (MBD). Points above the 45 degree line indicate the alternative model has a higher RMSPE than the SG-Gau-SRE model for that BAU.}
    \label{fig:appendix_SGGau_prediction_comparisons}
\end{figure}

Fig. \ref{fig:appendix_extra_empirical_coverages} shows the BAU-wise empirical coverages in \eqref{eqn:EC_definition} under the LG-Gau-SRE and SG-Gau-SRE models, as well as those for FRK and the NME LG-Gau-SRE and NME SG-Gau-SRE models. These results show that, when the missing BAUs are MAR, the NME SG-Gau-SRE model is the only no-measurement-error model to achieve the correct empirical coverages for its 90\% prediction intervals. Otherwise, all no-measurement-error models have 90\% prediction intervals that are invalid due to overcoverage. All other models (the SG-Gau-SRE model and FRK) have nominal 90\% prediction intervals whose BAU-wise empirical coverages are approximately 90\%. 

\begin{figure}
    \centering
    \includegraphics[width=\linewidth]{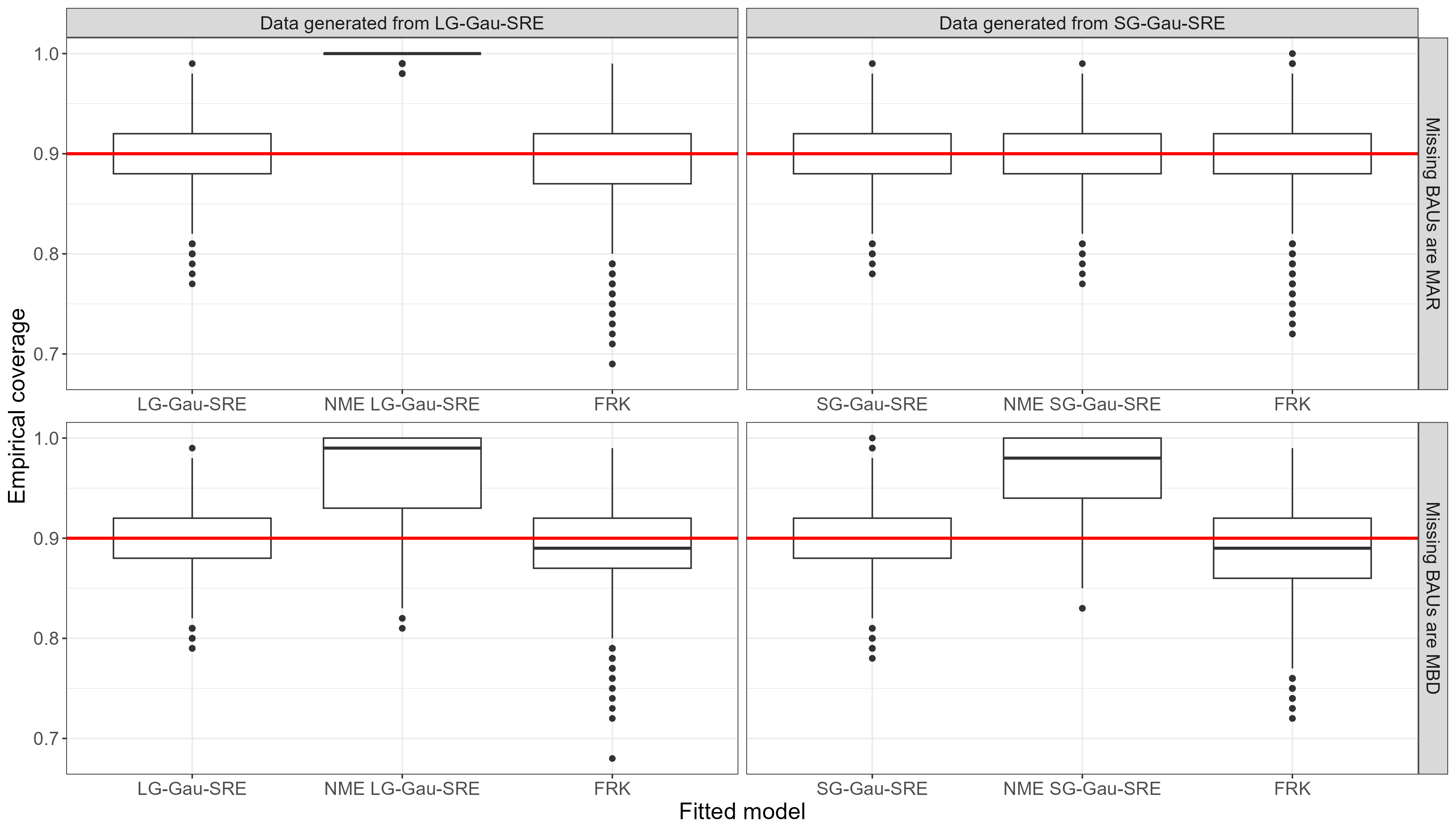}
    \caption{Empirical coverages in \eqref{eqn:EC_definition} calculated for all BAUs in $D$ for a range of models fitted to data simulated from either the LG-Gau-SRE model or the SG-Gau-SRE model. The panel titles indicate which copula-based model was used to simulate the data and whether the missing BAUs were `Missing at Random' (MAR) or `Missing by Design' (MBD).}\label{fig:appendix_extra_empirical_coverages}
\end{figure}

\subsection{Results for inference on the parameters}\label{sec:appendix_posterior_parameters}

Recall the parameters of the models and their true values. The true values of the parameters were set as $\beta_0^* = \log(1000)$, $\lambda^* = -5$, $\theta_s^* = 10$, $\theta_r^* = \sqrt{2}/4$, and $\nu^* = 4$ for all models. For the LG-Gau-SRE and LG-t-SRE models, we have $(\sigma^*_\TP)^2 = (0.1)^2$. For the SG-Gau-SRE and SG-t-SRE models, we have $(\sigma^*_\TP)^2 = (100)^2$. 

For every combination of our models (LG-Gau-SRE, SG-Gau-SRE, LG-t-SRE, and SG-t-SRE) and each arrangement of the missing BAUs (MAR and MBD), we examined the sample mean and sample standard error (SE) over the $R = 100$ simulated datasets of the posterior means of the parameters. The bias of the estimates was also computed. For generic parameter $\theta$, and estimator $\hat{\theta} = E(\theta \mid \Z)$, the bias is approximated as $\mathrm{Bias}(\hat{\theta}) \approx R^{-1}\sum_{r=1}^R (\hat{\theta}^{(r)} - \theta),$ where $\hat{\theta}^{(r)} = E(\theta \mid \Z^{(r)})$ from the $r$-th simulated dataset. Similarly, the sample SE is calculated as $\mathrm{SE}(\hat{\theta}) \approx \sqrt{(R-1)^{-1}\sum_{r=1}^R (\hat{\theta}^{(r)}- \bar{\theta})^2}$, where $\bar{\theta} = R^{-1}\sum_{r=1}^R \hat{\theta}^{(r)}$ Finally, the proportion of datasets for which the 95\% credible interval for each parameter contained the true value of the parameter was computed. 

Table \ref{tab:LG_Parameters_MAR} summarizes the results for the MAR datasets simulated from the LG-Gau-SRE and LG-t-SRE models. Table \ref{tab:LG_Parameters_MBD} summarizes the results for the MBD datasets. In all cases, the averages of the posterior means are clearly close to the true values. All averages of posterior means are within one standard error of the true values, with the only exception being the degrees of freedom parameter in the t-SRE copula (an inherently challenging problem).

\begin{table}[H]
    \centering
    \caption{Comparison table for true values of parameters against the average and standard error (SE) of the posterior mean values of the parameters across 100 simulated datasets from each of the LG-Gau-SRE and LG-t-SRE models with the missing BAUs being `missing at random' (MAR). The bias and the proportion of 95\% credible intervals (CIs) that covered the true value of the parameter across 100 simulated datasets, are also shown.}
    \begin{tabular}{|p{2cm}|p{1cm}|p{2.5cm}p{2.5cm}|p{1.5cm}p{1.5cm}|p{1cm}p{1cm}|}
        \hline
         Parameter & True value & \multicolumn{2}{p{5cm}|}{Average (SE) of posterior means } & \multicolumn{2}{p{2.5cm}}{Bias} & \multicolumn{2}{|p{2cm}|}{Proportion of true values in 95\% CI}\\
	\hline
	&  & Gau & t & Gau & t & Gau & t \\	
         \hline
  $\sigma_{\TP}$ & 0.1 & 0.093 (0.014) & 0.091 (0.020) & -0.007 & -0.009 & 0.97 & 0.94 \\ 
  $\beta_0$ & 6.9078 & 6.906 (0.033) & 6.906 (0.024) & -0.002 & -0.002 & 0.94 & 0.93 \\ 
  $\theta_s$ & 10  & 9.45 (1.756) & 9.337 (1.648) & -0.546 & -0.663 & 0.99 & 1.00 \\ 
  $\theta_r$ & 0.3536  & 0.362 (0.110) & 0.358 (0.111) & 0.008 & 0.005 & 0.99 & 0.96 \\ 
  $\nu$ & 4 &  & 6.011 (1.563) &  & 2.011 &  & 1.00 \\ 
         \hline
    \end{tabular}
    \label{tab:LG_Parameters_MAR}
\end{table}

\begin{table}[H]
    \centering
    \caption{Comparison table for true values of parameters against the average and standard error (SE) of the posterior mean values of the parameters across 100 simulated datasets  from each of the LG-Gau-SRE and LG-t-SRE models with the missing BAUs being `missing by design' (MBD). The bias and the proportion of 95\% credible intervals (CIs) that covered the true value of the parameter across 100 simulated datasets, are also shown.}
    \begin{tabular}{|p{2cm}|p{1cm}|p{2.5cm}p{2.5cm}|p{1.5cm}p{1.5cm}|p{1cm}p{1cm}|}
        \hline
         Parameter & True value & \multicolumn{2}{p{5cm}|}{Average (SE) of posterior means } & \multicolumn{2}{p{2.5cm}}{Bias} & \multicolumn{2}{|p{2cm}|}{Proportion of true values in 95\% CI}\\
	\hline
	&  & Gau & t & Gau & t & Gau & t \\	
         \hline
  $\sigma_{\mathrm{\TP}}$ & 0.1 & 0.092 (0.015) & 0.093 (0.019) & -0.008 & -0.007 & 0.98 & 0.97 \\ 
  $\beta_0$ & 6.9078 & 6.907 (0.039) & 6.906 (0.024) & 0.000 & -0.002 & 0.93 & 0.95 \\ 
  $\theta_s$ & 10  & 9.328 (1.675) & 9.164 (1.687) & -0.672 & -0.835 & 1.00 & 0.99 \\ 
  $\theta_r$ & 0.3536  & 0.366 (0.138) & 0.367 (0.130) & 0.012 & 0.013 & 0.96 & 0.98 \\ 
  $\nu$ & 4 &  & 5.953 (1.341) &  & 1.953 &  & 1.00 \\ 
         \hline
    \end{tabular}
    \label{tab:LG_Parameters_MBD}
\end{table}

Tables \ref{tab:SG_Parameters_MAR} and \ref{tab:SG_Parameters_MBD} show parameter-recovery results for the SG-Gau-SRE and SG-t-SRE models when the BAUs are MAR (Table \ref{tab:SG_Parameters_MAR}) and MBD (Table \ref{tab:SG_Parameters_MBD}). As with the results for the LG-Gau-SRE and LG-t-SRE models, it is evident that all parameter estimates are close to the true values. Except for the skewness ($\lambda$) and degrees of freedom ($\nu$) parameters in the SG-t-SRE model, all averages of the posterior means are within one standard error of the true values.

\begin{table}[H]
    \centering
    \caption{Comparison table for true values of parameters against the average and standard error (SE) of the posterior mean values of the parameters across 100 simulated datasets from each of the SG-Gau-SRE and SG-t-SRE models with the missing BAUs being `missing at random' (MAR). The bias and the proportion of 95\% credible intervals (CIs) that covered the true value of the parameter across 100 simulated datasets, are also shown}
    \begin{tabular}{|p{2cm}|p{1cm}|p{2.5cm}p{2.5cm}|p{1.5cm}p{1.5cm}|p{1cm}p{1cm}|}
        \hline
         Parameter & True value & \multicolumn{2}{p{5cm}|}{Average (SE) of posterior means } & \multicolumn{2}{p{2.5cm}}{Bias} & \multicolumn{2}{|p{2cm}|}{Proportion of true values in 95\% CI}\\
	\hline
	&  & Gau & t & Gau & t & Gau & t \\	
         \hline
  $\sigma_{\TP}$ & 100 & 97.074 (18.10) & 89.680 (17.847) & -2.926 & -10.320 & 0.94 & 0.96 \\ 
  $\beta$ & 6.9078 & 6.904 (0.030) & 6.898 (0.022) & -0.004 & -0.010 & 0.88 & 0.95 \\ 
  $\lambda$ & -5  & -4.492 (1.112) & -3.196 (1.796) & 0.508 & 1.804 & 0.99 & 0.91 \\ 
  $\theta_s$ & 10  & 9.630 (1.950) & 9.194 (1.889) & -0.370 & -0.806 & 0.98 & 0.96 \\ 
  $\theta_r$ & 0.3536  & 0.368 (0.106) & 0.379 (0.145) & 0.015 & 0.025 & 0.95 & 0.93 \\ 
  $\nu$ & 4 &  & 5.839 (1.301) &  & 1.839 &  & 1.00 \\ 
         \hline
    \end{tabular}
    \label{tab:SG_Parameters_MAR}
\end{table}

\begin{table}[H]
    \centering
    \caption{Comparison table for true values of parameters against the average and standard error (SE) of the posterior mean values of the parameters across 100 simulated datasets  from each of the SG-Gau-SRE and SG-t-SRE models with the missing BAUs being `missing by design' (MBD). The bias and the proportion of 95\% credible intervals (CIs) that covered the true value of the parameter across 100 simulated datasets, are also shown.}
    \begin{tabular}{|p{2cm}|p{1cm}|p{2.5cm}p{2.5cm}|p{1.5cm}p{1.5cm}|p{1cm}p{1cm}|}
        \hline
         Parameter & True value & \multicolumn{2}{p{5cm}|}{Average (SE) of posterior means } & \multicolumn{2}{p{2.5cm}}{Bias} & \multicolumn{2}{|p{2cm}|}{Proportion of true values in 95\% CI}\\
	\hline
	&  & Gau & t & Gau & t & Gau & t \\	
         \hline
  $\sigma_{\TP}$ & 100 & 96.572 (19.65) & 95.940 (22.883) & -3.428 & -4.060 & 0.95 & 0.96 \\ 
  $\beta$ & 6.9078 & 6.902 (0.033) & 6.897 (0.027) & -0.006 & -0.011 & 0.94 & 0.95 \\ 
  $\lambda$ & -5  & -4.300 (1.129) & -3.109 (1.637) & 0.700 & 1.891 & 0.97 & 0.95 \\ 
  $\theta_s$ & 10  & 9.522 (1.601) & 9.028 (1.774) & -0.478 & -0.972 & 1.00 & 0.99 \\ 
    $\theta_r$ & 0.3536  & 0.381 (0.142) & 0.363 (0.130) & 0.027 & 0.009 & 0.93 & 0.97 \\ 
  $\nu$ & 4 &  & 5.741 (1.137) &  & 1.741 &  & 1.00 \\ 
         \hline
    \end{tabular}
    \label{tab:SG_Parameters_MBD}
\end{table}

\subsection{Timing of computations}\label{sec:appendix_timing}

For the LG-Gau-SRE and SG-Gau-SRE models, for each configuration of missing BAUs (MAR and MBD), and for each of $r = 1, ..., 100$ simulated datasets, a total of 45,000 MCMC iterations (including burn-in) were run concurrently on Australia's National Computing Infrastructure's GADI high-performance computing cluster, using R version 4.2.2. Each set of 45,000 MCMC iterations took, on average, approximately 8 minutes for the LG-Gau-SRE model and 16-17 minutes for the SG-Gau-SRE model. The NME versions of the LG-Gau-SRE and SG-Gau-SRE models each took approximately 6-7 minutes and 11-12 minutes, respectively, for 45,000 MCMC iterations.

\section{Additional details for the methane (XCH$_4$) remote-sensing application}\label{sec:appendix_real_data}

This section provides additional figures, tables, and modeling details for the analysis of the atmospheric methane (XCH$_4$) data in Section \ref{sec:methane} in the main article. Section \ref{sec:appendix_methane_baus} presents the 7,370 BAUs that cover the study area. Sections \ref{sec:appendix_daily_histograms}-\ref{sec:appendix_measurement_error} present plots for exploratory analysis of these data. Section \ref{sec:appendix_methane_basis} describes the construction of the 101 bisquare spatial basis functions used in the analysis. Section \ref{sec:appendix_averaging_data} discusses what to do when two or more retrievals fall in a single BAU. Section \ref{sec:appendix_skew_Gaussian} discusses the SG-Gau-SRE model used to capture the skewness seen in the XCH$_4$ data in the Bowen Basin, Queensland, Australia. Section \ref{sec:appendix_FRK_methane} describes FRK for the methane data.  Section \ref{sec:appendix_methane_priors} presents the set of weakly informative priors used for the MCMC in the Bayesian analysis presented in Section \ref{sec:methane} of the main article, and Section \ref{sec:appendix_methane_initial} discusses the initial values used for the MCMC. Finally, trace plots for the parameters, random effects, latent-process values are shown in Section \ref{sec:appendix_methane_trace_plots}, and formal convergence diagnostics for the MCMC are reported in Section \ref{sec:appendix_convergence_diagnostics}. 

\subsection{Basic Areal Units in the Bowen Basin, Queensland, Australia}\label{sec:appendix_methane_baus}

For the Bowen Basin in Queensland, Australia, we defined the BAUs as a regular grid of 0.05dd $\times$ 0.05dd squares (`dd' $\equiv$ decimal degrees) that tesselate $D$, which is the area between 26$^\circ$S and 20.5$^\circ$S and between 147$^\circ$E and 151$^\circ$E (or the Queensland coastline); see Fig. \ref{fig:appendix_methane_baus}. Since retrieval characteristics of XCH$_4$ are different over land and over water, we only kept a BAU if more than 50\% of its area is terrestrial. There were a total of 7,370 BAUs in $D$.

\begin{figure}[H]
\centering
\includegraphics[width=0.65\textwidth]{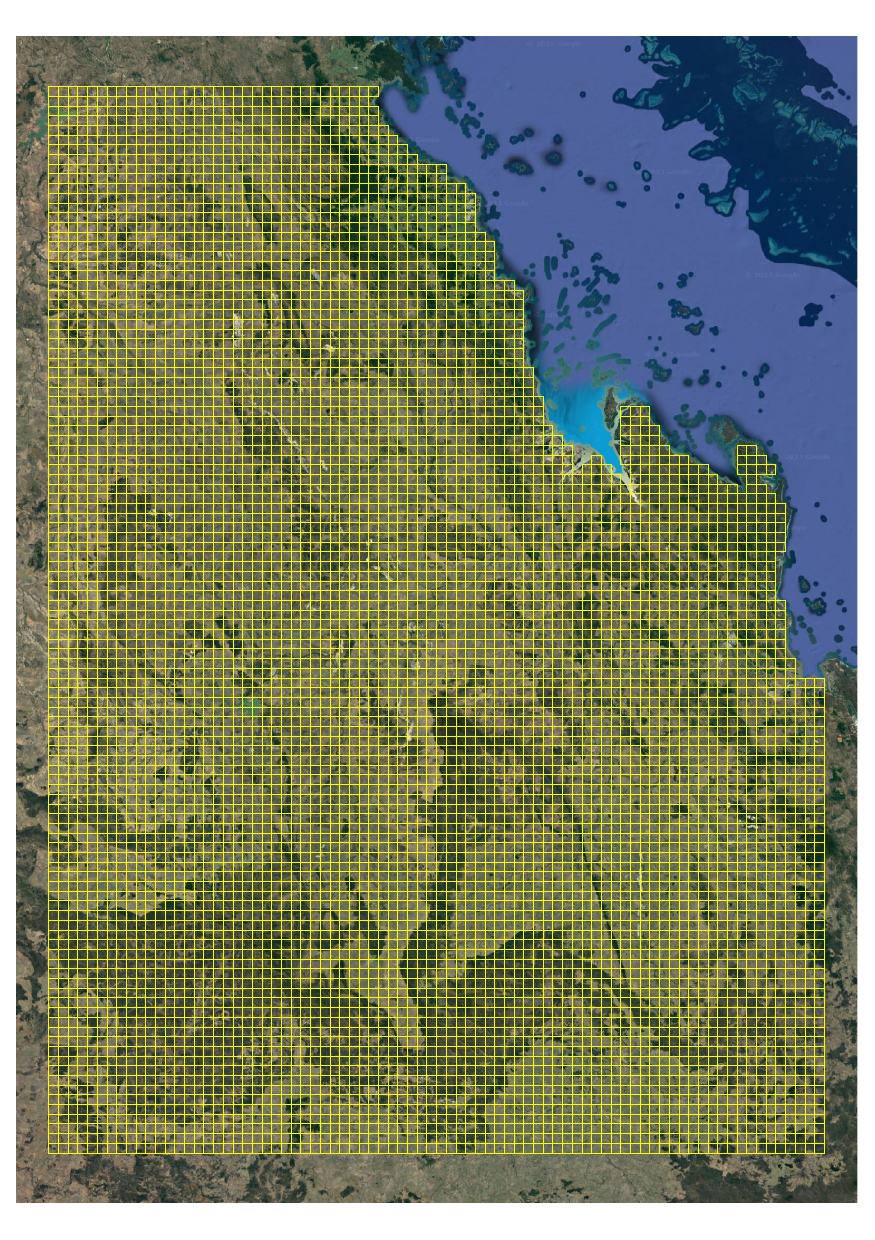}
\caption{A map of the 7,370 Basic Areal Units (BAUs; grid squares) that tesselate our study area in and around the Bowen Basin, Queensland, Australia. The map is plotted in the EPSG:4326 coordinate system.}
\label{fig:appendix_methane_baus}
\end{figure}

\subsection{Daily histograms}\label{sec:appendix_daily_histograms}

In August, 2020, there were 41,750 column-average methane (XCH$_4$) retrievals in the study area defined in and around the Bowen Basin in the state of Queensland, Australia. These retrievals occurred across 24 of the 31 days in August; no retrievals were recorded on the other 7 days of August. Fig. \ref{fig:methane_histograms} shows the histograms of XCH$_4$ recorded on each day: There are days where the histograms are positively skewed, days where the histograms are negatively skewed, and days where the histograms appear to be symmetric.

\begin{figure}[H]
\centering
\includegraphics[width=.8\textwidth]{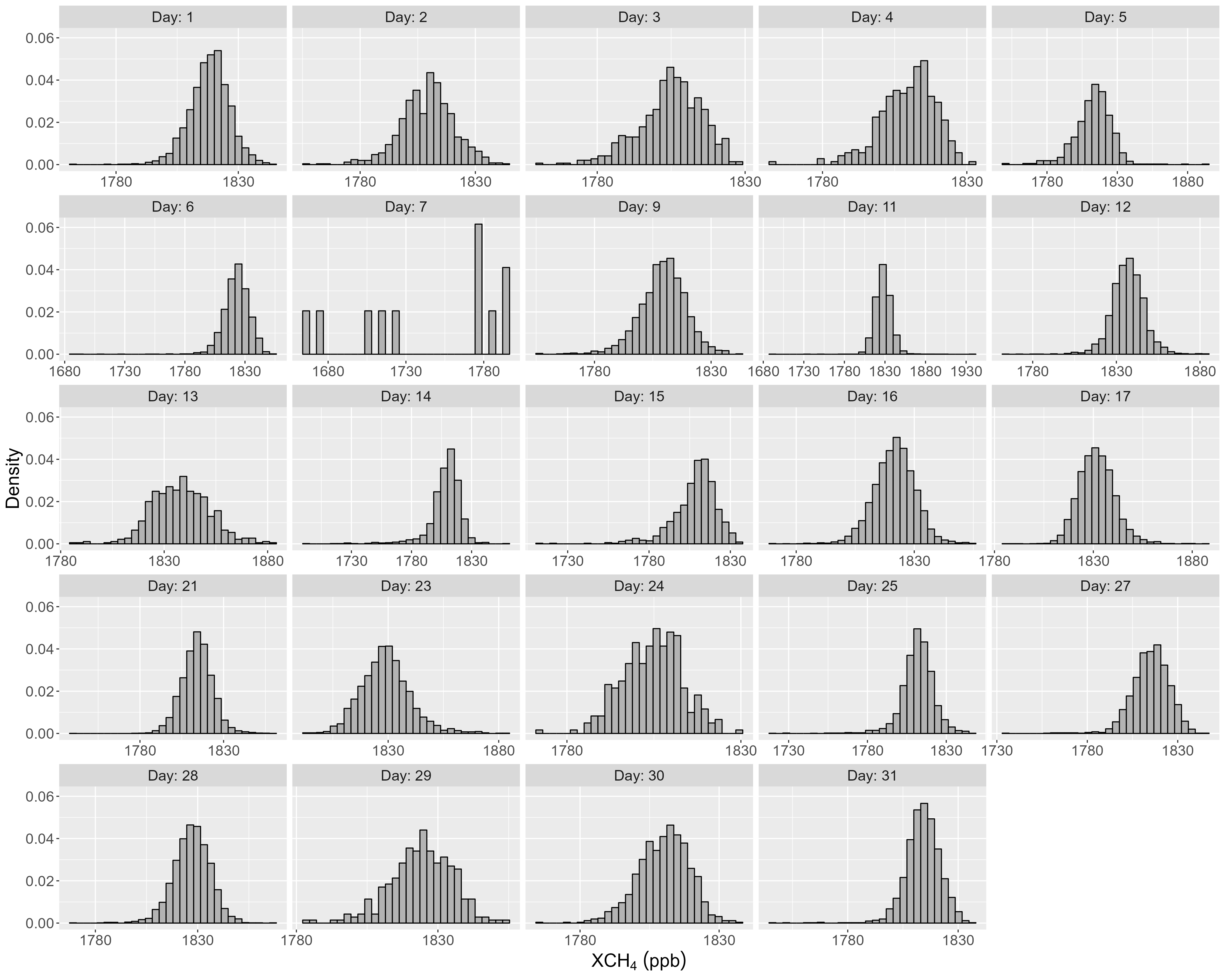}
\caption{Histograms of XCH$_4$ values retrieved by the Sentinel 5P satellite in August, 2020, by day. Some days of the month are missing because no retrievals occurred in the study area on those days.}
\label{fig:methane_histograms}
\end{figure}

\subsection{Counts of observations in each day}\label{sec:appendix_daily_counts}

The number of XCH$_4$ retrievals recorded on each day is shown in Table \ref{tab:number_of_obs}. Note that the numbers in Table \ref{tab:number_of_obs} represent the total number of retrievals on each day. If multiple retrievals fall into a single BAU, the values of the retrievals are averaged. This may mean that the total number of observed BAUs on each day is lower than the number listed in Table \ref{tab:number_of_obs}.

\begin{table}[H]
\centering
\caption{Number of XCH$_4$ retrievals on each day of retrievals during August, 2020, in and around the Bowen Basin in Queensland, Australia.}\label{tab:number_of_obs}
\begin{tabular}{rr}
  \hline
 Day (August 2020) & No. of XCH$_4$ retrievals \\ 
  \hline
   1 & 1611 \\ 
   2 & 845 \\ 
     3 & 624 \\ 
     4 & 298 \\ 
     5 & 768 \\ 
   6 & 2172 \\ 
     7 &  11 \\ 
    9 & 1811 \\ 
    11 & 4494 \\ 
    12 & 2625 \\ 
  13 & 555 \\ 
  14 & 1274 \\ 
  15 & 824 \\ 
  16 & 4152 \\ 
 17 & 3288 \\ 
  21 & 4611 \\ 
   23 & 2110 \\ 
   24 & 305 \\ 
  25 & 2169 \\ 
   27 & 1671 \\ 
  28 & 2537 \\ 
  29 & 289 \\ 
   30 & 1044 \\ 
 31 & 1662 \\ 
   \hline
Total: & 41750\\
\hline
\end{tabular}
\end{table}

\subsection{Measurement errors of the XCH$_4$ retrievals}\label{sec:appendix_measurement_error}

On any given day, let $\{\s_{i} \in D:i=1,...,n\}$ denote the set of $n$ spatial locations (reported by Sentinel 5P as a longitude-latitude coordinate pair) where XCH$_4$ has been retrieved. The retrievals are written as $\{Z(\s_i):i=1,...,n\}$. Individual measurement-error quantifications associated with the XCH$_4$ retrievals are provided at every location by the TROPOMI instrument onboard the Sentinel 5P satellite \citep{Hu2016}. Hence write the retrieval-specific measurement-error standard deviations as $\{\sigma_\trop(\s_i): i =1,...,n\}$ (in ppb). The subscript `$\trop$' is intended to recall `TROPOMI'. The notation here is deliberately different from the `$\sigma_\TO$' used to denote the (BAU-level) measurement-error standard-deviation parameter in various data models described elsewhere in the main article and in this Supplement. The reasons for this are made clear in Section \ref{sec:appendix_averaging_data}, where we describe the process of aggregating the point-level XCH$_4$ retrievals to the BAU-level.  

Fig. \ref{fig:tropomi_measurement_error} plots $\{\sigma_\trop(\s_i):i=1,...,n\}$ against $\{Z(\s_i):i=1,...,n\}$. It reveals that the behavior of the measurement errors and their relationship to the magnitude of the underlying measurements vary by day. There are some days where there appears to be no relationship. On other days, the measurement error increases or decreases with the magnitude of the retrieval. 

\begin{figure}[H]
\centering
\includegraphics[width=\textwidth]{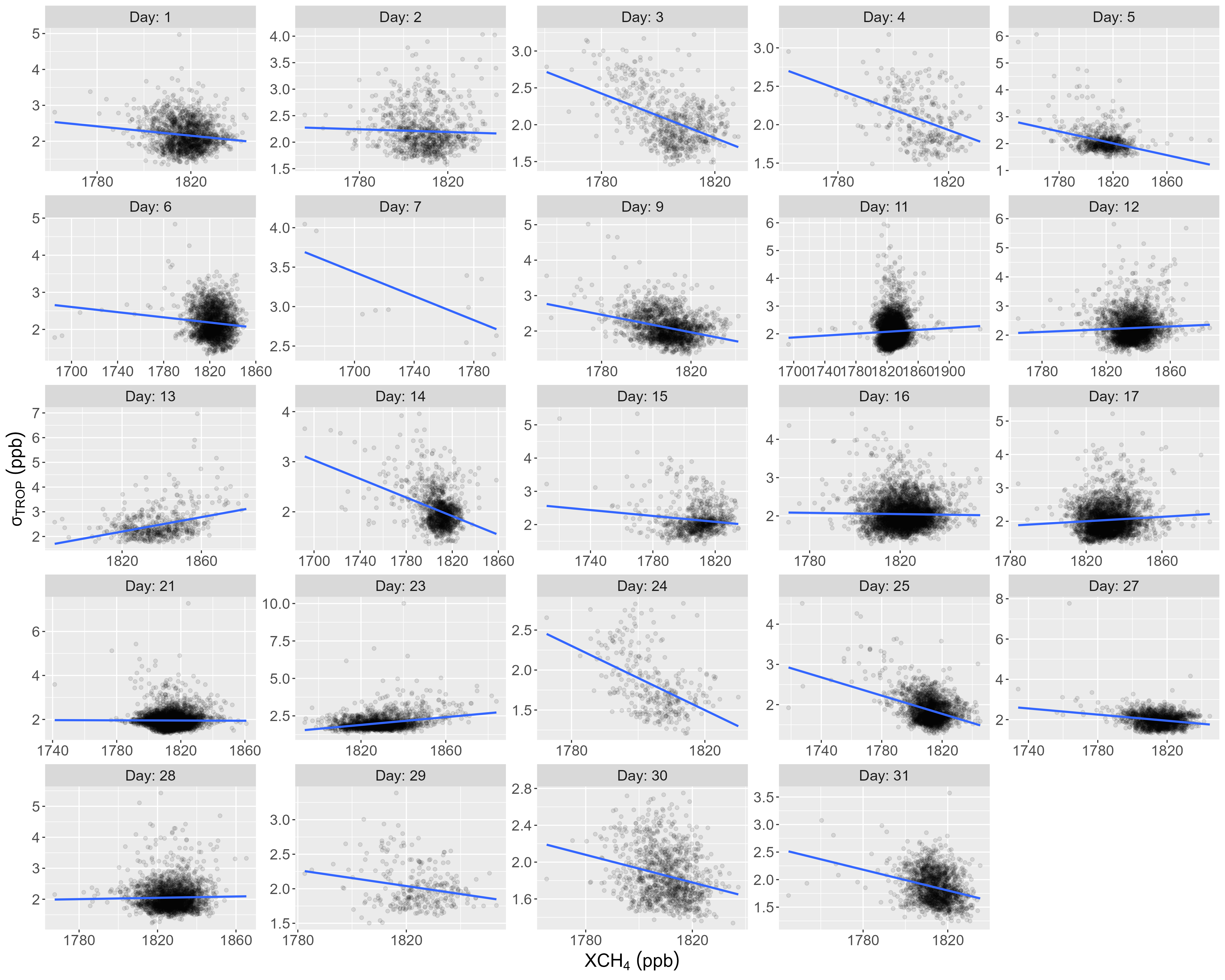}
\caption{Relationships between the retrieval-specific measurement-error standard deviation $\sigma_\trop(\cdot)$ and magnitude of the XCH$_4$ retrieval, both in parts per billion (ppb), by day in August 2020, in and around the Bowen Basin in the state of Queensland, Australia.}
\label{fig:tropomi_measurement_error}
\end{figure}

\subsection{Bisquare spatial basis functions over the Bowen Basin, Queensland, Australia}\label{sec:appendix_methane_basis}

Fig. \ref{fig:appendix_methane_basis_funcs} shows the 101 bisquare basis functions used to map XCH$_4$ in and around the Bowen Basin, Queensland, Australia. These basis functions were selected and generated using the FRK package \citep{ZM2021, SainsburyDale2024} in R. 

\begin{figure}[H]
\centering
\includegraphics[width=0.65\textwidth]{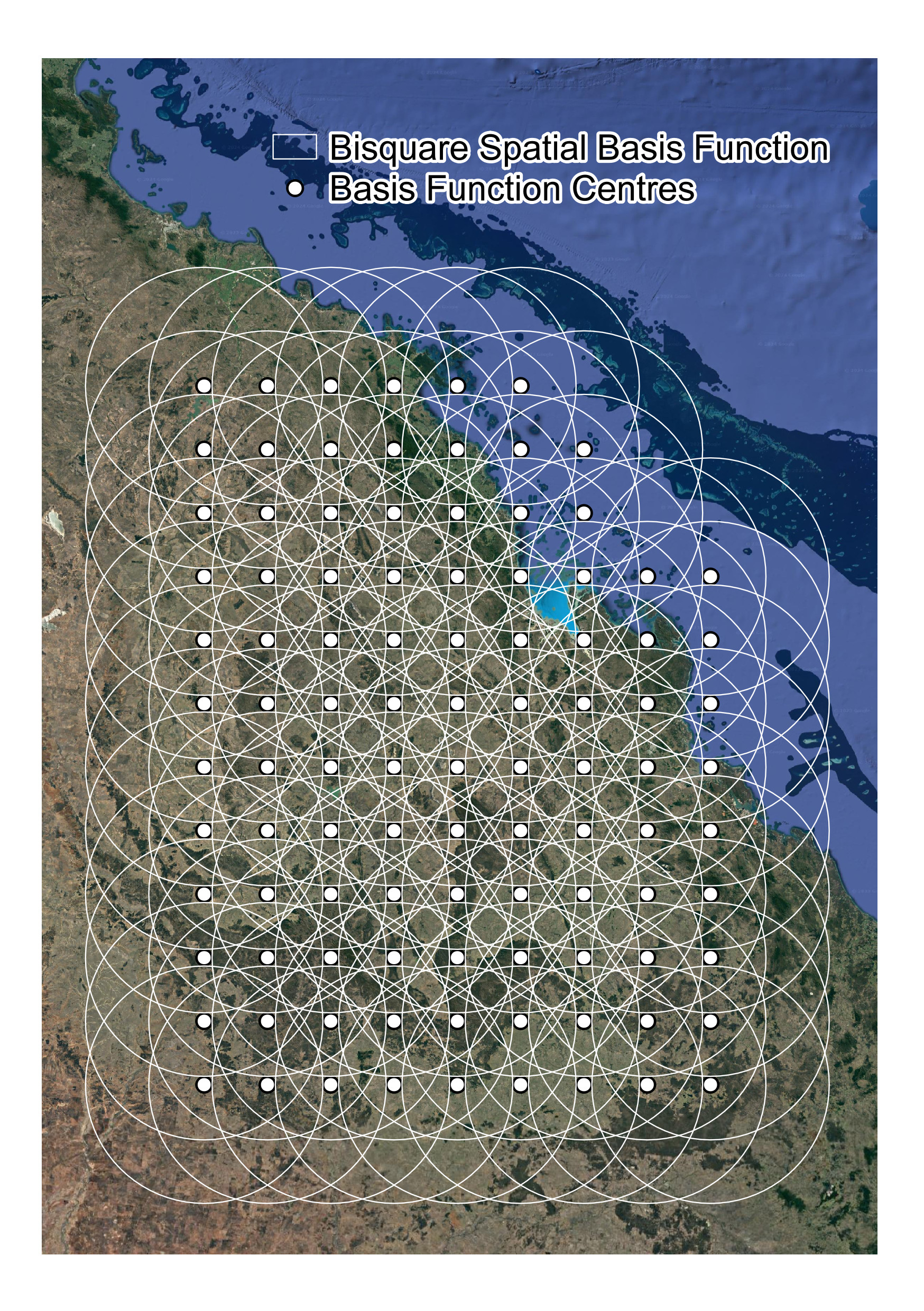}
\caption{A set of 101 bisquare spatial basis functions with apertures of 0.9258 decimal degrees (dd) covering the Bowen Basin, Queensland, Australia. The dots represent the basis-function centers, and the circles have a radius of 0.9258 dd. The map is plotted in the EPSG:4326 coordinate system.}
\label{fig:appendix_methane_basis_funcs}
\end{figure}

\subsection{Defining BAU-level data and averaging retrievals within the BAUs}\label{sec:appendix_averaging_data}

On any given day, the retrievals $\{Z(\s_i): i =1, ..., n\}$ may occur anywhere in the domain $D$ that falls along the track of the Sentinel 5P satellite. Hence the retrieval locations $\{\s_i:i=1,...,n\}$ do not form a regular grid, and they are not necessarily aligned with the BAUs shown in Fig. \ref{fig:appendix_methane_baus}. Since our copula-based models (Sections \ref{sec:hierarchical_models}, \ref{sec:appendix_LGGau}--\ref{sec:appendix_SGt}) are defined for BAU-level spatial data, we need to aggregate here the point-level retrievals  $\{Z(\s_i): i =1, ..., n\}$ to produce a dataset of BAU-level observations, $\{Z(A_{\TO k}): k = 1, ..., K\}$, where $K \leq n$ and, obviously, each location $\{\s_i:i=1,...,n\}$ falls into one and only one BAU.

For a given BAU, $A_{\TO k}$ ($k = 1, ..., K$), suppose that there are retrieval locations $\{\s'_1, ..., s'_{n_k}\}$ that fall into $A_{\TO k}$. Then $A_{\TO k}$ contains the retrievals $\{Z(\s'_1), ..., Z(\s'_{n_k})\}$ with measurement-error standard deviations $\{\sigma_\trop(\s'_1), ..., \sigma_\trop(\s'_{n_k})\}$, provided by the Sentinel 5P database. Assume all point-level retrievals are conditionally independent given $Y(A_{\TO k})$, and let $\{[Z(\s_{h}') \mid Y(A_{\TO k})]: h = 1, ..., n_k]\}$ denote the conditional distributions of the point-level retrievals given $Y(A_{\TO k})$, with conditional mean $E(Z(s_{h}') \mid Y(A_{\TO k})) = Y(A_{\TO k})$ and conditional variance $\mathrm{var}(Z(s_{h}') \mid Y(A_{\TO k})) = \sigma_\trop(\s'_{h})^2$, for $h = 1, ..., n_k$. Note that, unconditionally, the retrievals are spatially dependent.

For BAUs that contain at least one observation, we define the BAU-level observation as the average of the point-level retrievals: That is,
$$
Z(A_{\TO k}) = \frac{1}{n_k}\left(Z(\s'_1) + \cdots + Z(\s'_{n_k})\right).
$$
The consequences of this averaging for the measurement-error standard deviations are as follows: Due to conditional independence, we have
\begin{align*}
\mathrm{var}(Z(A_{\TO k})\mid Y(A_{\TO k})) &\equiv \mathrm{var}\!\left(\frac{1}{n_k}(Z(\s'_1) + \cdots + Z(\s'_{n_k}))\middle|~ Y(A_{\TO k})\right)\\
&= \frac{1}{n_k^2}\!\Big(\mathrm{var}\!\left(Z(\s_1') \mid Y(A_{\TO k})\right) + \cdots + \mathrm{var}\!\left(Z(\s_{n_k}') \mid Y(A_{\TO k})\right)\Big)\\
&= \frac{1}{n_k^2}\Big(\sigma_\trop(\s'_1)^2 + \cdots +\sigma_\trop(\s'_{n_k})^2\Big)\\
&\equiv \sigma_\TO(A_{\TO k})^2.
\end{align*}
It follows that the BAU-level measurement-error standard deviations are given by,
\begin{equation}
\sigma_\TO(A_{\TO k}) \equiv \frac{1}{n_k}\sqrt{\sigma_\trop(\s'_1)^2 + \cdots + \sigma_\trop(\s'_{n_k})^2}.\label{eqn:appendix_avg_ME_sd}
\end{equation}
Hence, we use $\sigma_\TO(A_{\TO k})$ as the measurement-error standard deviation in the data model of the copula-based hierarchical spatial statistical model for XCH$_4$ (see next section). BAUs that contain no data define which part of the latent vector $\Y$ are written as $\Y_\TM$.

For the XCH$_4$ dataset analyzed in Section \ref{sec:methane}, there are a small number of BAUs containing two retrievals, and no BAUs contain more than two retrievals. On 23 August (Day 23) and 28 August (Day 28), 2020, there are 33 and 31 BAUs containing two retrievals, respectively. The pairs of retrievals were averaged to define the BAU-level XCH$_4$ value, and the measurement-error standard deviations of the BAU-level XCH$_4$ values were calculated using \eqref{eqn:appendix_avg_ME_sd}. 

\subsection{The SG-Gau-SRE model for XCH$_4$}\label{sec:appendix_skew_Gaussian}

The SG-Gau-SRE model described in Section \ref{sec:appendix_SGGau} was used to analyze the XCH$_4$ data, with the following modifications. For the data model, we modelled $Z(A_{\TO k}) \mid Y(A_{\TO k}), \sigma_\TO (A_{\TO k})^2 \sim \Gau(Y(A_{\TO k}), \sigma_\TO(A_{\TO k})^2)$, and $[\Z_\TO \mid \Y_\TO, \{\sigma_\TO(A_{\TO k})^2: k = 1, ..., K\}] = \prod_{k=1}^K [Z(A_{\TO k}) \mid Y(A_{\TO k}), \sigma_\TO(A_{\TO k})^2]$, where $\sigma_\TO(A_{\TO k})$ is given by \eqref{eqn:appendix_avg_ME_sd} for all $k=1,...,K$. This differs slightly from the model described in Section \ref{sec:appendix_SGGau} in that the measurement-error variances are now allowed to differ for every observation. In the process model, the matrix of spatial basis functions $\S$ is $7,370\times 101$, where there are $7,370$ BAUs and $101$ bisquare basis functions (Sections \ref{sec:appendix_methane_baus} and \ref{sec:appendix_methane_basis}). The values of the basis functions for the $j$-th BAU, $j = 1, ..., 7,370$, were calculated by evaluating each spatial basis function at the centroid of the $j$-th BAU. The number of rows in $\S_\TO$ and $\S_\TM$ differed by day. Basis functions were used to capture the principal spatial variability, and no covariates were used.

\subsection{FRK for XCH$_4$}\label{sec:appendix_FRK_methane}

For fixed-rank kriging (FRK) of the XCH$_4$ dataset, the following SRE model was used. A Gaussian data model with observation-specific measurement-error variances described in Sections \ref{sec:appendix_averaging_data} and \ref{sec:appendix_skew_Gaussian} was assumed, and a Gaussian process model with no covariates and a $7,370\times101$ matrix of spatial basis functions was used (Section \ref{sec:appendix_skew_Gaussian}). Note that our approach to averaging multiple observations within a BAU, as described in Section \ref{sec:appendix_averaging_data}, differs from the default behavior of the R package FRK \citep{ZM2021}. Therefore, the averaging of pairs of retrievals within some BAUs was handled outside of the R package, but the model was otherwise fitted using the default arguments in the R package FRK \citep{SainsburyDale2024}.

\subsection{Prior distribution of the parameters for modeling the XCH$_4$ dataset}\label{sec:appendix_methane_priors}

Table \ref{tab:appendix_methane_priors} shows the priors chosen for the parameters of the hierarchical spatial copula model. The same priors were used on the two days, 23 August (`Day 23') and 28 August (`Day 28'), 2020, featured in our analysis.

\begin{table}[H]
    \centering
    \caption{Prior distributions and their 95\% probability intervals for the parameters in the model for the XCH$_4$ dataset. The Gamma (Gam) and truncated Gamma priors are parameterized according to the shape and scale; the Gaussian (Gau) prior is parameterized according to the mean and variance; the Half Cauchy (HC) prior is parameterized by a scale parameter.}
    \begin{tabular}{|c|c|c|c|}
        \hline
         Parameter & Prior & 95\% probability interval & Description\\
         \hline
         $\sigma_{\TP}$ & $\mathrm{HC}(1000)$ & (39.29, 25451.70) & Half Cauchy\\
         $\lambda$ & $\Gau(0, 2^2)$ & (-3.92, 3.92) & Gaussian\\
         $\beta_0$ & $\Gau(0, 100^2)$ & (-196, 196) & Gaussian\\
         $\theta_s$ & $\Gam(4, 0.5)$ & (0.5449, 4.3836) & Gamma\\
         $\theta_r$ & $\Gam(5, 0.1)$ & (0.1623, 1.0242) & Gamma\\
         \hline
    \end{tabular}
    \label{tab:appendix_methane_priors}
\end{table}

\subsection{Initial values of the parameters in MCMC}\label{sec:appendix_methane_initial}

Recall that the parameters of the SG-Gau-SRE model are $\sigma_\TP$, $\beta_0$, $\lambda$, $\theta_s$, and $\theta_r$. Let $\sigma_\TP^{(0,t)}$, $\beta_0^{(0,t)}$, $\lambda^{(0,t)}$, $\theta_s^{(0,t)}$, and $\theta_r^{(0, t)}$ denote the initial values of the parameters in the MCMC algorithm for $t \in \{23, 28\}$ (i.e., Days 23 and 28). The initial values were set as follows: The initial value $\sigma_{\TP}^{(0,t)}$ was set to the standard deviation of the respective day's XCH$_4$ values. The mean parameter, $\beta_0^{(0, t)}$, was set to be the mean of the logarithm of the respective day's XCH$_4$ values. The reasons for using the logarithm of the data are the same as those articulated in Section \ref{sec:appendix_sim_initial_values}. The initial value of the skewness parameter, $\lambda^{(0,t)}$, was set to the maximum likelihood estimate of the skewness parameter of a skew-Gaussian distribution fitted directly to the respective day's XCH$_4$ values (assuming independence of the observations). The initial values $\theta_s^{(0,t)}$ and $\theta_r^{(0,t)}$ (corresponding to parameters $\theta_s$ and $\theta_r$ in the spherical covariance function for $\E(\theta_s, \theta_r)$; see Section \ref{sec:daily_maps}) were set equal to 1 and $0.5$, respectively. 

\subsection{Trace plots}\label{sec:appendix_methane_trace_plots}

Fig. \ref{fig:appendix_latent_process_trace}--\ref{fig:appendix_parameters_trace} respectively show the trace plots for five randomly selected latent process values ($\Y$), five randomly selected random effects ($\etab$), and the parameters in the SG-Gau-SRE model for the XCH$_4$ data on Days 23 and 28. The trace plots suggest convergence has been achieved (also see next section), and the Markov chains appear to be adequately mixed.

\begin{figure}[!ht]
    \centering
    \includegraphics[width=0.8\textwidth]{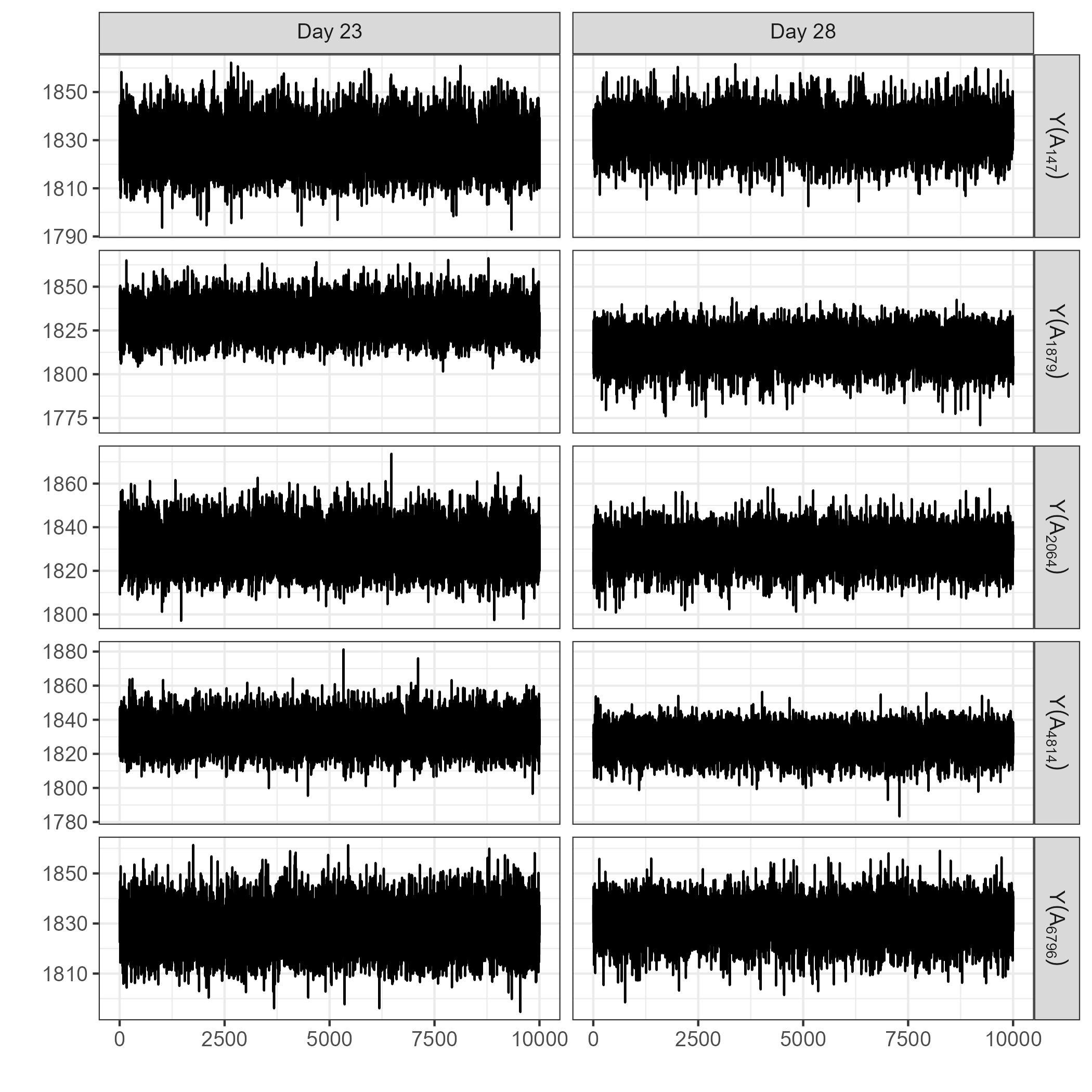}
    \caption{Trace plots for some elements of the latent-process values $\Y$ (both missing and observed BAUs) on Days 23 and 28. The trace plots are shown for five randomly selected BAUs, and the same BAUs were chosen for each day.}
    \label{fig:appendix_latent_process_trace}
\end{figure}

\clearpage

\begin{figure}[!ht]
    \centering
    \includegraphics[width=0.8\textwidth]{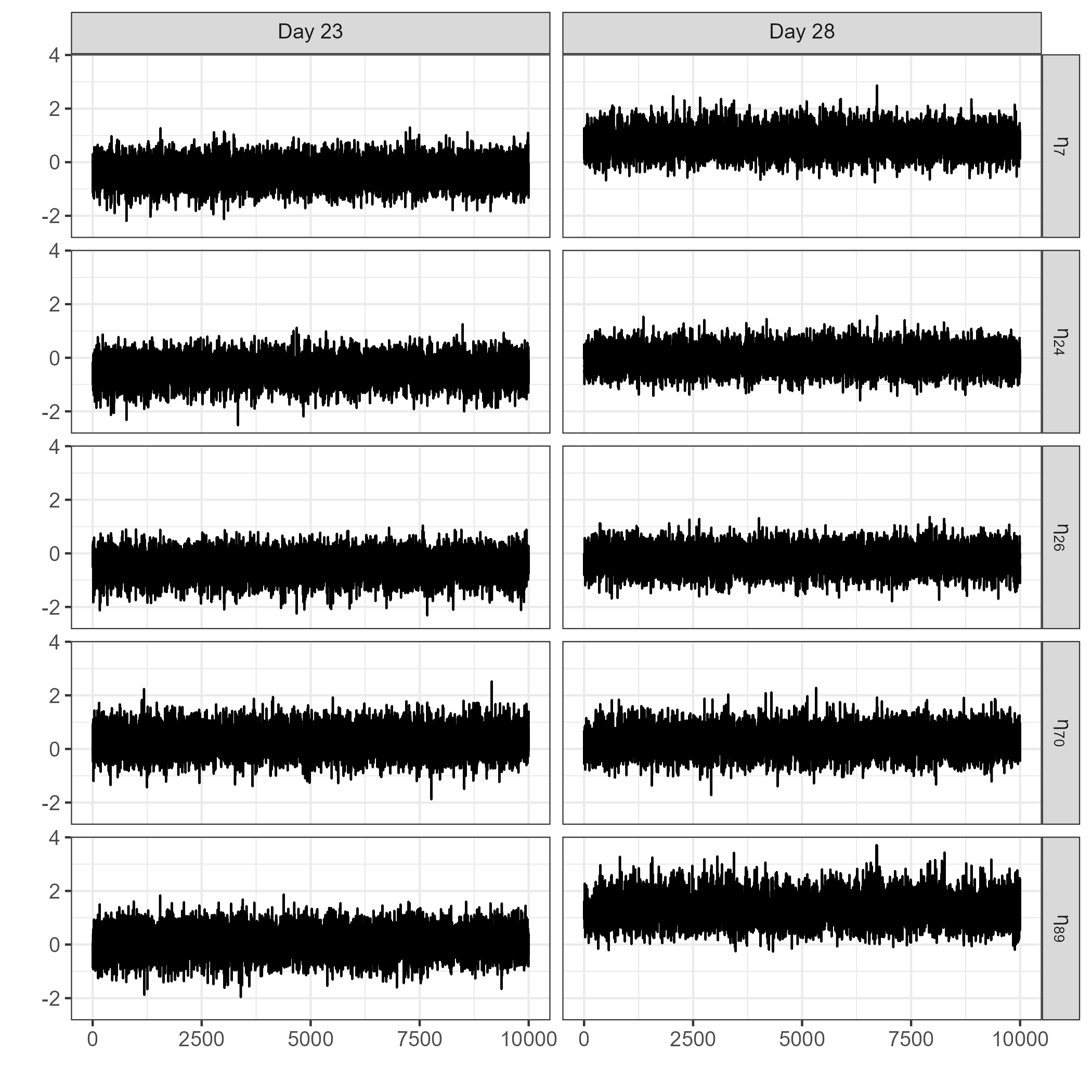}
    \caption{Trace plots for some elements of the spatial random effects $\etab$ on Days 23 and 28. The same random effects were chosen for display on each day.}
    \label{fig:appendix_random_effect_trace}
\end{figure}

\clearpage

\begin{figure}[!ht]
    \centering
    \includegraphics[width=0.8\textwidth]{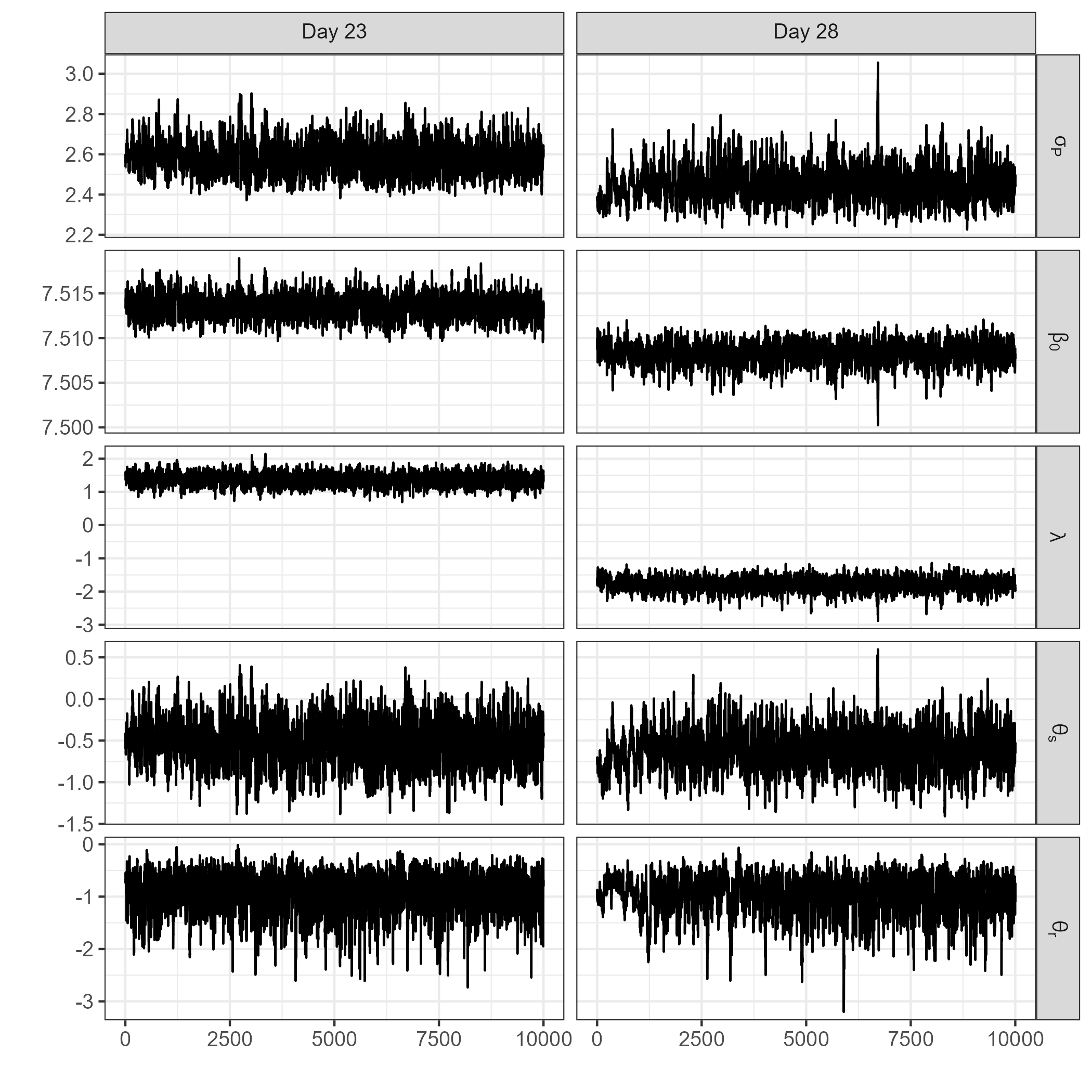}
    \caption{Trace plots for the parameters of the SG-Gau-SRE model for the XCH$_4$ data on Days 23 and 28.}
    \label{fig:appendix_parameters_trace}
\end{figure}

\clearpage

\subsection{Convergence diagnostics}\label{sec:appendix_convergence_diagnostics}

Gelman-Rubin statistics \citep{Gelman1992} were calculated to quantitatively assess convergence of the MCMC sampler. Two additional chains were run from different random seeds and initial starting values for the parameters, random effects, and latent-process values. The Gelman-Rubin statistics for the parameters in the SG-Gau-SRE model are shown in Table \ref{tab:appendix_GR_stats_parms}. All statistics are close to $1$, indicating that the MCMC sampler achieved convergence. Gelman-Rubin statistics were calculated for the latent-process values, $\Y$, and random effects, $\etab$, but none are presented here because, rounded to two decimal places, all were equal to $1$. The Gelman-Rubin statistics presented here were calculated using the R package coda \citep{coda}.

\begin{table}[!ht]
    \centering
        \caption{Gelman-Rubin (GR) statistics for the parameters of the SG-Gau-SRE model used to analyze atmospheric methane concentrations on Days 23 and 28. A value of the GR statistic close to $1$ indicates convergence.}
    \begin{tabular}{|c|c|c|}
    \hline
        Parameter & Day 23 & Day 28\\
         \hline
        $\sigma_\TP$ & 1.00 & 1.01 \\
        $\beta_0$ & 1.00 & 1.00\\
        $\lambda$ & 1.00 & 1.00\\
        $\theta_s$ & 1.00 & 1.01\\
        $\theta_r$ & 1.01 & 1.00\\
        \hline
    \end{tabular}
    \label{tab:appendix_GR_stats_parms}
\end{table}

\end{document}